\documentclass[aps,twocolumn,floatfix,superscriptaddress,longbibliography,normalem]{revtex4-2}

\usepackage{adjustbox}
\usepackage{amssymb,amsmath,amstext,amsthm}
\usepackage{algpseudocode}
\usepackage{algorithm}
\usepackage{graphicx}
\usepackage{epstopdf}
\usepackage{color}
\usepackage{bm}
\usepackage{appendix}
\usepackage[T1]{fontenc}
\usepackage{bbold}
\usepackage{bbm}
\usepackage{latexsym}
\usepackage{physics}
\usepackage{xcolor}
\usepackage[colorlinks=true,citecolor=blue,linkcolor=magenta]{hyperref}
\usepackage{ulem}
\usepackage{comment}
\usepackage{rotating}
\usepackage{overpic}

\definecolor{darkblue}{rgb}{0.0, 0.0, 0.55}
\algnewcommand\algorithmicinput{\textbf{Input:}}
\algnewcommand\INPUT{\item[\algorithmicinput]}

\algnewcommand\algorithmicoutput{\textbf{Output:}}
\algnewcommand\OUTPUT{\item[\algorithmicoutput]}

\algnewcommand\algorithmiccomplexity{\textbf{Asymptotic complexity:}}
\algnewcommand\COMPLEXITY{\item[\algorithmiccomplexity]}

\renewcommand{\vec}[1]{\boldsymbol{#1}}

\def\ad{^{\dagger}}

\newcommand{\dimg}{\operatorname{dim}(\mathfrak{g})}

\newcommand{\irrep}{\LC_\lambda}
\newcommand{\irrepa}{\LC_{\lambda_a}}

\newcommand{\g}{\mathfrak{g}}
\newcommand{\G}{\mathcal{G}}
\def\gsim{\g{\text - }{\rm sim}}
\newcommand{\adrep}[1]{\Phi_{\lambda}^{\operatorname{ad}}(#1)}

\newcommand{\stab}{{\rm N}(\mathcal{L}_{\lambda})}
\newcommand{\staba}{{\rm N}(\mathcal{L}_{\lambda_a})}

\usepackage[makeroom]{cancel}

\newcommand{\dya}[1]{\ket{#1}\!\bra{#1}}

\newcommand{\poly}{\operatorname{poly}}

\newcommand{\GC}{\mathcal{G}}

\newcommand{\OC}{\mathcal{O}}

\newcommand{\SC}{\mathcal{S}}

\renewcommand{\geq}{\geqslant}
\renewcommand{\leq}{\leqslant}

\newcommand{\LC}{\mathcal{L}}

\DeclareMathOperator*{\argmax}{arg\,max}
\DeclareMathOperator*{\argmin}{arg\,min}
\renewcommand{\vec}[1]{\boldsymbol{#1}}  

\newcommand{\bs}{\textsf{BS}}

\newcommand{\thv}{\vec{\theta}}

\def\be{\begin{equation}}
\def\ee{\end{equation}}
\def\bs{\begin{split}}
\def\e{\end{split}}
\def\ba{\begin{eqnarray}}
\def\bea{\begin{eqnarray}}

\def\tea{\end{eqnarray}}
\def\ea{\end{eqnarray}}
\def\eea{\end{eqnarray}}

\usepackage{mathtools}

\newcommand\mf[1]{\mathfrak{#1}}
\newcommand\mbb[1]{\mathbb{#1}}

\newtheorem{theorem}{Theorem}

\newtheorem{lemma}{Lemma}
\newtheorem{corollary}{Corollary}

\newtheorem{definition}{Definition}
\newtheorem{result}{Result}

\usepackage{amssymb}
\usepackage{dsfont}

\def\be{\begin{equation}}
\def\te{\end{equation}}
\def\ee{\end{equation}}
\def\ba{\begin{eqnarray}}
\def\bea{\begin{eqnarray}}

\def\tea{\end{eqnarray}}
\def\ea{\end{eqnarray}}
\def\eea{\end{eqnarray}}

\begin{document}

\title{Lie-algebraic classical simulations for quantum computing}

\author{Matthew L. Goh}
\email{matt.goh@merton.ox.ac.uk}
\affiliation{Theoretical Division, Los Alamos National Laboratory, Los Alamos, New Mexico 87545, USA}
\affiliation{Department of Materials, University of Oxford, Parks Road, Oxford OX1 3PH, United Kingdom}

\author{Martin Larocca}
\affiliation{Theoretical Division, Los Alamos National Laboratory, Los Alamos, New Mexico 87545, USA}
\affiliation{Center for Nonlinear Studies, Los Alamos National Laboratory, Los Alamos, New Mexico 87545, USA}

\author{Lukasz Cincio}
\affiliation{Theoretical Division, Los Alamos National Laboratory, Los Alamos, New Mexico 87545, USA}

\author{M. Cerezo}
\affiliation{Information Sciences, Los Alamos National Laboratory, Los Alamos, New Mexico 87545, USA}
 
\author{Fr\'{e}d\'{e}ric Sauvage}
\affiliation{Theoretical Division, Los Alamos National Laboratory, Los Alamos, New Mexico 87545, USA}
\affiliation{Quantinuum, Partnership House, Carlisle Place, London SW1P 1BX, United Kingdom}

\begin{abstract}
The classical simulation of quantum dynamics plays an important role in our understanding of quantum complexity, and in the development of quantum technologies. Efficient techniques such as those based on the Gottesman-Knill theorem for Clifford circuits, tensor networks for low entanglement-generating circuits, or Wick's theorem for fermionic Gaussian states, have become central tools in quantum computing. In this work, we contribute to this body of knowledge by presenting a framework for classical simulations, dubbed `$\gsim$', which is based on the underlying Lie algebraic structure of the dynamical process.
When the dimension of the algebra  grows at most polynomially in the system size, there exists observables for which the simulation is efficient. Indeed, we show that $\gsim$ enables new regimes for classical simulations, is able to deal with certain forms of noise in the evolution, as well as can be used to tackle several paradigmatic variational and non-variational quantum computing tasks.    
For the former, we perform  Lie-algebraic simulations to train and optimize  parametrized quantum circuits (thus effectively showing that some variational models can be dequantized), design  enhanced parameter initialization strategies, solve tasks of quantum circuit synthesis, and train a quantum-phase classifier. For the latter, we report large-scale noiseless and noisy simulations on benchmark problems.  By comparing the limitations of $\gsim$ and certain Wick's theorem-based simulations, we find that the two methods become inefficient for different types of states or observables, hinting at the existence of distinct, non-equivalent, resources for classical simulation.
\end{abstract}
\maketitle

\section{Introduction}\label{sec:intro}
At a fundamental level, the ability to \emph{classically} simulate quantum dynamics in certain regimes helps us refine our understanding of the true nature of quantum complexity. Different approaches allowing for scalable classical simulations have elucidated distinct aspects of quantumness by restricting the set of operations performed, the initial states evolved, and the observables measured. These include, but are not limited to, the role of the discrete-group structure of Clifford operations in the case of stabilizer simulations~\cite{gottesman1997stabilizer}, the importance of entanglement in the case of tensor networks~\cite{Vidal2003Efficient}, and the existence of transformations to non-interacting particles in the case of free fermions~\cite{wick1950evaluation,valiant2001quantum,terhal2002classical}. For all the these approaches,  the computational resources required for simulations scale only \emph{polynomially} in the system size, as opposed to the \emph{exponential} cost of the generic or unrestricted case. Such simulations are deemed classically efficient or scalable.

At a practical level, the ability to efficiently simulate quantum systems facilitates the development of quantum technology. As the size of experimentally demonstrated quantum systems is starting to exceed what can be simulated classically, it is of great importance to rely on faithful classical computational results~\cite{kim2023evidence}, even if they are restricted to special types of quantum evolutions, input states, and observables. In this context, classical simulations can serve the role of benchmarks for currently developed devices and can be used for the purpose of characterization. Furthermore, classical simulations can find many more applications along the development and study of variational quantum computing. 
These encompass variational quantum algorithms (VQAs)~\cite{cerezo2020variationalreview,bharti2021noisy} and quantum machine learning (QML) methods~\cite{schuld2015introduction,biamonte2017quantum,larocca2022group,cerezo2022challenges}, that most typically rely on the ability to optimize over parameters of quantum circuits. Despite the promises of such algorithms and demonstrations for relatively small system sizes, it remains unclear how viable these are when scaled to larger sizes~\cite{mcclean2018barren,cerezo2020cost,anschuetz2022beyond}. 
As such, scalable quantum simulators can serve to study the true potential
and limitations of VQAs and QML models, and to extend their capabilities. For instance, classical simulations can be used to provide exact inputs for learning-based error mitigation strategies~\cite{czarnik2020error,strikis2020learning,google2020observation,montanaro2021error}, to probe the trainability of quantum circuits at large scale~\cite{matos2022characterization},  for the purpose of pre-training the models~\cite{grant2019initialization,verdon2019learning,sauvage2021flip,rad2022surviving,liu2022mitigating,mitarai2022quadratic,ravi2022cafqa,cheng2022clifford,dborin2022matrix,mele2022avoiding,rudolph2022synergy}, but also to dequantize certain architectures~\cite{cerezo2023does,angrisani2023learning,bermejo2024quantum,lerch2024efficient,schreiber2023classical,jerbi2023shadows,shao2023simulating,basheer2023alternating,shaffer2023surrogate,anschuetz2022efficient,mele2024efficient}.

\begin{figure*}
    \centering
    \includegraphics[width=1\textwidth]{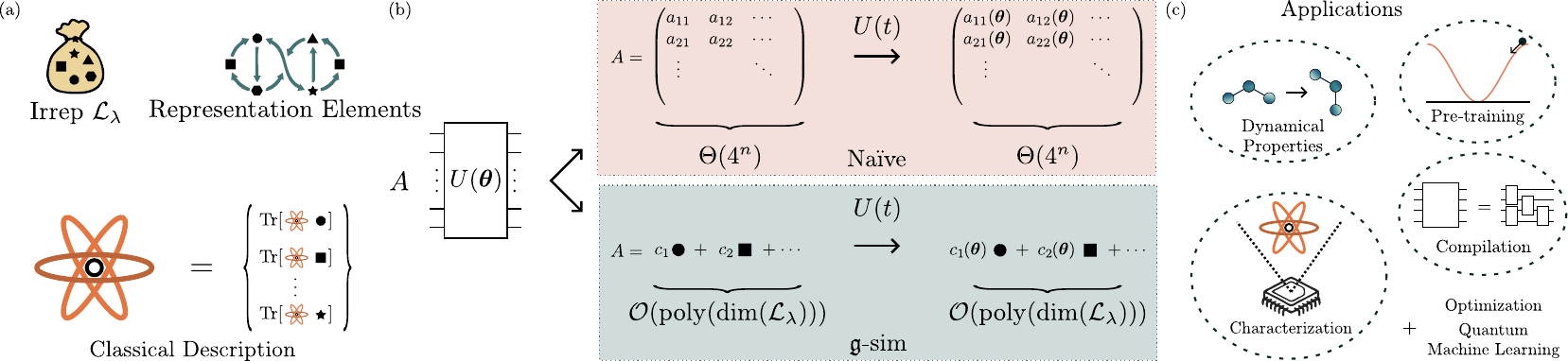}
    \caption{\textbf{Theoretical framework and applications of $\g$-sim.} \textbf{(a)} Essential components of the Lie-algebraic representation of quantum dynamics. First, we need a description of the Lie algebra $\g$ associated with a unitary evolution of interest~\eqref{eq:lie_alg}, as well as an irreducible representation (irrep) $\irrep$: a subspace of the operator space that is invariant under the action of $\mathcal{G}=e^{\g}$ (Definition.~\ref{def:invariant}).
    Here we depict the elements of $\irrep$ as a collection of black geometric shapes. Next, we need to know how these elements inter-connect to each other via the dynamics. This is determined entirely by the so-called representation elements~\eqref{eqn:structure_factors}. The final ingredient of $\gsim$ is a classical description of the input state, i.e., the expectation values with respect to the elements of $\irrep$~\eqref{eq:descr_inout}. \textbf{(b)} Comparison of na\"ive and Lie-algebraic approaches to quantum dynamics in the Heisenberg picture. In the na\"ive approach, an observable $A$ is viewed as a linear operator acting on an $n$-qubit Hilbert space $\mathcal{H}$, with $\operatorname{dim}(\mathcal{H})=2^n$. The observable is represented as a $2^n\times 2^n$ matrix whose coefficients evolve according to the von Neumann equation of motion, and thus the approach \emph{always} scales as $\Theta(4^n)$ and is not classically scalable. In the Lie-algebraic approach, the observable is instead viewed as a linear combination of distinct basis terms in $\g$, whose coefficients couple to each other over time via the representation elements.
    Whenever $\dimg\in\OC(\operatorname{poly}(n))$, $\gsim$ yields a scalable classical simulation framework for irreps with $\dim(\irrep)\in\OC(\operatorname{poly}(n))$. \textbf{(c)} Summary of the potential applications of  $\gsim$. These efficient classical simulations have utility in simulation and optimization of quantum systems, including pre-training of VQA or QML problems, compiling unitary processes to compact quantum circuits, and characterizing QML optimization landscapes. }
    \label{fig:schematic}
\end{figure*}

Among the scalable simulation methods that have been proposed, we focus on those based on the Lie-algebraic structure of the underlying dynamics~\cite{somma2005quantum,somma2006efficient,anschuetz2022efficient}.
Here, simulation complexity scales with the dimension (denoted $\dimg$) of the Lie-algebra (denoted $\mathfrak{g}$) induced by 
 the operators generating 
the dynamics of the system - these notions are defined more precisely soon. 
Crucially, in certain cases 
such dimension can be substantially smaller that the dimension of the 
Hilbert space of the system of interest, allowing for more efficient simulations that would be entailed generically.
In particular, such simulations become scalable whenever $\dimg$ grows polynomially with the system size $n$ (i.e., $\dimg\in\OC(\operatorname{poly}(n))$). For example, this is known to occur in systems with with permutation symmetries \cite{schatzki2022theoretical,anschuetz2022efficient}, underlying free-fermionic algebras \cite{bonet2020nearly,kokcu2022fixed} , continuous-variable or free-bosonic systems~\cite{qvarfort2022solving, barthe2024gate, somma2024shadow} and others~\cite{west2023provably,kazi2024analyzing}. The study and cataloging of the many Lie algebras occurring in quantum systems is an active field of research~\cite{oszmaniec2017universal,pozzoli2022lie,wiersema2024classification,kokcu2024classification,kazi2024analyzing, aguilar2024full}.

In this work, we present a framework for the efficient implementation of Lie-algebraic simulations (referred to as $\gsim$) and their utilization.  
 As a first contribution, we extend techniques from Refs.~\cite{somma2005quantum,somma2006efficient}, significantly widening their scope. We formulate these techniques in a modern language oriented to the quantum computing community, offer new optimizations that improve their efficiency and extend their utility to optimizing circuits (as opposed to merely simulating them). Comparing $\gsim$ against Wick-based simulation techniques we show that these can be efficient for different types of states and observables, hinting at some deep connections with resource theory~\cite{chitambar2019quantum}, and demonstrating that our proposed method enables new classical simulation scenarios beyond those reachable via Wick's theorem.
Our second contribution is to showcase the utility of $\gsim$ in four different tasks pertaining to the development of quantum technologies. These include the study of the optimization landscapes of parametrized quantum circuits, improving on the initializations of such circuits, problems of quantum circuit synthesis, and the training of a binary quantum-phase classifier. A high-level overview of our theoretical framework and applications for this work is provided in Fig.~\ref{fig:schematic}. 

In the first part of this work (Sec.~\ref{sec:framework}) we review the theory grounding Lie-algebraic simulations and provide guidelines for their efficient implementations.
The exposition here is aimed to be pedagogical, with an eye towards generic and efficient algorithms. 
It follows from Refs.~\cite{somma2005quantum,somma2006efficient}, and extends them in terms of scope (increased set of initial states and of observables supported together with the ability to deal with certain forms of noise in the circuits) and implementation (including both the ability to simulate and differentiate evolution). 

In the second part (Sec.~\ref{sec:demo}), we summarize necessary conditions for the scalability of $\gsim$ and demonstrate its capabilities in non-variational benchmark problems. These include both noiseless and noisy simulations of $n=200$ qubits systems and will serve us to highlight fundamental differences with established classical simulations techniques.

In the last part (Sec.~\ref{sec:applications}), $\gsim$ is put into practice in four distinct scenarios. 
As a warm-up example, we provide a large scale study of the overparametrization phenomenon (Sec.~\ref{sec:overparametrization}).
Then, $\gsim$ is applied for the purpose of the initialization of (non-classically-simulable) circuits in two different setups (Sec.~\ref{sec:pretraining}).
Next, tasks of circuit synthesis and dynamical evolution are addressed (Sec.~\ref{sec:compilation}). 
Finally, as an application in the context of QML, we resort to $\gsim$ to train a quantum-phase classifier  (Sec.~\ref{sec:supervised_QML}).
Overall, these examples aim at illustrating the diversity of tasks that can be accomplished via $\gsim$, and also showcase some avoidable pitfalls.

\section{Framework}
\label{sec:framework}

\subsection{Set-up}\label{sec:setup}
Through this article we focus on quantum dynamics realized by means of  quantum circuits (i.e., digital quantum computations), but stress that the same principles can equally be applied to continuously driven systems (i.e., analog quantum computations).
Also, we note that while exposed in the context of the qubit model, the same principles are readily ported to other models of quantum physics.

With this scope in mind, we will henceforth study a system of $n$ qubits with an associated Hilbert space $\mathcal{H}=(\mbb{C}^2)^{\otimes n}$ of dimension $d=2^n$, and we denote as $\mathcal{L}$ the space of operators acting on $\mathcal{H}$.
We further consider the case when the dynamics of the quantum system are determined by a unitary operator
\begin{equation}
    U(\bm{\theta})=\prod_{l=1}^L e^{-i\theta_{l}H_l} \,, \label{eqn:PeriodicStructureAnsatz}
\end{equation}
specified in terms of  angles (or parameters) $\bm{\theta}=[\theta_{1},\hdots, \theta_{L}]$ and Hermitian operators $H_l$ taken from some set of  \emph{gate generators} $\SC$ of the circuit.

In the following we will be concerned with tasks of simulation and optimization of quantum circuits. 
For an observable $O$, an initial state $\rho^{(\operatorname{in})}$, and a circuit $U(\bm{\theta})$, \emph{simulation} consists of evaluating the expectation value of $O$ for the evolved state 
$\rho^{(\operatorname{out})}(\bm{\theta})\equiv U(\bm{\theta}) \rho^{(\operatorname{in})} U^{\dagger}(\bm{\theta})$. That is, by simulation we mean the evaluation of a quantity of the form
\begin{equation}\label{eq:expectation}
    \langle O(\bm{\theta}) \rangle \equiv\Tr[O U(\bm{\theta}) \rho^{(\operatorname{in})} U^{\dagger}(\bm{\theta})]= \Tr[O \rho^{(\operatorname{out})}(\bm{\theta})]\,.
\end{equation}

In addition to simulation, one is often interested in actively \emph{optimizing} the system's dynamics to achieve a certain objective. 
For instance, this is at the core of most problems found in the field of VQAs and QML. These rely on optimizing circuit parameters to decrease the value of a loss function that can be evaluated through expectation values. 
As an example, in a variational quantum eigensolver (VQE)~\cite{peruzzo2014variational} aiming at preparing the ground state of an Hamiltonian $H$, this loss would be the expectation value $\langle H(\bm{\theta}) \rangle$. Similarly, 
in the context of QML, the same circuit unitary may be applied to a dataset of quantum states, and the loss to be minimized could be a function of  expectation values of a set of observables $\{O_i\}$ that are obtained for each state. 
In both cases, optimization of the circuit parameters can be enhanced by  the ability to compute gradients of the form $\partial_{\bm \theta} \langle O(\bm{\theta}) \rangle$. 

\subsection{Lie theory}
\label{sec:lie}
Here we recall elementary results of Lie theory that are relevant to this work. For a general treatment of Lie groups and Lie algebras, we recommend standard textbooks ~\cite{hall2013lie,kirillov2008introduction}. A more specific presentation of Lie theory in the context of quantum control can be found in Refs.~\cite{dalessandro2010introduction,zeier2011symmetry}, or in the context of quantum circuits in Refs.~\cite{larocca2021diagnosing,larocca2021theory}. 
We start by reviewing the concepts of dynamical Lie algebra (Definition~\ref{def:dynamical_lie_algebra}) and dynamical Lie group (Definition~\ref{def:dynamical_lie_group}) which characterize the expressiveness of quantum circuits of the form in Eq.~\eqref{eqn:PeriodicStructureAnsatz}.

\begin{definition}[Dynamical Lie algebra]\label{def:dynamical_lie_algebra} Consider an ansatz of the form in Eq.~\eqref{eqn:PeriodicStructureAnsatz}. 
Its dynamical Lie algebra $\mathfrak{g}$ is defined as the vector space spanned by all the possible nested commutators of $i\SC=\{ iH_1,\dots,iH_L\}$. That is, 
\begin{equation}\label{eq:lie_alg}
\mathfrak{g}={\rm span}_{\mathbb{R}}\left\langle \{iH_1, \ldots,iH_L \}\right\rangle_{\operatorname{Lie}} \subseteq \mathfrak{su}(2^n)\,,
\end{equation}
where $\left\langle \cdot\right\rangle_{\operatorname{Lie}}$ denotes the Lie closure, i.e., the set of elements obtained by repeatedly taking nested commutators until no new linearly independent elements are obtained.  We  denote as $\{iG_\gamma\}_{\gamma=1}^{\dimg}$ a set of anti-Hermitian operators which form a Schmidt-orthonormal basis of $\g$.
\end{definition}
The dynamical Lie algebra $\g$ constitutes a  Lie subalgebra of the special unitary Lie algebra $\mathfrak{su}(2^n)$, the space of linear anti-Hermitian operators acting on the $n$-qubit system. For details on numerical algorithms to compute the Lie closure we refer to Refs.~\cite{zeier2011symmetry,larocca2021diagnosing,wiersema2020exploring}.

In correspondence to the dynamical Lie algebra $\mathfrak{g}$, we can also define the dynamical Lie group $\mathcal{G}$ as follows.

\begin{definition}[Dynamical Lie group]\label{def:dynamical_lie_group} The dynamical Lie group $\mathcal{G}$ of a circuit of the form Eq.~\eqref{eqn:PeriodicStructureAnsatz} is determined by its dynamical Lie algebra $\mathfrak{g}$ and defined as
\begin{equation}
\mathcal{G}=e^{\mathfrak{g}}\equiv\{e^{iA},\quad iA\in\mathfrak{g}\} \,.
    \label{eqn:DLAfromDLG}
\end{equation}
\end{definition}
Notably, the dynamical Lie group  $\mathcal{G}$ corresponds to all possible unitaries that can be implemented by circuits of the form in Eq.~\eqref{eqn:PeriodicStructureAnsatz}. 
That is, for every $ V \in \mathcal{G}$ there exists (at least) one choice of parameter values $\bm{\theta}$ for a sufficiently large, but finite, number of layers $L$ such that $U(\bm{\theta})=V$~\cite{dalessandro2010introduction}. Overall, $\g$ determines the group $\G$ of unitaries that could be realized, and we can take their dimensions in correspondence: ${\rm dim}(\G) \text{ (as a smooth differentiable manifold)}  = \dimg$.

Having defined the Lie algebra and group associated with the dynamical process in Eq.~\eqref{eqn:PeriodicStructureAnsatz}, we can now study how they act on the operator space $\mathcal{L}$.
Our strategy to evaluate Eq.~\eqref{eq:expectation} will consist in evolving operators -- either the state $\rho^{(\operatorname{in})}$ or the observable $O$ -- through conjugation by $U(\bm{\theta})$. Hence we are interested in the action $A\mapsto VAV^{\dagger}$ of the group unitaries $V \in \mathcal{G}$ onto operators  $A \in \mathcal{L}$.
For arbitrary $V$, evaluating such action would scale polynomially with $\dim(\mathcal{L})=d^2$.
However, we can exploit the structure of $\mathcal{G}$ to decompose $\mathcal{L}$ into smaller subspaces in order to simplify the problem. These are the invariant subspaces that are now defined.

\begin{definition}[Invariant subspace]\label{def:invariant}  We say that an operator subspace $\irrep \subset \mathcal{L}$ is invariant under the action of $\mathcal{G}$ whenever 
\begin{equation}
    VAV^\dag \in \irrep, \forall A\in \irrep\text{ and } \forall V \in \mathcal{G}.
\end{equation}
Furthermore we say that $\irrep$ is irreducible (and call it an irrep) when it is invariant and when there exists no non-trivial $\mathcal{L}_{\mu}\subset\irrep$ that is also invariant. 
\end{definition}

For compact groups acting on finite dimensional spaces, as considered throughout, one can always decompose $\mathcal{L}$ in terms of invariant subspaces $\mathcal{L}_{\lambda}$ as
\begin{equation}\label{eq;irreps-operator-space}
 \LC\cong\bigoplus_\lambda \irrep.
\end{equation}
The maximal of such decomposition, whereby each $\irrep$ does not further contain invariant subspaces, is called the irreducible (irrep) decomposition. 
The benefits of such decomposition is that the complexity of evaluating Eq.~\eqref{eq:expectation} becomes related to $\dim(\irrep)$ that can be substantially smaller than $ \dim(\mathcal{L})$.

In what follows we will generically refer to the $\mathcal{L}_{\lambda}$ as \emph{the irreps}, but we note that all results hold true provided that $\irrep$ are invariant subspaces, not necessarily irreducible ones.
For instance, in the demonstration and application sections of Sec.~\ref{sec:demo} and~\ref{sec:applications}, we will focus on the case where $\irrep = i \g$. Although not irreducible in general, $i\g$ is an invariant subspace, as shown in Appendix~\ref{appendix:unit_evolve}.
Furthermore, we will denote as $B_\alpha^{(\lambda)}$ (with $\alpha=1, \hdots,\dim(\irrep)$) the Hermitian operators such that $\{B_\alpha^{(\lambda)} \}_\alpha$ forms a   Schmidt-orthonormal basis of $\irrep$. We say that an observable $A$ is \emph{supported} by the irrep $\irrep$ whenever $A \in \irrep$, such that $A$ can be entirely decomposed in the \emph{basis of observables} $\{B_\alpha^{(\lambda)} \}_\alpha$. 

The fact that operator space decomposes into irreps allows us to precisely characterize how a given operator transform under the action of a unitary in the dynamical Lie group.  For this purpose, we begin by recalling that given an operator $iG_\gamma\in\mathfrak{g}$, its adjoint action on an irrep $\irrep$ is  fully characterized by taking a Hermitian basis of $\irrep$  and computing the  \textit{representation elements} $f^{(\lambda)\gamma}_{\alpha\beta}= \Tr[iB_\beta^{(\lambda)}\left[iG_\gamma,iB_\alpha^{(\lambda)},\right]] \in \mbb{R}$, through
\begin{equation}
\left[iG_\gamma,iB_\alpha^{(\lambda)}\right]=\sum_{\beta=1}^{\dim(\irrep)}f^{(\lambda)\gamma}_{\alpha\beta} iB_\beta^{(\lambda)}\,.
    \label{eqn:structure_factors}
\end{equation}
By linearity, these constants fully capture the action of any Lie-algebra elements over any operators belonging to $\irrep$, in terms of $\dimg$ matrices of sizes $\dim(\irrep)\times \dim(\irrep)$. We leverage the following definition.

\begin{definition}[adjoint representation
of $\g$ in the $\lambda$-th irrep]
Given an  operator  $iG_\gamma$ in $\g$, its  adjoint representation on the irrep $\irrep$  is obtained via the map $\Phi_{\lambda}^{\operatorname{ad}}:\mathfrak{g}\mapsto \mathbb{R}^{\dim(\irrep)\times \dim(\irrep)}$ and is defined by
\begin{equation}    
    \left(\Phi_{\lambda}^{\operatorname{ad}}(iG_\gamma)\right)_{\alpha \beta}\equiv f^{(\lambda)\gamma}_{\alpha \beta} = \Tr[iB_\beta^{(\lambda)}\left[iG_\gamma,iB_\alpha^{(\lambda)}\right]] \,.
    \label{eqn:adjoint_representation_defn}
\end{equation}
\label{defn:adjoint_representation}
\end{definition}
We note that if the underlying algebra  is compact, and if the basis observables $B_\alpha^{(\lambda)}$ of the irreps are Schmidt-orthonormal, then the adjoint representation is faithful~\cite{somma2005quantum} (i.e., the map $\Phi_{\lambda}^{\operatorname{ad}}$ is injective). It is clear that since the adjoint representation is linear, Definition~\ref{defn:adjoint_representation} implies that knowledge of the adjoint  representations of the $\dimg$ basis elements is sufficient to obtain the representation of any element of $\g$. That is, for any $A =\sum_\gamma (\bm{w})_\gamma G_\gamma$ (with  $\bm{w}\in \mathbb{R}^{\dimg}$) supported by the algebra, one has  $\Phi_{\lambda}^{\operatorname{ad}}(A)=\sum_\gamma (\bm{w})_\gamma \Phi_{\lambda}^{\operatorname{ad}}(G_\gamma)$.

Next, by means of the exponential map, we can  obtain the adjoint representation of the elements of the Lie group $\G$. 
\begin{definition}[Adjoint representation of $\G$ in the $\lambda$-th irrep]
The adjoint representation of $\mathfrak{g}$ in $\irrep$ induces the adjoint representation $\Phi_{\lambda}^{\operatorname{Ad}}$ of $\mathcal{G}$. This representation is a linear map $\Phi_{\lambda}^{\operatorname{Ad}}:\G\mapsto \mbb{GL}(\mbb{R}^{\dim(\irrep)})\subset\mathbb{R}^{\dim(\irrep)\times \dim(\irrep)}$ from the group $\GC$ to the group of invertible linear operators $\mbb{GL}(\mbb{R}^{\dim(\irrep)})$, defined as
\begin{equation}
    \label{eqn:adjoint_rep_exponentiation}
    \Phi_{\lambda}^{\operatorname{Ad}}( U=e^{iA} )=e^{i \bar{\bm{A}}} \,,
\end{equation}
 for all $iA \in \g$ (or any $U \in \G$), and with $\bar{\bm{A}}\equiv \Phi_{\lambda}^{\operatorname{ad}}(A)$.
 \label{defn:adjoint_representation_unit}
\end{definition}

The main appeal of such representation is that  we can evaluate  the action of any $U \in \G$ over any basis element $B_\alpha^{(\lambda)}$ of any $\irrep$ as 
\begin{equation}
    U^\dagger B_\alpha^{(\lambda)} U = \sum_{\beta}\left(\Phi_{\lambda}^{\operatorname{Ad}}(U)\right)_{\alpha \beta}B_\beta^{(\lambda)}\,.
    \label{eqn:adj_heis}
\end{equation}
We refer the reader to Appendix~\ref{appendix:unit_evolve} for more details on Eq.~\eqref{eqn:adj_heis}.  
As further detailed below, these adjoint representations allow us to  perform (Heisenberg) evolution of any operator in an irrep $\irrep$ under the adjoint action of any $U \in \G$ with a time complexity scaling with $\OC(\poly(\dim(\irrep)))$ as opposed to $\Theta(4^n)$.

Before moving further, we highlight that faithfulness of a Lie algebra representation does not imply faithfulness of its corresponding Lie group representation. To understand this, we need to recall what the center $Z(\mathcal{G})$ of $\mathcal{G}$ is.
\begin{definition}[Center of a group]\label{def:center}
    The center     of a group $\mathcal{G}$ is the subset of $\mathcal{G}$ that simultaneously commutes with all elements of $\mathcal{G}$. That is,
    \begin{equation}
        Z(\mathcal{G})\equiv \{W\in \mathcal{G} \mid \forall U\in\mathcal{G},WU=UW \}\,.
    \end{equation}
\end{definition}
Specifically, one can readily verify that for any $U=e^{iA}$ and  $W\in Z(\mathcal{G})$, we have that $\forall \lambda$
\begin{equation}\label{eq:center-adjoint}
    \Phi^{\operatorname{Ad}}_{\lambda}( WU)=\Phi^{\operatorname{Ad}}_{\lambda}( UW)=\Phi^{\operatorname{Ad}}_{\lambda}( U)\,, 
\end{equation}
showing that a non-trivial center $Z(\mathcal{G})$ results in unfaithfulness of the $\Phi^{\operatorname{Ad}}_{\lambda}$. 
    Even though the Lie algebras $\mathfrak{g}$ considered in this work are centerless, the Lie groups $\mathcal{G}=e^{\mathfrak{g}}$ can still have a nontrivial center, leading to unfaithfulness of $\Phi^{\operatorname{Ad}}_{\lambda}$. While not an issue for most situations encountered, we will see later in Sec.~\ref{sec:compilation} that this introduces additional considerations for unitary compilation.

\subsection{$\mathfrak{g}$-sim principles}
\label{sec:gsim_principles}

In this section we present the main results which comprise the foundation of $\gsim$. These include noiseless simulations with observables supported by one or multiple irreps (Sec.~\ref{sec:gsim_simulation}) or products of observables (Sec.~\ref{sec:gsim_correlator}), noisy simulations (Sec.~\ref{sec:noisy_sim}), and optimizations (Sec.~\ref{sec:gsim_optim}). Additional details are provided in App.~\ref{appendix:proofs},~\ref{appendix:efficient_impl} and~\ref{appendix:gradients}.

\subsubsection{Simulation with observables supported by a single irrep}
\label{sec:gsim_simulation}

Let us first consider the case of  simulation problems where the observable of interest is supported by a single irrep. That is, we want to evaluate Eq.~\eqref{eq:expectation} when $O\in \irrep$ for some $\lambda$.

Let us define $\dim(\irrep)$-dimensional vectors of expectation values that captures the description of the input and output states $\rho^{(\operatorname{in/out})}$ over the basis of observables $B_{\alpha}^{(\lambda)}$:
\begin{equation}\label{eq:descr_inout}
\left(\bm{e}^{(\operatorname{in/out})}_\lambda\right)_\alpha\equiv\Tr[B_{\alpha}^{(\lambda)}\rho^{(\operatorname{in/out})}]\,.
\end{equation}
The following result holds.
\begin{result}[Simulation of observables in the $\lambda$-th irrep]
Consider a circuit of the form in Eq.~\eqref{eqn:PeriodicStructureAnsatz}, and let $O$ be an observable with support in $\irrep$ such that $O=\sum_\alpha (\bm{w})_\alpha B_\alpha^{(\lambda)}$ for $\bm{w}\in \mathbb{R}^{\dim(\irrep)}$.  Then, given $\bm{e}^{(\operatorname{in})}_\lambda$, we can compute
\begin{equation}
    \langle O(\bm{\theta}) \rangle=\bm{w}\cdot\bm{e}^{(\operatorname{out})}_\lambda \,,
    \label{eqn:expectations_from_e_vec}
\end{equation}
with the vector of output expectation values obtained as
\begin{equation}
    \bm{e}^{(\operatorname{out})}_\lambda=\left(\prod_{l=1}^Le^{-i\theta_{l} \bar{\bm{H}}_{l}}\right)\cdot\bm{e}^{(\operatorname{in})}_\lambda, 
    \label{eqn:gsim_ansatz_evolution}
\end{equation}
where $\bar{\bm{H}}_{l}\equiv \Phi_{\lambda}^{\operatorname{ad}}(H_{l})$. 
\label{thm:basic_evolution}
\end{result}

Result~\ref{thm:basic_evolution} indicates that, in order to compute $\langle O(\bm{\theta}) \rangle$, it suffices to have a decomposition of $O$ in the Hermitian basis of the irrep, the adjoint representations $\bar{\bm{H}}_{k}$ of the gate generators, and the vector of expectation values of the basis observables $\{B_\alpha^{(\lambda)}\}_{\alpha=1}^{\dim(\irrep)}$ for the input state. We stress that although Eq.~\eqref{eqn:gsim_ansatz_evolution} resembles unitary evolution in the Hilbert space $\mathcal{H}$, it differs subtly. The vectors $\bm{e}^{(\operatorname{in/out})}_\lambda$ are not state vectors in the usual sense, but vectors of expectation values of observables, and consequently the phases (signs) are indeed physical. Since the $\bar{\bm{H}}_{l}$ are purely imaginary Hermitian matrices, the gate representations $e^{-i\theta_l\bar{\bm{H}}_{l}}$ are real-valued, and describe linear coupling between observables induced by unitary evolution.

An immediate consequence of Result~\ref{thm:basic_evolution} is that we can compute expectation values of observables supported by the irrep with a time complexity scaling as $\OC( L\dim(\irrep)^3)$, with the cubic power arising from matrix exponentiation. 
However, as detailed in Appendix \ref{sec:efficient_implementation_gsim}, we provide  a more efficient algorithm based on pre-computation of the eigendecomposition of the $\bar{\bm{H}}_{k}$ matrices. Further improvements for algebras with a Pauli basis are provided in Appendices \ref{appendix:sparsity} and \ref{appendix:eigendecompositions}. This leads us to our first main contribution, which improves on the time complexity of $\gsim$.
\begin{theorem}\label{theo:scaling}
    Computing expectation values of observables supported by a given irrep $\irrep$ using $\gsim$ has a time complexity in $\OC(L\dim(\irrep)^{2})$ for circuits of the form in Eq.~\eqref{eqn:PeriodicStructureAnsatz}.
\end{theorem}

Result~\ref{thm:basic_evolution} and Theorem~\ref{theo:scaling} can be naturally extended to expectation values of operators supported in multiple irreps $\irrep$ due to linearity. In particular, we have that the following result holds. 

\begin{result}[Simulation of operators in multiple irreps.]
Consider a circuit of the form in Eq.~\eqref{eqn:PeriodicStructureAnsatz}, and let $O$ be an observables with support in a set of irreps $\{\LC_{\lambda_i}\}_i$ such that $O=\sum_iO^{(\lambda_i)}$ where  $O^{(\lambda_i)}=\sum_\alpha(\bm{w}^{(\lambda_i)})_\alpha B_\alpha^{(\lambda_i)}$ for $\bm{w}^{(\lambda_i)}\in \mathbb{R}^{\dim(\irrep)}$.   Then, given the vectors $\bm{e}^{(\operatorname{in})}_{\lambda_i}$, we can compute
\begin{equation}
    \langle O(\bm{\theta}) \rangle=\sum_i\bm{w}^{(\lambda_i)}\cdot\bm{e}^{(\operatorname{out})}_{\lambda_i} \,,
\end{equation}
with each vector of output expectation values is obtained as in Eq.~\eqref{eqn:gsim_ansatz_evolution}. 
\label{thm:multi_evolution}
\end{result}
Here, we can see that Theorem~\ref{theo:scaling} directly implies the following corollary.
\begin{corollary}\label{corro:one}
    Computing expectation values of an observable with support in a constant set of irreps $\{\LC_{\lambda_i}\}_i$  has a time complexity in $\OC(L\max_i\{\dim(\LC_{\lambda_i})^{2}\})$.
\end{corollary}

\subsubsection{Simulation with products of observables each supported by a single irrep}
\label{sec:gsim_correlator}
Taking a step further, we can simulate expectation values of product of observables each supported by a single irrep.
First, let us focus on simulating correlators of the form
 $ \langle O^{(1)} O^{(2)}(\thv) \rangle=\Tr[O U(\bm{\theta}) \rho^{(\operatorname{in})} U^{\dagger}(\bm{\theta})]$ for $O=O^{(1)}O^{(2)}$ with $O^{(1)} \in \mathcal{L}_{\lambda_1}$, and $O^{(2)}\in \mathcal{L}_{\lambda_2}$. For these, we define a $(\dim (\mathcal{L}_{\lambda_1})\times \dim (\mathcal{L}_{\lambda_2}))$-dimensional matrix of expectation values that captures the description of the states $\rho^{(\operatorname{in/out})}$ over products of the basis observables
  \begin{equation}
(\bm{E}^{(\operatorname{in/out})})_{\alpha\beta}\equiv \Tr[B_{\alpha}^{(\lambda_1)} B_{\beta}^{(\lambda_2)}\rho^{(\operatorname{in/out})}]\,.        \label{eqn:initial_final_correlators}
    \end{equation}
 As shown in Appendix~\ref{appendix:simul_product}, we have the following result.

\begin{result}[Simulation of product of two operators]
Consider a circuit of the form in Eq.~\eqref{eqn:PeriodicStructureAnsatz}, and let $O^{(1)}$ and $ O^{(2)}$ be  observables with support in $\mathcal{L}_{\lambda_1}$ and $\mathcal{L}_{\lambda_2}$ respectively such that $O^{(1)}=\sum_\alpha(\bm{w}^{(1)})_\alpha B_{\alpha}^{(\lambda_1)}$ and $O^{(2)}=\sum_{\beta}(\bm{w}^{(2)})_{\beta}B_{\beta}^{(\lambda_2)}$ for $\bm{w}^{(1)} \in \mathbb{R}^{\dim(\mathcal{L}_{\lambda_1})}$ and $\bm{w}^{(2)} \in \mathbb{R}^{\dim(\mathcal{L}_{\lambda_2})}$.  Then, given $\bm{E}^{(\operatorname{in})}$, we can compute
\begin{equation}
    \langle O^{(1)} O^{(2)}(\thv) \rangle=(\bm{w}^{(1)})^T\cdot  \bm{E}^{(\operatorname{out})}\cdot \bm{w}^{(2)}\,,
\end{equation}
with the superscript $T$ denoting the matrix transpose, and with the matrix of output expectation values obtained as 
\small
\begin{equation}
\bm{E}^{(\operatorname{out})}\!\!=\left(\prod_{l=1}^Le^{-i\theta_{l} \bar{\bm{H}}_{l}}\right) \bm{E}^{(\operatorname{in})} \left(\prod_{l=1}^L e^{-i\theta_{l} \bar{\bm{H}}_{l}}\right)^T\!\!.\label{eqn:adjoint_correlator_evolution}
\end{equation}
\normalsize
\label{thm:second_order_correlator_evolution}
\end{result}

The complexity of such simulations is provided in the following Corollary, and we note that it is the same as Corollary.~\ref{corro:one}.
\begin{corollary}\label{corro:two}
    Computing expectation values of products of observables  $O=O^{(1)}O^{(2)}$ with $O^{(1)} \in \mathcal{L}_{\lambda_1},O^{(2)}\in \mathcal{L}_{\lambda_2}$ has a time complexity in $\OC(L\max_i\{\dim(\LC_{\lambda_i})^{2}\})$.
\end{corollary}

Going further, $\gsim$ can be extended to compute expectation values of correlators of any orders, such as a $M$-th order product $O^{(1)}\dots O^{(M)}$ with each $O^{(m)}$ supported by potentially distinct $\irrep$. However, as detailed in Appendix~\ref{appendix:simul_product}, simulating an $M$-th order correlator for arbitrary initial state has complexity $\mathcal{O}(\max_i\{\dim(\LC_{\lambda_i})^{M}\})$ that becomes impractical for large-order $M$. In this work, however we require only $M\leq 2$.

Finally, we note that, in general, a given observable supported by an irrep $\irrep$ can also be decomposed as a product of two or more observables supported in different irreps. Hence, depending on the decomposition adopted, its simulation could be performed through Result~\ref{thm:basic_evolution} or through the results of this section. Different approaches entail different computational scalings such that, in implementations, one may be favored to the other.

\subsubsection{Noisy simulations}\label{sec:noisy_sim}
From the previous discussion, one can already see how some restricted forms of noise can be introduced in the simulations.
Over-- and under-- rotations in the circuits gates of Eq.~\eqref{eqn:PeriodicStructureAnsatz} are readily captured through updates of the angles $\theta_l \rightarrow \theta_l + \varepsilon_l$ by some perturbations $\varepsilon_l$.
Fluctuations or uncertainty in such angles are modeled by repeatedly sampling perturbations $\varepsilon_l$ and averaging over the different realizations.
Going further, one can also incorporate noise channels consisting of probabilistic occurrences of elements of $\mathcal{G}$.
All these, however, may be rather limited as we are forced to consider noise that is generated by $\g$, and introduce some overhead due to repeated sampling of the noise realizations.
However, as we now discuss (and detail further in App.~\ref{app:noise}), often we can do much more and do it more efficiently. 

To capture the noise that can be simulated in $\gsim$, we introduce the notion of the normalizer of an irrep. 
\begin{definition}[Normalizer of $\irrep$]\label{def:normalizer}
Given an irrep $\irrep$, we define its normalizer as all the operators that leaves the subspace invariant under conjugation: 
\begin{equation}
    \stab := \{A \in \mathcal{L}\, | \, A^{\dagger} X A \in \irrep, \, \forall X \in \irrep \}.
\end{equation}
\end{definition}
We stress that by definition of the irreps, as per Definition.~\ref{def:invariant}, we have $\G \subset {\rm N}(\irrep)$ for any $\irrep$. However, 
the normalizer may contain many more operators $A \notin \G$. Notably, whenever the Lie algebra $\g$ has a basis of Pauli operators, any of the exponentially many Pauli operators will belong to the normalizer, and that independently of $\dimg$. As we will soon see, this means that action of Pauli noise channels can be dealt with in $\gsim$. A demonstration of such possibilities is provided in Sec.~\ref{sec:demo}.

Let us consider a noise channel $\Lambda$ with decomposition
\begin{equation}\label{eq:main_kraus_dec}
    \Lambda[\cdot] = \sum_k E_{k} \cdot E^{\dagger}_{k},
\end{equation}
in terms of Kraus operators $E_k$ satisfying $\sum_k E_k^{\dagger} E_k = I$.
Whenever $E_k \in \stab$ \emph{for all} $k$, and if $O \in \irrep$, we have the guarantee that $\Lambda[O] \in \irrep$.
Akin to Definition~\ref{defn:adjoint_representation_unit}, we can define the adjoint representation of any $A\in \stab$, and by linearity of $\Lambda$, as $\dim(\irrep) \times \dim(\irrep)$ matrices.
The latter fully captures the action of $\Lambda$ on $\irrep$: given the adjoint representation $\Phi^{\rm Ad}_{\lambda}(\Lambda)$ of a channel $\Lambda$ and an observable $O=\sum_\alpha (\bm{w})_\alpha B_\alpha^{(\lambda)} \in \irrep$, we get that $\Lambda (O) = \sum_\alpha (\bm{w}')_\alpha B_\alpha^{(\lambda)}$ with updated weights $\bm{w}' = \Phi^{\rm Ad}_{\lambda}(\Lambda)\cdot \bm{w}$.

In noisy simulations with $\gsim$ adjoint representations of the noise channels are simply interleaved in between adjoint representations of the unitary gates and does only incur constant overhead in the simulations. This is formalized in Result~\ref{thm:basic_noisy_evolution} of Appendix~\ref{app:noisy_one_irrep} and captured by the following corollary:
\begin{corollary}\label{corro:three}
    Provided that the noise channels have a Kraus decomposition as per Eq.~\eqref{eq:main_kraus_dec} such that all $E_k \in \stab$, then 
    Corrolary~\ref{corro:one} extends to exact noisy simulations with the same complexities.
\end{corollary}
We see that whenever the noise channels admit a Kraus decompositions with operators supported by $\stab$, we can perform \emph{exact} noisy simulations for observables with support in a constant set of irreps, with the same time and memory complexity as the noiseless simulations.
This is in stark contrast with typical noisy simulations of quantum systems that are either approximate (based on sampling many random realizations of the noise), or exact but incurring a quadratic increase in memory and computing requirements. 
However, as detailed In Appendix.~\ref{app:noisy_product_irrep}, the case of noisy simulations for product of observables needs to be performed though sampling of noise trajectories. Given $S$ of such trajectories, this incurs a complexity $\mathcal{O}(S\max_i\{\dim(\LC_{\lambda_i})^{2}\})$ for an error scaling as $\mathcal{O}(\sqrt{S}^{-1})$.

\subsubsection{Simulation of gradients}
\label{sec:gsim_optim}
In addition to simulating parametrized quantum circuits, in a variational quantum computing setting one aims to \textit{optimize} the parameters to minimize a loss function. In order to benefit from gradient-based training schemes one needs to compute derivatives of expectation values with respect to the circuit parameters. 

Given that Lie-algebraic techniques were originally envisioned for simulating fixed unitary dynamics and not for their optimization, there are no existing methods that use $\gsim$  to compute partial derivatives. However, due to the  form of the evolution in the adjoint representation of Eq.~\eqref{eqn:gsim_ansatz_evolution}, we can port reverse-mode differentiation methods~\cite{jones2020efficientcalculation} to $\gsim$ yielding an efficient algorithm for gradient computation. Such an algorithm is presented in detail in Appendix~\ref{appendix:observable_gradients} and, here, we only remark on the complexity entailed.  
\begin{theorem}[Gradient calculations in $\mathfrak{g}$-sim]
    Computing the gradient of the expectation value of an observable supported by an irrep $\irrep$ using $\gsim$ has a time complexity in $\OC(L\dim(\irrep)^{2})$ for circuits of the form in Eq.~\eqref{eqn:PeriodicStructureAnsatz}.
\end{theorem}

Overall, the previous results form the basis of $\gsim$ comprising noisy and noiseless simulations and optimizations of quantum ciruits. In Appendices~\ref{appendix:analog_computing} and~\ref{appendix:analog_computing_gradients} we detail how this is extended to analog quantum computing. In the following we demonstrate situations where $\gsim$ leads to scalable computations when we specialize to the case where $\mathcal{L}_{\lambda}=i\g$.

\section{Scalability and comparison}\label{sec:demo}
In this section, we aim to demonstrate the capabilities of $\gsim$ and at contrasting it to other established simulation techniques. 
In Sec.~\ref{sec:gsim_scalability}, we collect and summarize conditions for scalability. In Sec.~\ref{sec:wick_comparison}, we contrast $\gsim$ to simulation techniques based on Wick's theorem, including the standard treatment of free-fermion systems. In Sec.~\ref{sec:polynomial_lie_algebra}, we introduce a Lie algebra of interest $\g_{0}$. This Lie algebra will form the basis of both the numerical examples of this section and of many of the tasks tackled later on in the application section (Sec.~\ref{sec:applications}). In Sec.~\ref{sec:magic_state_demonstration_and_benchmark} we report results for noiseless and noisy simulations for system sizes of $n=200$ qubits and discuss distinction of $\gsim$ compared to related free-fermion methods. Additionally, we present details of a benchmark test to numerically verify polynomial resource scalings.

\subsection{Scalability of $\gsim$}
\label{sec:gsim_scalability}
As discussed in the previous section, $\gsim$ recasts quantum evolution as linear algebra problems on vectors in $\mathbb{R}^{\dim(\irrep)}$ and matrices in $\mathbb{R}^{\dim(\irrep) \times \dim(\irrep)}$, as well as linear combinations thereof. Although such techniques  can be applied to \emph{any} system, most dynamical Lie algebras, and their associated irreps, have dimension scaling as $\Theta(4^n)$ (e.g., randomly sampled anti-Hermitian operators generate the whole special unitary algebra $\mf{su}(d)$~\cite{lloyd1995almost}), and thus this framework does not in general yield an asymptotic advantage. Yet, as pointed out earlier, special cases exist where such dimension scales as $\OC(\poly(n))$, such that $\gsim$ can be used to perform classically efficient simulations despite the exponential dimension of $\mathcal{H}$. 

For convenience, we explicitly reiterate the conditions required for classically efficient simulations via $\gsim$.
\begin{enumerate}
    \item The Lie closure of the gate generators must lead to an algebra $\mathfrak{g}$ such that there exists irreps $\irrep$ with $\dim(\irrep)\in\OC(\poly(n))$.
    \item We must know  a Schmidt-orthonormal basis of the polynomially-large $\irrep$, as well as the non-zero representation elements. 
    \item Observables of interest must be supported by polynomially-large irreps, or products of terms, each supported by polynomially-large irreps, up to some constant order $M$.
    \item The expectation values of the irreps basis elements (or their products up to some fixed order $M$), over the initial state must be known. 
\end{enumerate}
Let us make some brief comments regarding these requirements. First, we recall that while most evolutions lead to exponentially sized algebras, there exists examples where the algebra grows polynomially with the number of qubits~\cite{schatzki2022theoretical,anschuetz2022efficient,bonet2020nearly,kokcu2022fixed,qvarfort2022solving, barthe2024gate, somma2024shadow,west2023provably}, including free-fermions, free-bosons systems with permutation symmetries and more.
We further note that studies of Lie-algebra, originally motivated by problems of controllability, have received  renewed attention lately. While recent works~\cite{wiersema2024classification,kokcu2024classification,kazi2024analyzing, aguilar2024full}, have been mostly limited to qubit systems and to Lie algebras generated by Pauli operators, we expect that it would extend to more families of Lie-algebras. 
In turn, these could offer many more applications for the techniques presented.

\subsection{Comparison of $\gsim$ to other Wick-based simulation methods; connections to resource theory}\label{sec:wick_comparison}

At this point, we find it convenient to briefly compare $\gsim$ to a different Lie-algebraic simulation method: Wick's theorem. Commonly, the latter is used to simulate non-interacting free-fermionic evolutions~\cite{mattuck2012guide}, i.e., unitaries generated by Hamiltonians expressed as a combination of quadratic products of fermion creation and annihilation operators. We refer the reader to Appendix~\ref{app:wick}, but also to Refs.~\cite{somma2005quantum} for additional details on this method. In particular, one can show that there exist polynomially-sized Lie algebras $\g$ for which Wick's theorem enables the efficient classical simulation of expectation values for a $O$ in any $\irrep$ that is expressed as a product of polynomially-many products of elements in $i\g$~\cite{wick1950evaluation}, so long as the initial state is either a generalized coherent state~\cite{gilmore1974properties,perelomov1977generalized,zhang1990coherent,delbourgo1977maximum} (i.e., $\rho^{(\operatorname{in})}=V\dya{\rm hw}V\ad$ for $V\in e^{\g}$ and $\ket{\rm hw}$ the \emph{highest weight state} of $\g$), or a linear combination of polynomially-many generalized coherent state. As such, simulation techniques based on Wick's theorem will become exponentially expensive when the initial state has to be decomposed into an exponential number of generalized coherent states. 
On the other hand, when using $\gsim$, the simulation is efficient for general states (provided we can get the associated expectation values in the basis of $\irrep$) but for observables in irreps whose dimension scales only polynomially with the system size, as per Theorem~\ref{eqn:PeriodicStructureAnsatz}. Broadly speaking, we can summarize these results as the following constraints:
\begin{enumerate}
    \item Wick-based simulation: General observable, special state.
    \item $\gsim$-based simulation: General state, special observable.
\end{enumerate}

The previous distinction between Wick-based techniques and  $\gsim$ shows that the two methods have merit
for different types of states. On the one hand, Wick-based simulations fail for states with exponential ``extent''~\cite{aaronson2004improved,heimendahl2021stabilizer,reardon2024fermionic}, or, in the language of resource theory~\cite{chitambar2019quantum}, for states that are expressed as an exponential combination of free states. On the other hand, if the input state has zero component on the poly-size dimensional irreps, then its dynamics therein are trivial (i.e., $\gsim$ would predict zero for all observables taken from such irreps). Hence, to perform non-trivial simulations, one would have to implement $\gsim$ on the exponentially large irreps, in which case it becomes inefficient. In fact, it has been recently proposed that the irrep decomposition of a quantum state is a measure of its  resourcefulness~\cite{bermejo2025characterizing}, and that highly-resourceful states have little-to-no component in the smaller irreps. Crucially, since this notion of resourcefulness is non-equivalent to the  extent~\cite{bermejo2025characterizing} (i.e., states with both order-one and exponentially large extent can fail to have component in the smaller dimensional irreps), we can see that the distinction between Wick-based techniques and $\gsim$ appears to arise on a more fundamental level --- both methods allow for efficient simulations for low-resource states, with the nature of the resource quantified differently for each method. Hence, $\gsim$ can enable new scenarios for simulation beyond those achievable with Wick-based techniques.

\subsection{A polynomially-large Lie-algebra}
\label{sec:polynomial_lie_algebra}
As a specific example of an algebra with $\dimg\in\operatorname{poly}(n)$, we consider a special model with free-fermion mappings where $\dim(\LC_{\g})\in\OC(n^2)$. This algebra, or its subalgebras, will be used in the numerical simulations performed. The algebra under consideration is given by
\small
\begin{equation}
    \mathfrak{g}_{0}={\rm span}_{\mathbb{R}}\left\langle\left(\bigcup_{\mu,\nu=x,y}\{i\sigma^\mu_j \sigma^\nu_{j+1}\}_{j=1}^{n-1}\right)\cup \{i\sigma^z_j\}_{j=1}^n \right\rangle_{\operatorname{Lie}}\!\!\!\!\!,
    \label{eqn:g_tfxy}
\end{equation}
\normalsize
with $\sigma^\alpha_j$ denoting a single-qubit Pauli operator labelled by $\alpha \in \{x, y, z\}$ acting on the qubit $i$.
A basis for $\mathfrak{g}_{0}$ is given by the set of Paulis~\cite{kokcu2022fixed}
\begin{align}\label{eq:algebra-g0}
i\{ \sigma^z_j, \widehat{\sigma^x_i \sigma^x_j},\widehat{\sigma^y_i \sigma^y_j},\widehat{\sigma^x_i \sigma^y_j},\widehat{\sigma^y_i \sigma^x_j}\,|\, 1\leq i <j\leq n\}\,,
\end{align}
where we used the notation $\widehat{A_iB_j}=A_i \sigma^z_{i+1} \cdots \sigma^z_{j-1} B_j$. This algebra  is a representation of $\mathfrak{so}(2n)$, and hence its  dimension is $\operatorname{dim}(\g_0)=n(2n-1)$~\cite{kokcu2022fixed,kazi2022landscape}. All the representation elements can be obtained with a time complexity of $\OC(n^5)$; in the GitHub repository of Ref.~\cite{gsim_github} we have reported these representation elements together with the basis.

Expectation values of the basis elements in Eq.~\eqref{eq:algebra-g0} (or their products) can be efficiently computed on a classical computer for many families of initial states. These include the highest weight states~\cite{somma2005quantum}, any product states or stabilizer states, and even superpositions of polynomially many of these. Interestingly, as we explore in the following section, these also include initial states that are magic states for the class of circuits under consideration, allowing $\gsim$ to efficiently evolve initial states that typical free-fermionic methods could not support --- even some with exponentially large fermionic extent. 
One could even envision the situation where a generic input state is provided as a physical quantum state or as a description of its preparation. One would then estimate these expectation values by using a quantum computer, preparing or accessing $\rho^{(\operatorname{in})}$, and making measurements of all the irrep basis elements. This procedure yields a form of ``classical description''~\cite{anschuetz2022efficient} of the input state that $\gsim$ can use.

\subsection{Demonstration and resource benchmark of $\gsim$}
\label{sec:magic_state_demonstration_and_benchmark}
To demonstrate the power of $\gsim$, we first showcase its application to a quantum simulation task which, to the best of our knowledge, \emph{would not be possible with any other known classical simulation method}. We seek to simulate the dynamics of the transverse field XY (TFXY) model with random magnetic fields, whose Hamiltonian reads
\begin{equation}
    H_{\text{TFXY}}=\sum_{j=1}^{n-1}(\sigma^x_j\sigma^x_{j+1}+\sigma^y_j\sigma^y_{j+1})+\sum_{j=1}^n b_j \sigma^z_j\,,
    \label{eqn:hamiltonian_tfxy_randomfields}
\end{equation}
where the coefficients $b_j$ are randomly drawn from a Gaussian distribution $N(0,\xi^2)$. We note that Hamiltonian dynamics in $\gsim$ can be implemented either by directly exponentiating the adjoint representation of the Hamiltonian \eqref{eqn:adjoint_rep_exponentiation} or by simulating an appropriate quantum circuit \eqref{eqn:gsim_ansatz_evolution}. Here, we apply the latter as it will allow us to insert noise channels in between gates for subsequent noisy demonstration. The circuit is obtained resorting to a first-order Trotterization of the dynamics with a total of 300 Trotter steps with stepsize $\Delta t = 2$.

For our simulations, we initialize the system in the magic state
\begin{equation}
    \ket{\psi}=\left[\frac{\ket{0000}+\ket{0011}+\ket{1100}+e^{i\tau}\ket{1111}}{2}\right]^{\otimes \frac{n}{4}},
    \label{eqn:initial_magic_state}
\end{equation}
for $n=200$ qubits, and aim at computing the dynamics of observables $O \in \irrep=i\g_0$. Although this state has exponential fermionic extent (i.e., it can only be decomposed into an exponential number of Gaussian states)~\cite{dias2023classical,cudby2023gaussian}, the vector of expectation values needed for $\gsim$ can be trivially classically computed, such that simulations are scalable.
This highlights the differences between $\gsim$ and Wick-based methods, as using Wick's theorem along with the state in Eq.~\eqref{eqn:initial_magic_state} would lead to an exponential computational complexity, whereas the cost of $\gsim$ remains polynomial. Hence, this is an example of a simulation that can be performed efficiently with $\gsim$, but not with closely related methods. However, we note that such simulations with $\gsim$ were possible as we limited ourselves to observables supported by a polynomially-large irrep.

\begin{figure*}
    \centering
    \includegraphics[width=0.9\textwidth]{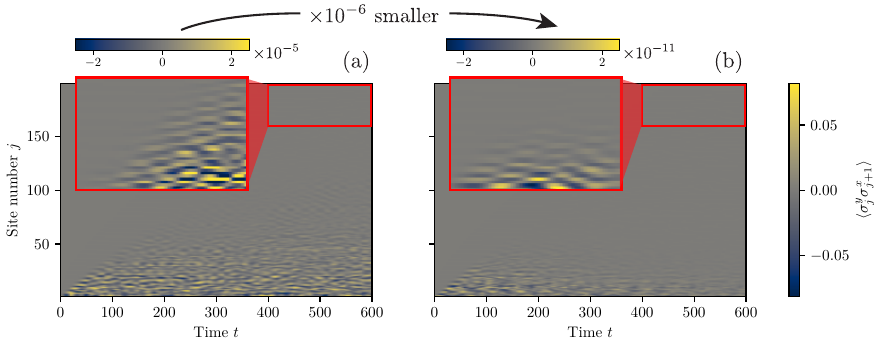}
    \caption{\textbf{Noiseless and noisy dynamics of a 200-qubit magic state using $\gsim$.} We simulate the Trotterized dynamics of an initial magic state \eqref{eqn:initial_magic_state} under a TFXY spin-chain Hamiltonian \eqref{eqn:hamiltonian_tfxy_randomfields}, for $n=200$ qubits and $\tau=2.81$ in $\gsim$. We report values for  correlators of the form $\langle \sigma^y_j \sigma^x_{j+1} \rangle$ along the dynamics. \textbf{(a)} In the absence of noise, we see that the correlations propagate across the whole chain. This simulation ran in 16 minutes on a single CPU core. \textbf{Inset:} Small oscillations in the correlators reach the far end of the chain, confirming that the light cone of the dynamics reaches from the first to last qubit. \textbf{(b)} Including noise in the dynamics, in the form of random 2-qubit Pauli channels~\eqref{eq:pnc} acting after each entangling gate with a fault probability of $p=3\times10^{-4}$, we see that the dissipation weakens the correlations by $6$ orders of magnitude.}
    \label{fig:magic_state}
\end{figure*}

In Figure \ref{fig:magic_state}(a) we explore the dynamics of the system in the absence of noise. In particular we explore the propagation of correlations across the system by evaluating the $\langle \sigma^y_j \sigma^x_{j+1} \rangle $ correlators at different evolution times $t = 0, \hdots, 600$. 
As can be seen, over the time scales probed the correlations propagate to the end of the chain, albeit at small values.
Resolving these values is possible due to the exact nature of the simulations performed.
The entire simulation was completed in 16 minutes on a single CPU core, demonstrating scalability of $\gsim$.

For the Figure \ref{fig:magic_state}(b), we repeat this simulation under the presence of noise using the methods outlined in Section \ref{sec:noisy_sim}. 
Given the irrep of interest, we readily see that any  Pauli operator belongs to the normalizer ${\rm N}(\mathcal{L}_{\lambda}=i \g_0)$ as per Definition.~\ref{def:normalizer}.
Notably, albeit of polynomial dimension, $\irrep$ admits a normalizer containing an exponential number of linearly independent elements.
In turn, this allows us to simulate arbitrary Pauli noise channels.
In our numerical explorations, following every entangling gate, we apply a random 2-qubit Pauli noise channel
\begin{equation}\label{eq:pnc}
    \rho \to (1-p)\rho + \sum_{\mathclap{\nu_1, \nu_2 \in \{1,x,y,z\} \atop (\nu_1, \nu_2) \neq (1,1)}} w_{\nu_1,\nu_2}\sigma_j^{\nu_1} \sigma_k^{\nu_2} \rho\ \sigma_k^{\nu_2} \sigma_j^{\nu_1}
\end{equation}
where the $w_{\nu_1,\nu_2}$ are drawn from a uniform distribution $\mathcal{U}(0,1)$ and then normalized such that they sum to $p$. For the simulation depicted here, we choose an overall per-gate fault probability of $p=3\times 10^{-4}$. Under the presence of noise, the correlations dissipate, although are not entirely eliminated: oscillations can still be seen at the far end of the chain, albeit $6$ orders of magnitude smaller than for the noiseless case. 

These results highlight the exact nature of noisy simulations with $\gsim$: Any method relying on sampling noise realizations would have required of the order of $10^{11}$ many repetitions to resolve such oscillations.
Similarly, classical simulations of noisy dynamics involving the truncation of low-weight contributions in the Heisenberg picture~\cite{fontana2023classical,cirstoiu2024fourier} also requires simulating many trajectories.

\begin{figure}
    \centering
    \includegraphics[width=1\columnwidth]{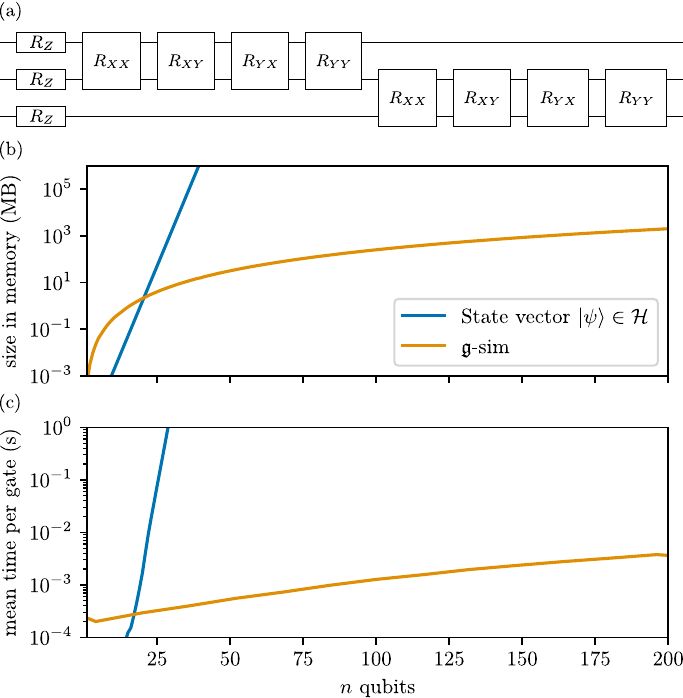}
    \caption{\textbf{Performance benchmarks of $\mathfrak{g}$-sim.} (a) Circuit used for the benchmark at $n=3$ qubits. The set of generators are presented in Eq.~\eqref{eqn:g_tfxy}. Here, $R_Z$ denotes a rotation about the $Z$ axis, and $R_{\mu\nu}$ a rotation generated by the Pauli operator $\sigma^\mu\sigma^\nu$.  We compare the memory (b) and compute time (c) requirements for state vector simulations (blue) and $\mathfrak{g}$-sim (orange) for the unitary shown in panel (a). }
    \label{fig:benchmarks}
\end{figure}

With a non-trivial application of $\gsim$ demonstrated, we now consider its performance and scalability more systematically with a benchmarking task, aiming to verify its polynomial resource scaling numerically. We take the circuits $U(\thv)$ to be composed of a single layer of gates generated by the set of operators $\left(\bigcup_{\mu,\nu=x,y}\{\sigma^\mu_j \sigma^\nu_{j+1}\}_{j=1}^{n-1}\right)\cup \{\sigma^z_j\}_{j=1}^n $.
In Fig.~\ref{fig:benchmarks}(a) we schematically show the circuit for $n=3$.

For the benchmarking task, we evaluate expectation values of random observables in the algebra $\g_0$  for the state obtained by applying $U(\bm{\theta})$, with parameters uniformly sampled in the interval $[0,2\pi]$, to an input state $\rho^{(\operatorname{in})}=\dya{0}^{\otimes n}$. As was the case for the state in Eq.~\eqref{eqn:initial_magic_state}, we can easily compute  expectation values of $\rho^{(\operatorname{in})}$ for any of the Pauli operators forming the basis of $\g_0$.

In Fig.~\ref{fig:benchmarks}(b) we show  benchmark results in terms of memory and compute time for $\gsim$ and for full  state-vector simulations methods. The benchmarks are computed on a single CPU core - using TensorFlow Quantum for state vector methods, and our Python implementation for $\gsim$. Therein we can see that  both memory and compute time scale as $\Omega(2^{n})$ for state vector methods, and as $\mathcal{O} (\poly(n))$ for $\mathfrak{g}$-sim. The exponential scaling of state-vector simulations makes simulations with $n\geq 30$ qubits intractable on the device used. However, using $\gsim$ we are able to simulate systems of up to 200 qubits on a single core of a modern CPU. Further comments on the benchmarking procedure are given in Appendix~\ref{appendix:benchmarking}.

These demonstrations highlight the capabilities and scalability of our framework for simulations, and in particular its differences and advantages relative to similar methods. These did not however require any optimization; in the next section, we make use of optimization and explore applications of $\gsim$ to more sophisticated key problems in variational quantum computing.

\section{Applications}\label{sec:applications}
Having presented the framework for $\gsim$, we now explore its benefits in concrete problems. These include a study of the landscape of VQAs highlighting the  overparametrization phenomenon (Sec.~\ref{sec:overparametrization}), designing improved initialization of quantum circuits parameters (Sec.~\ref{sec:pretraining}), problems of circuit synthesis (Sec.~\ref{sec:compilation}), and demonstration of the training of a quantum-phase classifier (Sec.~\ref{sec:supervised_QML}).
These illustrate the broad range of applications that can be addressed with $\gsim$. 

We note that in all cases, we will implement $\gsim$ based on either $\g_0$ or subalgebras $\g\subseteq \g_0$ obtained by removing some of the generators from Eq.~\eqref{eqn:g_tfxy}. In each section, we will specify which subalgebra of $\g_0$ we will be working with, and a more detailed outline of the relevant subalgebras is given in Appendix \ref{appendix:algebra_and_subalgebras}. We further remark that by definition, we will always be simulating circuits composed of single- and two-qubit gates with local connectivity - that is, circuits that could be implemented on most quantum hardware without incurring large compilation overhead. Details of these circuits are summarized in Table~\ref{tab:ansatze} of Appendix~\ref{appendix:ansatz}.

\subsection{Characterizing VQAs}
\label{sec:overparametrization}
VQAs aim to enable near-term quantum advantage by means of a hybrid quantum-classical training loop, where some of the problem difficulty is offloaded to an optimization problem on a classical co-processor. However, this optimization problem is difficult in general~\cite{bittel2021training}, and much remains to be learned about its scalability. Analytic study of optimization landscapes is difficult and limited in scope, while computational study is hindered by the exponential resource costs of state vector simulation. 
In this context, scalable classical simulations allow one to dequantize certain architectures~\cite{cerezo2023does}, but also to  characterize trainability of VQAs at system sizes that would otherwise be intractable. A previous study has used free-fermion simulations for this purpose \cite{matos2022characterization}, but using $\mathfrak{g}$-sim expands the set of systems that can be studied. As a warm-up problem, we demonstrate the onset of overparametrization in VQE problems for the TFXY model at up to $n=50$ qubits.

\subsubsection{Overparametrization in VQE}

Overparametrization \cite{neyshabur2018role} is a surprising phenomenon in classical neural networks, where training a neural network with a capacity larger than that which is necessary to represent the training data distribution may lead to improved performance \cite{zhang2021understanding,allen2019learning,du2019gradient,buhai2020empirical} and even provable convergence results \cite{du2018gradient,brutzkus2018sgd}, rather than to the overfitting and training difficulties one may na\"ively expect. This phenomenon was generalized to quantum circuits in Ref.~\cite{larocca2021theory}, 
where it was shown that the model capacity of circuits of the form Eq.~\eqref{eqn:PeriodicStructureAnsatz} can be quantified by $\dimg$. In turn, overparametrization is characterized by a `computational phase transition' happening at a critical number of parameters $N_p^{(\text{crit})}\leq \dimg$, below which the circuit is hard to train and above which it becomes easy to train. Exact values of this critical threshold $N_p^{(\text{crit})}$ are state-dependent, and often hard to assess analytically as they rely on the conjugation relationship between the initial state and the Lie algebra $\mathfrak{g}$ of the circuit. Thus, exact details of this phenomenon are best probed numerically. However, initial numerical demonstrations of the phenomenon were only provided for systems of 2-10 qubits~\cite{larocca2021theory}. Here we demonstrate the phenomenon in problems of up to 50 qubits.

\subsubsection{Simulation results}
To probe overparametrization, we consider a VQE task where the goal is to prepare the ground state of the TFXY model of Eq.\eqref{eqn:hamiltonian_tfxy_randomfields}.
As an ansatz for $U(\thv)$ we consider a  Hamiltonian variational ansatz~\cite{wecker2015progress,wiersema2020exploring} of the form in Eq.~\eqref{eqn:PeriodicStructureAnsatz}  with gate generators taken from the set $\{\sigma^x_j\sigma^x_{j+1},\sigma^y_j\sigma^y_{j+1}\}_{j=1}^{n-1}\cup\{\sigma^z_j\}_{j=1}^n$. We refer the reader to Appendix~\ref{appendix:ansatz} for additional details on this circuit. One can verify that the dynamical Lie algebra associated with this set of generators is precisely $\g_0$ in Eq.~\eqref{eqn:g_tfxy}, such that we can use the representation elements reported  in Ref.~\cite{gsim_github}. Given that $\operatorname{dim}(\g_0)=n(2n-1)$ and that the gates in $U(\thv)$ can be parallelized, then overparametrization is achievable with linear circuit depth.
To solve the VQE problem, the parameters are optimized using the L-BFGS algorithm with the gradient evaluated according to Appendix~\ref{appendix:observable_gradients} to minimize the energy
\begin{equation}
    E_{\operatorname{TFXY}}(\bm{\theta})=\bra{\bm{0}}U^{\dagger}(\bm{\theta})H_{\operatorname{TFXY}}U(\bm{\theta})\ket{\bm{0}}.
\end{equation}
Note that the measurement operator $H_{\operatorname{TFXY}}$ is, by definition, fully supported within $\g_0$. Moreover, since the initial state is the all zero state we can readily construct the vector $\bm{e}^{(\operatorname{in})}$. That is, we have all the ingredients for $\gsim$.

For each of the system sizes probed, from $n=20$ to $50$ qubits, we perform optimization over circuits with a varied number of layers $L$ chosen to result in a number $N_p$ of circuit parameters spanning the range $[1,\operatorname{dim}(\mathfrak{g}_0)]$, and random magnetic field amplitude $\xi=0.1$. For any of the circuits studied, optimizations are repeated over $50$ randomized initial parameter values and fields. Results of this study are reported in Fig.~\ref{fig:overparametrization}.

For a fixed system size of $n=50$ qubits, we display in Fig.~\ref{fig:overparametrization}(a) all the optimization traces. We observe the convergence behavior of the approximate training error $\epsilon_{\operatorname{TFXY}}(\bm{\theta})=(E_{\operatorname{TFXY}}(\bm{\theta})-E_{\operatorname{min}})/\|H_{\operatorname{TFXY}}\|_{\operatorname{HS}_1}$, where  $E_{\operatorname{min}}$ is the lowest energy discovered by VQE across all depths and $\|\cdot\|_{\operatorname{HS}_1}=\|\cdot\|_{\operatorname{HS}}/2^n$ for $\|\cdot\|_{\operatorname{HS}}$ the Hilbert-Schmidt norm (i.e., we normalize the error relative to the energy scale of the Hamiltonian). Optimization traces exhibit a clear change in the hardness of the optimization problem as the depth of the circuits is varied. 
Below a critical threshold of $L\approx 6$ layers, we can see severe trainability issues where none of the optimization manages to converge to the ground state energy. 
However, for circuits with $L>6$ layers a sudden transition to trainability is observed and solutions converge towards the minimum energy. This phase transition in computational complexity indicates the onset of overparametrization.

More systematically, in Fig.~\ref{fig:overparametrization}(b), we study this phenomenon across varying system sizes $n$, from 20 to 50 qubits, and report the probability of convergence to $\epsilon_{\operatorname{TFXY}}<10^{-4}$ as a function of the number of circuit parameters in units of $\operatorname{dim}(\mathfrak{g}_0)$. 
We observe that the transition to the trainable region (large convergence probability) occurs consistently across all system sizes at $N_p^{(\text{crit})}\approx 0.3\dimg$ irrespective of $n$. 
This supports the analytic results of Ref.~\cite{larocca2021theory}, and demonstrates the phenomenon of overparametrization at system sizes beyond what could be simulated with state vector simulations. Overall, such detailed analysis of VQA trainability at scale, that requires repetitions over many optimizations and many circuit sizes, is made possible by $\gsim$.

\begin{figure}
    \centering
    \includegraphics[width=0.48\textwidth]{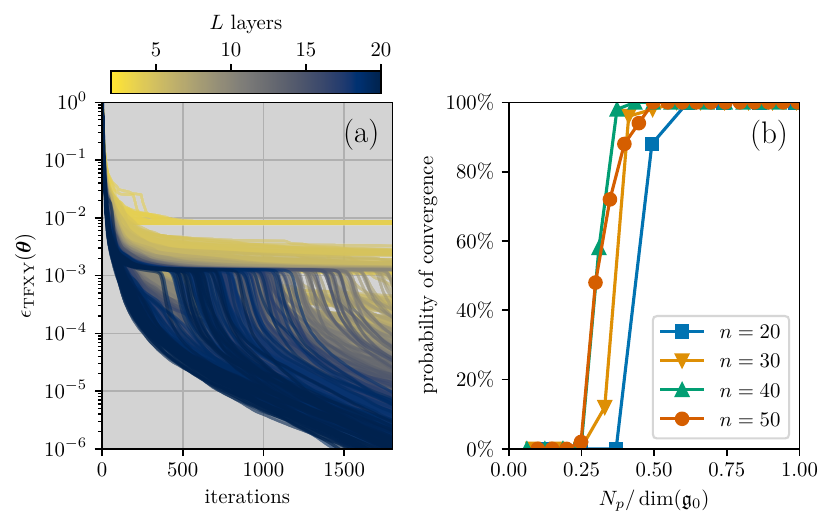}
    \caption{\textbf{Overparametrization in large system sizes.} \textbf{(a)} Convergence traces at $n=50$ qubits of the approximate training error $\epsilon_{\operatorname{TFXY}}(\bm{\theta})$. \textbf{(b)} The probability of converging to $\epsilon_{\operatorname{TFXY}}<10^{-4}$ for uniform random initialization of $\bm{\theta}$ as measured by 50 samples, at varied circuit depths, with corresponding number of parameters $N_p$ reported as a fraction of $\operatorname{dim}(\g_0)$, and for $n$ ranging from 20 to 50 qubits. }
    \label{fig:overparametrization}
\end{figure}

Overparametrization is crucial throughout the remainder of this work. It is provably achieved for circuits with $N_p \geq \dimg$ parameters, and thus the tractably overparmeterizable models are those with $\dimg\in\OC(\operatorname{poly}(n))$.  
In cases where the relevant loss function depends entirely on observables supported by the algebra of the circuit (e.g., VQE with Hamiltonian variational ansatz~\cite{wecker2015progress,ho2019efficient,wiersema2020exploring,cade2020strategies} or QAOA), circuits overparametrizable in polynomial depth \emph{are exactly those that can be efficiently simulated with $\mathfrak{g}$-sim}. 
Hence, this phenomenon ensures that for any problem that is efficiently simulable in $\gsim$, we can guarantee good trainability using a circuit of only polynomial gate-count. Crucially, this ensures a correctly trained prior model for our pre-training scheme (Sec.~\ref{sec:pretraining}), the ability to efficiently approximately compile unitaries in $\mathcal{G}$ to linear-depth circuits (Sec.~\ref{sec:compilation}), and good trainability of QML models (Sec.~\ref{sec:supervised_QML}).

\subsection{Pre-training quantum circuits}
\label{sec:pretraining}
While the overparametrization regime guarantees trainability from \emph{arbitrary} initial parameter values, in general cases where $\dimg\in\Theta(4^n)$ (i.e., cases where classical simulations are not possible anymore) it does not result in a scalable strategy. Indeed, trying to access the overparametrized regime  would require constructing and training exponentially deep circuits, which is intractable for large problem sizes. 
Still, it is hoped that with an \emph{adequate} choice of initial parameters one could train the circuits before the onset of overparametrization. 
Moreover, issues of barren plateaus~\cite{mcclean2018barren,cerezo2021cost,holmes2021connecting,larocca2021diagnosing} (another crucial aspect of trainability) indicate exponentially vanishing gradients \textit{on average}, but do not imply that the entire optimization landscape is flat. 
In fact, these are always accompanied by `narrow gorges' in which minima are surrounded by regions of large gradient~\cite{arrasmith2021equivalence}. 
Hence, improvement in the trainability of quantum circuits can be achieved by means of `pre-training' an ansatz such that its initial parameter values are sufficiently close to the optimal solution. As such, initialization by means of pre-training has received significant attention lately \cite{grant2019initialization,verdon2019learning,sauvage2021flip,rad2022surviving,liu2022mitigating,mitarai2022quadratic,ravi2022cafqa,cheng2022clifford,dborin2022matrix,mele2022avoiding,rudolph2022synergy}, and are proposed as one of the most promising methods to requantize certain VQAs~\cite{cerezo2023does}.

In this section, we demonstrate the use of $\gsim$ for the initialization of VQAs. 
The idea (Sec.~\ref{sec:pre_training_strategy}) is to  perform pre-training on a related auxiliary problem that induces a scalable Lie algebra, and to transfer the solution found as initial parameters for the circuit addressing the target problem. 
It is expected that the closer the auxiliary problem is to the target one, the more efficient such a transfer will be. As a first example (Sec.~\ref{sec:LTFIM}), we study ground state preparation of the longitudinal-transverse field Ising model (LTFIM) via solving the transverse field Ising model (TFIM) in the first place.
LTFIM only differs from TFIM by a few additional generators (the longitudinal fields) and we find substantial improvement both in terms of the fidelity of the ground states prepared and the magnitude of the initial gradients when utilizing pre-training.
More surprisingly, we also show that in problems of QAOA (Sec.~\ref{sec:QAOA}), for which target and auxiliary tasks differ substantially, improved performances can still be achieved over a significant number of (but not all) problem instances. 

\subsubsection{Pre-training strategy with $\mathfrak{g}$-sim}
\label{sec:pre_training_strategy}

Let us consider again here a VQE task where the goal is to prepare the ground state of an Hamiltonian $H$ that can be decomposed as $H=\sum_j c_j h_j$, where $c_j$ are real valued coefficients. We then define  the Lie algebra $\mathfrak{g}_{\text{target}}=\langle\{ih_j\}\rangle_{{\rm Lie}}$, i.e., the algebra generated by the individual terms in $H$. The goal is to construct a circuit $U(\thv)$ to prepare the ground state of $H$ that is generated by some elements of $\mathfrak{g}_{\text{target}}$. In what follows, we will  assume that $\operatorname{dim}(\mathfrak{g}_{\text{target}})\notin \OC(\poly(n))$ (else the full VQE problem could be efficiently simulated with $\mathfrak{g}$-sim), meaning that $U(\thv)$ should not be constructed from a generating set of $\mathfrak{g}_{\text{target}}$. Hence, the scheme is as follows:
\begin{enumerate}
    \item Identify a subset of the operators in $i\mathfrak{g}_{\text{target}}$, denoted as $\{h_k\}_k$ such that their Lie closure  $\mathfrak{g}_{\text{aux}}=\langle\{ih_k\}\rangle_{{\rm Lie}}$, is a subalgebra $\mathfrak{g}_{\text{aux}}\subset \mathfrak{g}_{\text{target}}$ with $\operatorname{dim}(\mathfrak{g}_{\text{aux}})\in\OC(\operatorname{poly}(n))$. Construct a proxy Hamiltonian $H_{\text{aux}}=\sum_k c_k h_k$. As we will see in the examples below, the choice of $c_k$ is informed by the task at hand.    
    \item On a classical computer, use $\gsim$ to solve VQE on $H_{\text{aux}}$ using an ansatz generated by terms of $\mathfrak{g}_{\text{aux}}$ and with a number of parameters allowing for overparametrization. 
    \item Extend the trained ansatz with \emph{new} gates generated by (some) of the terms in $\mathfrak{g}_{\text{target}}\setminus\mathfrak{g}_{\text{aux}}$. These new gates are initialized to the identity. 
    \item On a quantum computer, solve the VQE for $H$ starting from the ansatz constructed in step 3.
\end{enumerate}
Although presented in the context of VQE, similar strategies could be applied to QML problems.
\subsubsection{Pre-training VQE for the LTFIM}
\label{sec:LTFIM}

We begin by recalling that the LTFIM is a paradigmatic spin-chain model providing a prototypical example of a quantum phase transition. Its Hamiltonian reads
\begin{equation}
    H_{\text{LTFIM}}=h_{xx}\sum_{j=1}^{n-1}\sigma^x_j\sigma^x_{j+1}+h_z\sum_{j=1}^n \sigma^z_j+h_x\sum_{j=1}^n \sigma^x_j \,,
    \label{eqn:ltfim}
\end{equation}
where $h_{xx},h_z,h_x\in\mathbb{R}$. It can be verified that the algebra $\mathfrak{g}_{\text{target}}$ obtained from the terms in $H_{\text{LTFIM}}$ has exponential dimension $\mathfrak{g}_{\text{LTFIM}}\in\Theta(4^n)$ \cite{larocca2021diagnosing}. This renders it intractable in $\mathfrak{g}$-sim, induces a barren plateau for deep Hamiltonian variational circuits~\cite{larocca2021diagnosing}, and precludes efficient overparametrization \cite{larocca2021theory}, thus making application of the VQE to identify the ground state of Eq.~\eqref{eqn:ltfim} highly non-trivial.

Instead, by dropping some terms in $H_{\text{LTFIM}}$ we obtain the TFIM Hamiltonian,  given by
\begin{equation}
    H_{\text{TFIM}}=h_{xx}\sum_{j=1}^{n-1}\sigma^x_j\sigma^x_{j+1}+h_z\sum_{j=1}^n \sigma^z_j\,.
    \label{eqn:TFIM_hamiltonian}
\end{equation}
Notably, taking the Lie closure of the terms in $H_{\text{TFIM}}$ we obtain $\mathfrak{g}_{\text{aux}}\subseteq \g_0$.  Thus, we have successfully identified a subset of operators in $i\mathfrak{g}_{\text{target}}$ leading to a an algebra whose dimension is in $\OC(\operatorname{poly}(n))$, making it overparametrizable~\cite{larocca2021theory} and classically tractable in $\gsim$.

The pre-training strategy is now applied to ground state preparations of $H_{\text{LTFIM}}$.
Following Sec.~\ref{sec:pre_training_strategy}, we begin by constructing an ansatz with $L=\operatorname{dim}(\mathfrak{g}_{0})$ and gates generated by terms appearing in $H_{\text{TFIM}}$:
\begin{equation}
    U(\bm{\theta})=\prod_{l=1}^{\operatorname{dim}(\mathfrak{g}_{0})}e^{-i\theta_{l,1}\sum_{j=1}^{n-1}\sigma^x_j\sigma^x_{j+1}}e^{-i\theta_{l,2}\sum_{j=1}^n \sigma^z_j},\label{eq:u-tfim}
\end{equation}

Parameters of the ansatz are initialized with random uniform values, and we train them to prepare the ground state of $H_{\text{TFIM}}$ using $\mathfrak{g}$-sim. Note that since the dynamical Lie algebra $\mathfrak{g}_{\text{aux}}$ associated with Eq.~\eqref{eq:u-tfim} is a subalgebra of $\g_0$, we can use the representation elements that we have pre-computed for $\g_0$.
This can be seen by noting that $e^{-i\theta_{l,1}\sum_{j=1}^{n-1}\sigma^x_j\sigma^x_{j+1}}=\prod_{j=1}^{n-1}e^{-i\theta_{l,1}\sigma^x_j\sigma^x_{j+1}}$, and that each gate generator is an element $\sigma^x_j\sigma^x_{j+1}$ of  $\g_0$. We detail further how to best make use of $\gsim$ with generators that are sums of Pauli operators in Appendix~\ref{appendix:sparsity}.

Once  the parameters in Eq.~\eqref{eq:u-tfim} are trained to prepare the ground state of  $H_{\text{TFIM}}$,  we modify the circuit by inserting new gates generated by the remaining term of $H_{\text{LTFIM}}$, yielding the ansatz
\begin{align}
    U(\bm{\theta},\bm{\phi})=\prod_{l=1}^{\operatorname{dim}(\mathfrak{g}_{0})} \Big[ &e^{-i\theta_{l,1}\sum_{j=1}^{n-1}\sigma^x_j\sigma^x_{j+1}}e^{-i\theta_{l,2}\sum_{j=1}^n \sigma^z_j} \nonumber\\ &\quad  e^{-i\phi_{l}\sum_{j=1}^n \sigma^x_j} \Big] \,.
    \label{eqn:LTFIM_ansatz}
\end{align}

The new gates are initialized to the identity by setting $\bm{\phi}=(0,0,\ldots)$, while the other gates retain their pre-trained values. Finally, we proceed by training the full ansatz of Eq.~\eqref{eqn:LTFIM_ansatz} to minimize the expectation value of $H_{\text{LTFIM}}$. Since $\dim(\mathfrak{g}_{\operatorname{LTFIM}})\in\Theta(4^n)$, this last step must be performed using state vector simulations (we use TensorFlow Quantum \cite{broughton2020tensorflow}) thus limiting the system sizes that can be probed. 

\begin{figure}
    \centering
    \includegraphics[width=0.48\textwidth]{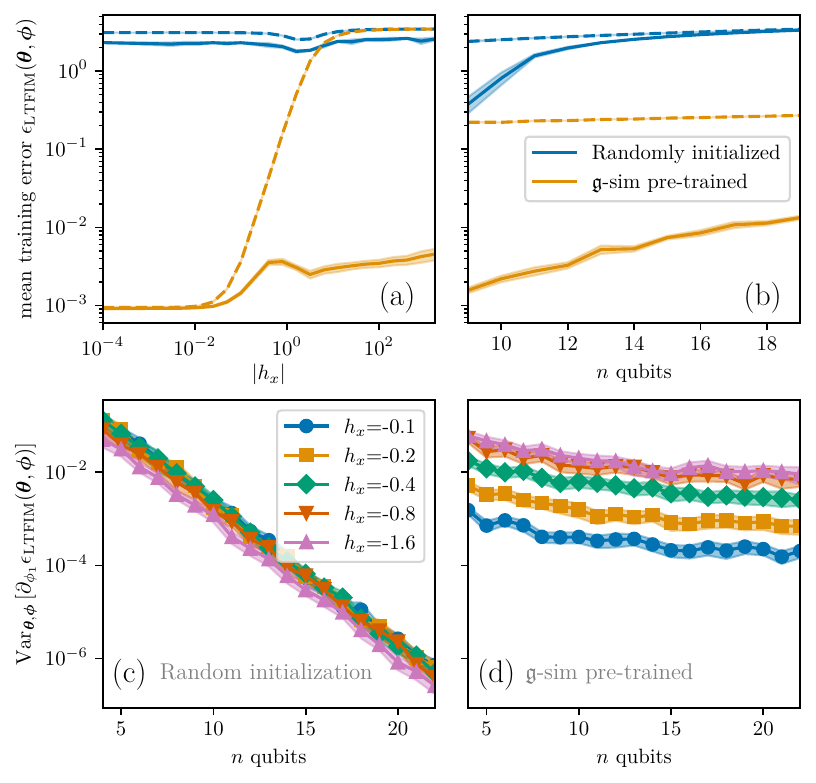}
    \caption{\textbf{Pre-training VQE for LTFIM using $\mathfrak{g}$-sim.} \textbf{(a, b)} Mean training error $\epsilon_{\operatorname{LTFIM}}(\bm{\theta},\bm{\phi})$ (Eq.~\eqref{eqn:error_ltfim}) for VQE  with random initialization (blue) and pre-trained with $\mathfrak{g}$-sim (orange), (a) at $n=12$ qubits and varying longitudinal field strengths $h_x$ and (b) at field strength $h_x=-1.0$ and varying system sizes $n$. Dashed lines represent initial ansatz configurations, while solid lines represent trained ans\"{a}tze. \textbf{(c, d)} Comparison of gradient variances at varying system size $n$ and longitudinal field strengths $h_x$ for (c) uniform random parameter initialization and (d) $\mathfrak{g}$-sim pre-training. Shaded bars represent bootstrapped 95\% confidence intervals.}
    \label{fig:pretraining_ltfim}
\end{figure}
In Fig.~\ref{fig:pretraining_ltfim}, we report a comparison of the pre-training strategy discussed versus uniform random initialization of the circuit parameters, for Hamiltonians with  $h_z=-1$ and $h_{xx}=1$. As can be seen in Figures \ref{fig:pretraining_ltfim}(a, b), pre-training yields improvements by several orders of magnitude in terms of the error 
\begin{equation}
    \epsilon_{\operatorname{LTFIM}}(\bm{\theta},\bm{\phi})=\frac{\bra{\bm{0}}U^\dagger(\bm{\theta},\bm{\phi})H_{\text{LTFIM}}U(\bm{\theta},\bm{\phi})\ket{\bm{0}}-E_{\text{min}}}{\|H_{\text{LTFIM}}\|_{\operatorname{HS}_1}},
    \label{eqn:error_ltfim}
\end{equation}
where $E_{\text{min}}$ is the ground state energy of $H_{\text{LTFIM}}$ obtained by exact diagonalization. 

In Fig.~\ref{fig:pretraining_ltfim}(a), we compare these errors as a function of the longitudinal field strength $|h_x|$. 
The prepared state resulting from the pre-training strategy (with errors depicted as an orange dashed line) becomes a poorer approximation to the ground state as $|h_x|$ increases, and even sometimes at par with random initialization (blue dashed line). Still, after the final step of optimization (with errors depicted as plain lines) we found pre-training to consistently enable accurate ground state preparation of $H_{\text{LTFIM}}$. 
Across all the values of $h_x$ studied, the pre-training strategy is the most favorable initialization.

In Fig.~\ref{fig:pretraining_ltfim}(b), we report final errors scaling with the system size $n$, noting that the randomly initialized circuits quickly become untrainable as $n$ increases while the pre-trained ansatz only exhibits mild decline in trainability.

In addition to these improved errors, we observe a mitigation of the barren plateau effect. Due to the exponential dimension of $\mathfrak{g}_{\operatorname{LTFIM}}$,  in the case of random initialization one would expect gradient variance scaling as \cite{larocca2021diagnosing}
\begin{equation}
    \operatorname{Var}_{\bm{\theta},\bm{\phi}}\left[\partial_{\phi_m}\epsilon_{\text{LTFIM}}(\bm{\theta},\bm{\phi})\right]\in\mathcal{O}\left(\frac{1}{2^n}\right),
\end{equation}
thus necessitating $\Theta(2^n)$ circuit repetitions to distinguish small gradient values from statistical shot noise. This is verified numerically in Fig.~\ref{fig:pretraining_ltfim}(c) which, over varied values of the field $h_x$, showcases exponentially vanishing gradients. 
However, in the case of $\mathfrak{g}$-sim pre-training, we can see in Fig.~\ref{fig:pretraining_ltfim}(d), that the gradient variances vanish at a much slower rate in $n$, effectively mitigating appearance of the barren plateau effect and thus extending the scalability of VQE on this system.
\subsubsection{Pre-training QAOA}
\label{sec:QAOA}
The quantum approximate optimization algorithm (QAOA) is a VQA that attempts to solve combinatorial optimization problems \cite{farhi2014quantum}. Specifically, we consider here its use to approximate solutions of MaxCut problems. We recall that given a graph $G$ with edges $E$ and vertices $V$, the maximum cut (MaxCut) problem is to find a partition of V into two sets $S$ and $T$ that maximizes the number of edges $e \in E$ having endpoints in both $S$ and $T$.    Encoding a partition as a bitstring $\bm{z}\in\{0,1\}^n$, its fitness for the MaxCut problem can be quantified by the approximation ratio $r$ defined as
    \begin{equation}
        r(\bm{z})\equiv\frac{C(\bm{z})}{\argmax_{\bm{z}}C(\bm{z})}, \; C(\bm{z})\equiv\sum_{(m,l)\in E}z_m(1-z_l) \,.
    \label{eqn:classical_approx_ratio}
    \end{equation}
For general graph problems, finding an exact MaxCut solution is NP-hard~\cite{karp1972reducibility}. 
Still, the Goemans-Williamson (GW) algorithm \cite{goemans1995improved} allows one to efficiently find an approximation with a ratio of at least $r_{\rm GW}\approx0.878$. To obtain quantum advantage with QAOA, one must outperform this threshold.

QAOA recasts a MaxCut problem as a VQE problem, where the goal is to find the ground state of the phase Hamiltonian
\begin{equation}
    H_G = \frac{1}{2}\sum_{(m,l)\in E}(\sigma^z_m\sigma^z_l-I) \,,
    \label{eqn:ham_maxcut}
\end{equation}
that depends on the underlying graph $G$. 
To prepare such a ground state, one applies a circuit
\begin{equation}
U(\bm{\beta},\bm{\gamma})=\prod_{m=1}^p e^{-i\beta_m H_M}e^{-i\gamma_m H_G} \,,
\label{eqn:standard_qaoa_ansatz}
\end{equation}
where $H_M=\sum_{j=1}^n \sigma^x_j$ is the so-called mixing Hamiltonian, to an initial state $\ket{+}^{\otimes n}$. 
By optimizing the variational parameters ($\bm{\beta}$ and $\bm{\gamma}$) to minimize the expectation value $\langle +|^{\otimes n}U\ad(\bm{\beta},\bm{\gamma})H_PU(\bm{\beta},\bm{\gamma})|+\rangle^{\otimes n}$ and measuring the resulting state in the computational basis, one may construct approximate solutions to MaxCut. In correspondence to Eq.~\eqref{eqn:classical_approx_ratio}, the approximation ratio of the solution generated by QAOA is defined as
\begin{equation}
    r(\bm{\beta},\bm{\gamma})\equiv\frac{\bra{+}^{\otimes n} U^\dagger(\bm{\beta},\bm{\gamma}) H_G U(\bm{\beta},\bm{\gamma})\ket{+}^{\otimes n}}{\bra{\psi_{\operatorname{GS}}} H_G \ket{\psi_{\operatorname{GS}}}} \,,
\end{equation}
where $\ket{\psi_{\operatorname{GS}}}$ is the ground state of the Hamiltonian $H_G$. 

While optimal parameters can be identified for $p=1$~\cite{ozaeta2022expectation}, optimizing them for larger depth (where improved solutions can be found) remains a challenge. Indeed,
several works show that general problem instances are likely to experience unfavorable optimization landscapes in the absence of any special underlying structure \cite{zhou2020quantum,kossmann2022deep}. These point toward the necessity of finding pre-training strategies for deep-circuit QAOA.

It has been reported that for most choices of $G$ the  Lie algebra associated with Eq.~\eqref{eqn:standard_qaoa_ansatz} will have dimension in $\Omega(2^n)$~\cite{larocca2021diagnosing,kazi2022landscape}. Hence, to apply our pre-training strategy, we  need to identify an algebra with polynomial dimension that is related to the original problem. For the circuit of Eq.~\eqref{eqn:standard_qaoa_ansatz}, this is the case for the path graph $G=P_n$ on $n$-vertices. For such a choice, the Lie algebra  $\mathfrak{g}_{\text{QAOA,}P_n}={\rm span}\langle \{iH_{P_n}, iH_M \}\rangle_{\operatorname{Lie}}$ is such that (up to a change of basis, where the Pauli $\sigma^x$ and $\sigma^z$ are interchanged) $\mathfrak{g}_{\text{QAOA,}P_n}\subset \g_0$, and therefore has $\operatorname{dim}(\mathfrak{g}_{\text{QAOA,}P_n})\in\OC(\operatorname{poly}(n))$ (see Appendix~\ref{appendix:algebra_and_subalgebras}). Hence, we can again use the representation elements of $\g_0$, and start by training such circuit for $p=\operatorname{dim}(\mathfrak{g}_{\text{QAOA,}P_n})$. 

The pre-trained parameters $\bm{\beta}$ and $\bm{\gamma}$  can then be  used to initialize the modified ansatz
\begin{align}
\begin{split}
    U_G(\bm{\alpha},\bm{\beta},\bm{\gamma})= \prod_{l=1}^{p} \Big[ & e^{-i\beta_l H_M}e^{-i\gamma_l H_{P_n}} 
    e^{-i\alpha_l H_G}  \Big]
    \label{eqn:modified_qaoa_ansatz}
\end{split}
\end{align}
for any general graph $G\neq P_n$, with the new parameters initialized to $\bm{\alpha}=(0,0,\ldots)$. We then train the parameters $\bm{\alpha}$, $\bm{\beta}$, $\bm{\gamma}$  to minimize the expectation value of $H_G$.

\begin{figure}
    \centering
    \includegraphics[width=0.48\textwidth]{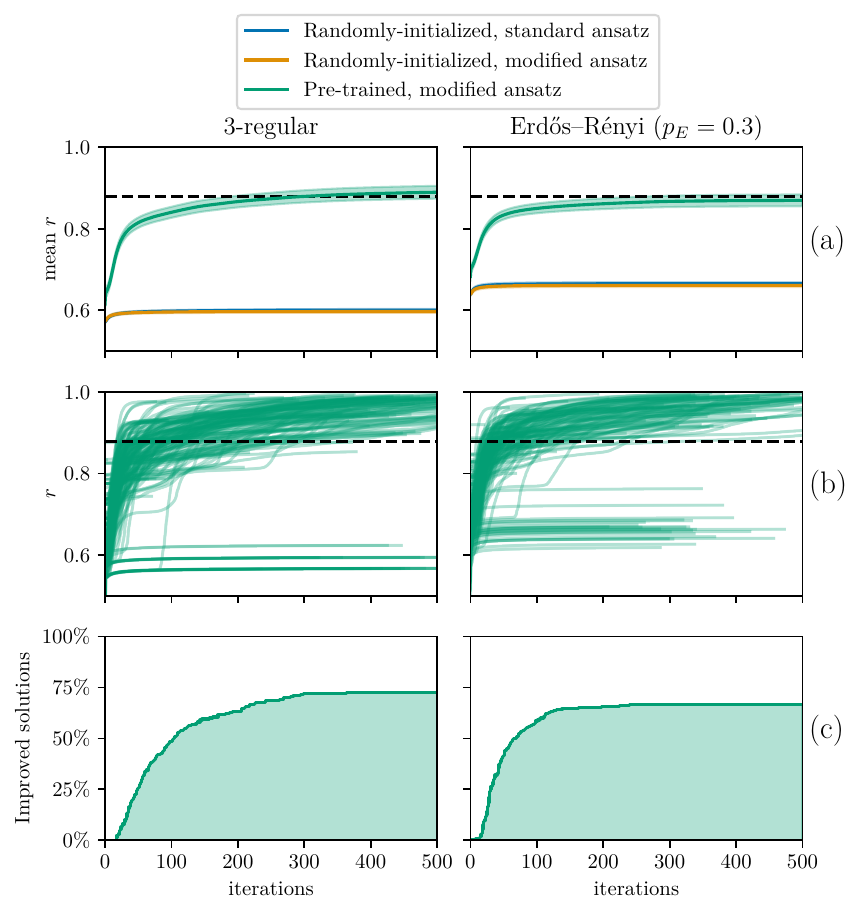}
    \caption{\textbf{Pre-training QAOA using $\mathfrak{g}$-sim.} Performances of QAOA on MaxCut at $n=16$ qubits for the standard randomly-initialized QAOA ansatz (blue), the modified randomly-initialized ansatz (orange), and the modified ansatz pre-trained with $\mathfrak{g}$-sim (green). Comparisons are made across ensembles of 200 random 3-regular graphs (left panels) and 200 Erd\"{o}s-R\'{e}yni graphs with edge probability $p_E=0.3$ (right panels). We compare \textbf{(a)} the mean approximation ratios $r$ (shaded bars are bootstrapped 95\% confidence intervals), \textbf{(b)} individual training trajectories, and \textbf{(c)} the fraction of solutions outperforming the GW threshold (reported as an horizontal dashed line). 
    }
    \label{fig:qaoa}
\end{figure}
We study the benefits of this pre-training at $n=16$ qubits on random 3-regular graphs and Erd\"{o}s-R\'{e}nyi graphs with edge probability $p_E=0.3$.
Results are reported in Fig.~\ref{fig:qaoa}, showing that pre-training significantly outperforms randomly-initialized  QAOA circuits.

In terms of the mean approximation ratio $r$ (Fig.~\ref{fig:qaoa}(a)), both randomly-initialized strategies vastly under-perform the GW threshold $r_{\rm GW}$ (horizontal dashed line), while the pre-training strategy achieves comparable average performance. More strikingly, when looking at details of the approximation ratios $r$ of the individual graph problems (Fig.~\ref{fig:qaoa}(b)), we find a majority of individual graphs achieve an approximation ratio $r > r_{\rm GW}$, with the rest at par with random initialization.

Overall, the fraction of solutions that improve on the GW threshold (Fig.~\ref{fig:qaoa}(c)) converges to $72.5\%$ for 3-regular graphs, and $66.5\%$ for $p_E=0.3$ Erd\"{o}s-R\'{e}nyi graphs. On the other hand, no randomly-initialized circuits achieved better than this threshold. This indicates that $\mathfrak{g}$-sim pre-training is advantageous for QAOA on a substantial fraction of random graphs, even for a relatively crude initialization strategy. We note that even in the case of pre-training, convergence to the improved solutions often requires performing a couple of hundreds of optimization steps. This remains challenging in current quantum devices, and even more so when accounting for noise. 
Nonetheless starting closer to the solution is always a desirable feature, and further improvements could likely be achieved by more carefully aligning the path graph $P_n$ along $G$, adaptively transforming the ansatz and cost, or resorting to another polynomially-size algebra for the pre-training.

\subsection{Circuit synthesis}
\label{sec:compilation}
Until now, we have been concerned with tasks of state-preparation. We now address more difficult problems of \emph{unitary compilation}.
Here, we seek to use $\gsim$ to identify the parameters of an ansatz circuit $U(\bm{\theta})\in\mathcal{G}$ of the form Eq.~\eqref{eqn:PeriodicStructureAnsatz} to implement a target unitary $V\in\mathcal{G}$. Despite the fact that both the target and circuit must belong to a unitary Lie group whose associated Lie algebra is of   polynomial dimension, such strategy can already enlarge the reach of $\gsim$. In particular, previous uses of $\gsim$  (as reported in Results \ref{thm:basic_evolution} \& \ref{thm:second_order_correlator_evolution}) were restricted to certain initial states and observables, but our circuit synthesis compilation goes beyond these cases.  Now, $\gsim$ can be used to classically compile a polynomial-gate-count circuit, and then implement it on a quantum computer to evaluate the evolution of \emph{any} observable or state. More generally, the unitary to be synthesized could be part of a larger protocol such that the initial state or observable may not even be known beforehand.
Finally, we note that there may be utility in compiling random unitaries corresponding to polynomially sized Lie algebras, as sampling from the output of some of these circuits have strong hardness guarantees that can be used to demonstrate quantum advantage~\cite{oszmaniec2022fermion}.

The scheme proposed is presented in Sec.~\ref{sec:compilation_basis}. We probe its convergence properties for random targets in Sec.~\ref{sec:compilation_random} displaying polynomial optimization effort in most cases. 
However, as documented in Sec.~\ref{sec:compilation_faithfulness}, given that we work in a reduced representation of the algebra, issues of faithfulness can arise and compromise compilation. Still, we provide a strategy to overcome this issue. Resorting to this strategy we demonstrate exact compilations with $\gsim$ in a task of dynamical evolution in Sec.~\ref{sec:compiling_dynamics}.

\subsubsection{Variational compilation of unitaries in $\gsim$}\label{sec:compilation_basis}
In variational circuit compilation, we typically seek to train the parameters of an ansatz circuit $U(\bm{\theta})$ to approximate some target unitary $V$. Existing techniques for compiling variational circuits to target unitaries usually seek to minimize the cost
\begin{equation}
\mathcal{L}_{\text{HST}}(U(\bm{\theta}),V)\equiv 1-\frac{1}{d^2}|\Tr(U^\dagger(\bm{\theta}) V)|^2,\label{eqn:hilbertschmidttest}
\end{equation}
that has computational complexity growing quadratically with $d=2^n$ rendering it quickly classically intractable, and thus would require evaluation on a quantum computer for even modest system sizes. 
In this situation, evaluation of Eq.~\eqref{eqn:hilbertschmidttest} could rely on the Hilbert-Schmidt test \cite{khatri2019quantum,sharma2019noise} or on estimation via state sampling~\cite{caro2022outofdistribution,gibbs2022dynamical}. 
In any case, these assume implementation of the target $V$ in the first place, which is often unrealistic.

In contrast, provided a description of $V\in\mathcal{G}$ with dimension of the associated Lie algebra $\dimg\in\OC(\operatorname{poly}(n))$, the present approach to compilation can be performed entirely on a classical computer with $\OC(\operatorname{poly}(n))$ resources. To that intent, we propose the loss function
\begin{align}
    \mathcal{L}_{\mathfrak{g}}(U(\bm{\theta}),V)&\equiv \frac{1}{2\dimg}\|\bar{\bm{U}}(\bm{\theta})-\bar{\bm{V}}\|^2_{\operatorname{HS}}\nonumber\\
    &=1-\frac{1}{\dimg}\operatorname{Re}\left[\Tr(\bar{\bm{U}}^T(\bm{\theta})\bar{\bm{V}})\right],\label{eqn:adjointspaceunitarycostfunc}
\end{align}
where $\|\cdot\|_{\operatorname{HS}}$ is again the Hilbert-Schmidt norm and $\bar{\bm{U}}(\bm{\theta})\equiv \Phi_{\lambda}^{\operatorname{Ad}}(U(\bm\theta))$, $\bar{\bm{V}}\equiv \Phi_{\lambda}^{\operatorname{Ad}}(V)$ are the adjoint representations of $U(\bm{\theta})$ and $V$, respectively. Gradients of $\mathcal{L}_{\mathfrak{g}}$ can be calculated efficiently, with implementation details provided in Appendix~\ref{appendix:compilation_gradients}.

In effect, $\mathcal{L}_{\mathfrak{g}}$ measures how accurately $U(\bm{\theta})$ approximates the evolution of the expectation values $\langle G_\alpha\rangle$ under evolution by $V$ for all initial states. Although the Hilbert-Schmidt loss $\mathcal{L}_{\text{HST}}$, in Eq.~\eqref{eqn:hilbertschmidttest}, and its adjoint-space version $\mathcal{L}_{\mathfrak{g}}$, in Eq.~\eqref{eqn:adjointspaceunitarycostfunc}, are similar in form, we highlight that $\mathcal{L}_{\mathfrak{g}}$ is sensitive to global phase differences between $\bar{\bm{U}}(\bm{\theta})$ and $\bar{\bm{V}}$, whereas $\mathcal{L}_{\text{HST}}$ is not sensitive to global phase differences between $U(\bm{\theta})$ and $V$. This is because global phase differences in the full Hilbert space are nonphysical, whereas global phase differences in the adjoint representation correspond to a sign flip on all basis observables, which is indeed physical. As we shall see (Sec.~\ref{sec:compilation_faithfulness}) related subtleties may compromise our ability to perform faithful compilation. For now we leave these aside, and assess convergence of optimizations with the loss in Eq.~\eqref{eqn:adjointspaceunitarycostfunc} for random target unitaries.

\subsubsection{Compiling random unitaries}
\label{sec:compilation_random}
We first benchmark our compilation scheme against random unitaries in $\mathcal{G}_{0}=e^{\g_0}$, which is the most general and demanding task for this scheme. Any unitary in $\mathcal{G}$ may be written in the form
\begin{equation}
    V=e^{-iT\sum_{\alpha}^{}(\vec{w})_\alpha G_{\alpha}} \,,
    \label{eqn:randomunitaries}
\end{equation}
where we fix $|\vec{w}|_2=1$ such that $T$ parametrizes the effective duration of the corresponding dynamics. To generate the weights we sample a matrix from a Haar distribution over the orthogonal group and set the weight vector $\vec{w}$ equal to one of its columns (potentially padding with zeros to account for elements of the basis not present in the sum).

\begin{figure}
    \centering
    \includegraphics[width=0.48\textwidth]{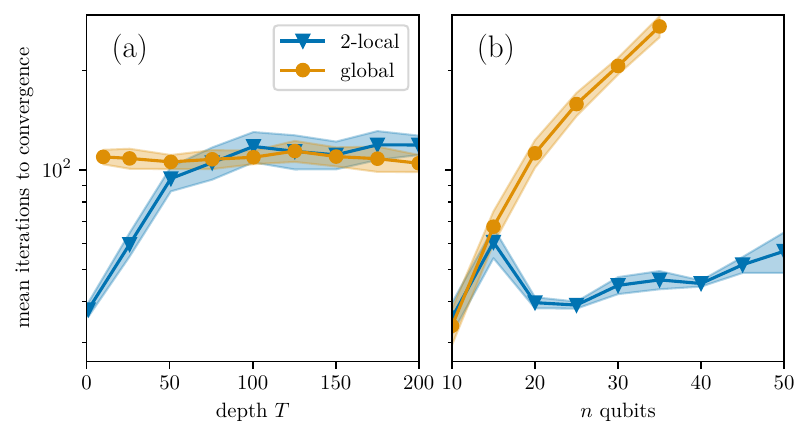}
    \caption{\textbf{Scaling properties of circuit compilation with $\mathfrak{g}$-sim.} We compare the mean number of iterations to converge to a loss (defined in Eq.~\eqref{eqn:adjointspaceunitarycostfunc}) $\mathcal{L}_{\mathfrak{g}}<10^{-3}$ for random target unitaries (defined in Eq.~\eqref{eqn:randomunitaries}) with varied levels of locality, including 2-local and global generators. Results for \textbf{(a)} varying duration $T$ at $n=20$ qubits, and  for \textbf{(b)} varying system size $n$ at duration $T=10$. The  
    Shaded bars depict bootstrapped 95\% confidence intervals.    The scaling of the iterations number is found to be polynomial in $n$ and $T$. 
    }
    \label{fig:unitary_training}
\end{figure}

In general, Eq.~\eqref{eqn:randomunitaries} involves highly non-local interactions, since many elements of $\mathfrak{g}_{0}$ are non-local Pauli operators (see Eq.~\eqref{eq:algebra-g0}).  
By excluding greater-than-$k$-local $G_\alpha$ from Eq.~\eqref{eqn:randomunitaries}, we can refine our study to $k$-local Hamiltonian dynamics. 

We test our scheme by compiling global and 2-local dynamics, both with the 2-local ansatz of Fig.~\ref{fig:benchmarks}(a), (see Appendix~\ref{appendix:ansatz} for more details), using $L=\left\lceil{2\dimg/K}\right\rceil$ layers, with $K$ the number of generators, ensuring overparametrization. In Fig. \ref{fig:unitary_training}, we show that the compilation of random unitaries $V \in \mathcal{G}_{0}$ performs and scales well. At fixed $n=20$ qubits (Fig. \ref{fig:unitary_training}(a)), the convergence requirements (measured as the number of optimization iterations required for convergence) plateau to a constant value irrespective of $T$, indicating that our approach allows fixed-circuit-depth Hamiltonian simulation for arbitrary evolution time. We note that this phenomenon was demonstrated in Ref.~\cite{kokcu2021fixed} on related systems, but with different methods and only targeting local dynamics.

\begin{figure*}
    \centering
    \includegraphics[width=0.8\textwidth]{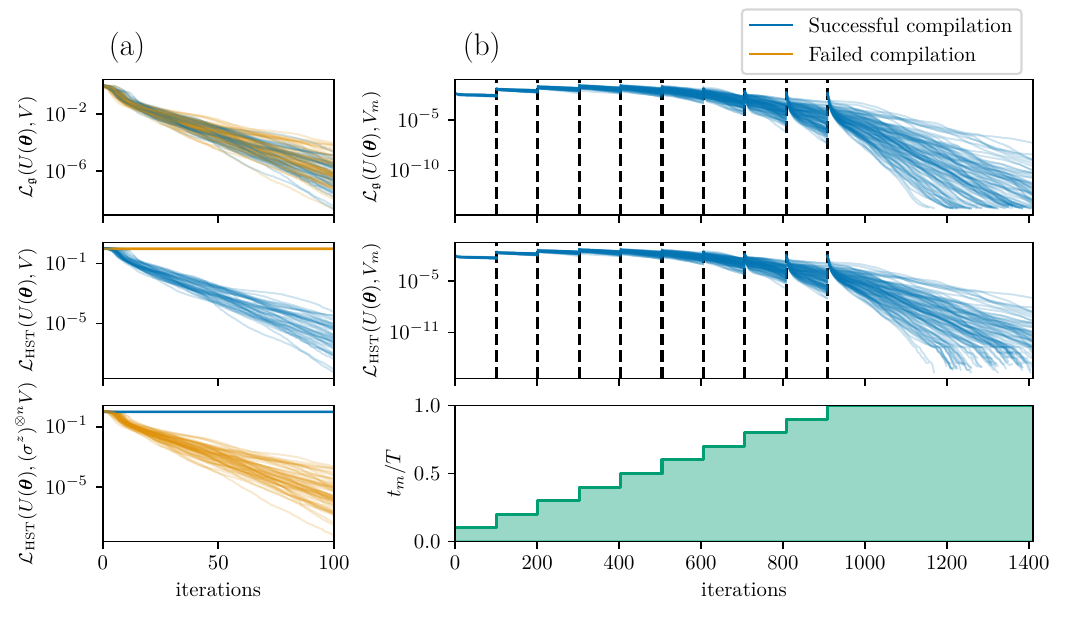}
    \caption{\textbf{Successfully compiling circuits in an algebra with non-trivial group center elements.} \textbf{(a)} Although the classical loss function $\mathcal{L}_{\mathfrak{g}}(U(\bm{\theta}),V)$ can be efficiently minimized (Fig. \ref{fig:unitary_training}), this is a necessary but insufficient condition to minimize the Hilbert-Schmidt loss function $\mathcal{L}_{\text{HST}}$, and thus to faithfully compile a target unitary $V$. We train a linear-depth ansatz $U(\bm{\theta})$ (see Fig.~\ref{fig:benchmarks}(a)) initialized with uniform random parameters to minimize $\mathcal{L}_{\mathfrak{g}}(U(\bm{\theta}),V)$ at $T=1$ and $n=10$. Although $\mathcal{L}_{\mathfrak{g}}$ is successfully minimized in all instances (top panel),  $\mathcal{L}_{\text{HST}}$ is only minimized in approximately half the instances, and rather \textit{maximized} otherwise (middle panel). 
    In all cases where the compilation fails, the optimized ansatz instead approximates $(\sigma^z)^{\otimes n}V$ (bottom panel). \textbf{(b)} The scheme outlined in Sec.~\ref{sec:compilation_faithfulness} successfully rectifies this issue. The ansatz is initialized at $\bm{\theta}=\bm{0}$, and compiled to a sequence of intermediate targets $V_m$ corresponding to increased $t_m$ from 0 to T (bottom panel), with the parameters obtained from compiling a target used to initialize the next . For sufficiently small time steps $\Delta t_m$, successfully minimizing $\mathcal{L}_{\mathfrak{g}}$ (top panel) ensures minimization of $\mathcal{L}_{\text{HST}}$ (middle panel) and thus faithful compilation of the final target $V$. }
    \label{fig:unitary_distinguishing}
\end{figure*}

At fixed $T$ (Fig. \ref{fig:unitary_training}(b)), the number of steps to converge appears constant in $n$ for 2-local targets, and polynomial in $n$ for global ones. The case of global targets with small duration $T$ is detailed further in Appendix~\ref{appendix:shallow_compilation}. Overall, this demonstrates that our methods can efficiently minimize the loss $\mathcal{L}_{\mathfrak{g}}$ on an overparametrized ansatz for a range of random unitaries in $\mathcal{G}_{0}$ with varying locality, system size $n$, and dynamics duration $T$, providing a strong foundation for our compilation scheme.

\subsubsection{Faithfulness of compilation}
\label{sec:compilation_faithfulness}
Although we saw consistent convergence with respect to $\mathcal{L}_{\mathfrak{g}}$ (Fig.  \ref{fig:unitary_training}), one should question whether minimizing $\mathcal{L}_{\mathfrak{g}}$ is sufficient for unitary compilation. It can be seen from Eq.~\eqref{eqn:adjointspaceunitarycostfunc} that $\mathcal{L}_{\mathfrak{g}}$ is a faithful loss function for unitary training iff $\Phi_{\lambda}^{\operatorname{Ad}}$ is a faithful representation. 

As discussed in Sec.~\ref{sec:lie}, faithful representation $\Phi^{\operatorname{ad}}_{\mathfrak{g}}$ of the Lie algebra does not guarantee faithful representation $\Phi_{\lambda}^{\operatorname{Ad}}$ of the Lie group. 
In particular, we have already seen in Eq.~\eqref{eq:center-adjoint} that for unitaries $W\in Z(\G)$ (i.e., for unitaries in the center of the group), it is the case that $\Phi_{\lambda}^{\operatorname{Ad}}(W)=I$ such that $\Phi_{\lambda}^{\operatorname{Ad}}(W V)=\Phi_{\lambda}^{\operatorname{Ad}}(V)$. That is, in the adjoint representation one cannot distinguish $V$ from $WV$.
For the case of $\mathfrak{g}_{0}$, the center is  $Z(\G)= \{(\sigma^z)^{\otimes n}, I\}$ up to a global phase (as detailed in Appendix~\ref{appendix:unfaithfulness}).

Such an issue can indeed be seen in our numerics. In the first row of Fig.~\ref{fig:unitary_distinguishing}(a) we report systematic success in minimizing $\mathcal{L}_{\mathfrak{g}}$. 
However, as can be distinguished by the loss $\mathcal{L}_{\text{HST}}$, only half of the optimizations yields the correct target $V$ while the other half rather yields the unitary $(\sigma^z)^{\otimes n}V$ (second and third row). 
To guarantee successful compilation of $V$, one must either ensure convergence to the manifold of correct solutions, or be able to flag when an error has occurred such that it could be corrected (i.e., by applying an additional unitary $(\sigma^z)^{\otimes n}$). We now detail a strategy achieving the former.

Any unitary target $V\in\mathcal{G}$ may be written in the form $V=e^{-iTH}$ with $iH\in\mathfrak{g}$ as in Eq.~\eqref{eqn:randomunitaries}. Rather than directly aiming for the compilation of $V$, we consider a family of intermediary targets $V_m=e^{-iHt_m}$ for increasing steps $t_m\in[0,T]$, with $t_0=0$. We then solve the corresponding compilation problems  $\bm{\theta}_m=\argmin_{\bm{\theta}}\mathcal{L}_{\mathfrak{g}}(U(\bm{\theta}),V_m)$ sequentially, with $\bm{\theta}$ initialized to $\bm{\theta}_{m-1}$ at each step.
For sufficiently small $\Delta t_m = t_m - t_{m-1}$, it is expected that no jump between solution manifolds will occur. 
Crucially, given that $V_0=I$, we can ensure that we start in the correct manifold by setting $\bm{\theta}_0=\bm{0}$ such that $U(\bm{\theta}_0)=I$ (rather than $(\sigma^z)^{\otimes n}$) . 

Viability of the proposed scheme is confirmed numerically and reported in Fig. \ref{fig:unitary_distinguishing}(b). For all the random unitaries $V$ assessed, the optimization concludes with parameters replicating accurately the desired target, as evidenced by the low values of $\mathcal{L}_{\text{HST}}(U(\bm{\theta}),V)<10^{-6}$. This demonstrates that faithful compilation is possible even when the cost function $\mathcal{L}_{\g}$ is not faithful.

\subsubsection{Application to dynamical simulation}
\label{sec:compiling_dynamics}
As noted in Sec.~\ref{sec:compilation_random}, given a Hamiltonian supported by a Lie algebra $\mathfrak{g}$ with $\dimg\in\OC(\operatorname{poly}(n))$, our scheme enables compilation in polynomial-depth circuits of the corresponding time-evolution operators. This allows the study of dynamics of observables \textit{not} supported by $\gsim$ and arbitrary states, which are in general not classically tractable. 
Here we demonstrate the utility of our scheme by synthesizing circuits for the Hamiltonian time-evolution of the TFXY spin chain of Eq.~\eqref{eqn:hamiltonian_tfxy_randomfields} with open boundary conditions and random local magnetic fields $b_j$ drawn from $N(0,\xi^2)$. 
This Hamiltonian supports the phenomenon of Anderson localization~\cite{anderson1958absence}, and has previously been utilized to demonstrate related compilation techniques based on Cartan decomposition \cite{kokcu2021fixed}.

We aim to train an ansatz $U(\bm{\theta}_{t_m})$ to approximate $V_{t_m}\equiv e^{-i{t_m}H_{\text{TFXY}}}$ at a range of discrete times $t_m\in[0,200]$. For each $t_m$, the ansatz has a structure
\begin{align}
    U(\bm{\theta})=\prod_{l=1}^{L} \Big[ & \prod_{j=1}^{n-1}e^{-i\theta_{l,(n+2j+1)}\sigma^x_j\sigma^x_{j+1}}e^{-i\theta_{l,(n+2j)}\sigma^y_j\sigma^y_{j+1}} \nonumber\\ &\quad  \prod_{j=1}^{n}e^{-i\theta_{l,j}\sum_{j=1}^n \sigma^z_j}\Big] \,,
    \label{eqn:LTFIM_ansatz-2}
\end{align}
One can verify that the dynamical Lie algebra associated with this circuit is again $\g_0$ as defined in Eq.~\eqref{eq:algebra-g0}, meaning that we can again utilize the representation elements already computed. 
At this point we note that while $U(\bm{\theta})$ has the exact same structure as that of a first-order Trotterization of any of the $V_{t_m}$ unitaries, each term of the Hamiltonian is associated with an independent trainable parameter in $U(\bm{\theta})$. This fact is important as $U(\bm{\theta})$ is not trying to learn a Trotterized version of $V_{t_m}$. In fact, because we use $\gsim$, we do not need to ever perform a Trotterization of the target unitary, as $V_{t_m}$ can be compiled exactly for all evolution times $t_m$. This is due to the fact that one can efficiently compute the adjoint representation $\Phi_{\lambda}^{\operatorname{Ad}}(V_{t_m})=e^{-it_m\Phi^{\operatorname{ad}}_{\mathfrak{g}}(H_{\text{TFXY}})}$ on a classical computer. This is advantageous compared to variational compilation schemes for time evolution that necessitate the target to be implementable on a quantum device in the first place, which is often achieved by means of a Trotter approximation~\cite{cirstoiu2020variational,gibbs2022dynamical,lin2021real,berthusen2022quantum,gibbs2022dynamical,goh2024direct}, and therefore introduces approximation errors.

\begin{figure}
    \centering
    \includegraphics[width=0.5\textwidth]{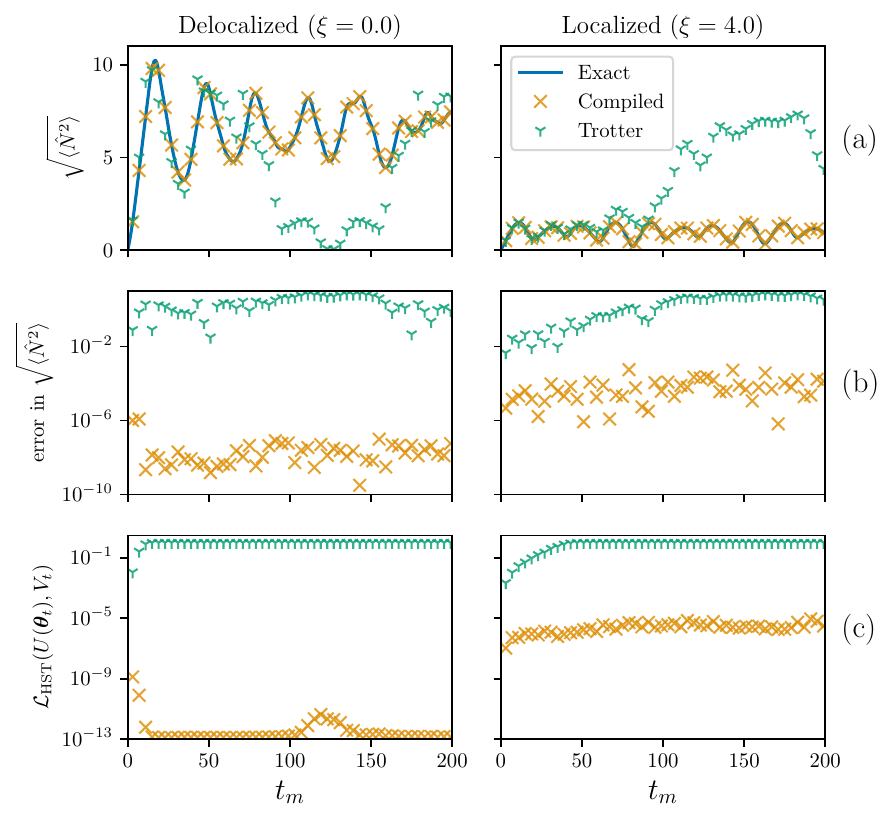}
    \caption{\textbf{Compiling Hamiltonian dynamics with $\mathfrak{g}$-sim.} Dynamics of a single-spin-flip excitation in a TFXY spin chain (Eq.~\eqref{eqn:hamiltonian_tfxy_randomfields} with $n=12$) in the absence of magnetic fields ($\xi=0$, left panels) and with random fields inducing Anderson localization ($\xi=4$, right panels). The exact dynamics (blue) is compared to the dynamics yielded by our compilation scheme (orange X), and first-order Trotterization (green Y). The compilation scheme accurately showcases Anderson localization, reproducing the RMS position $\sqrt{\hat{N}^2}$ of the excitation (a) with low errors (b), while Trotterization quickly diverges beyond small simulation times. Furthermore, comparison of the compiled unitaries to the exact time-evolution operators by means of the Hilbert-Schmidt test $\mathcal{L}_{\text{HST}}$, show small errors at all times (c), guaranteeing faithful dynamics of all observables.}
    \label{fig:anderson}
\end{figure}

For the sake of concreteness, we focus on the dynamics of a single-spin-flip initial state $\ket{\downarrow \uparrow \uparrow \uparrow \uparrow \uparrow \uparrow \uparrow \uparrow \uparrow \uparrow \uparrow }$ for $n=12$ qubits. 
In the absence of a magnetic field ($b_j=0$ for all $j$) this excitation diffuses throughout the system, but a disordered field ($\xi > 0$) restricts this diffusion (Anderson localization). Following the example of Ref.~\cite{kokcu2021fixed}, we study the position operator for the excitation
\begin{equation}
    \hat{N}=\sum_{j=1}^n (j-1)\frac{1-\sigma^z_j}{2} \,,
\end{equation}
whose moments provide a key signature of Anderson localization. In particular, for this system the $p$-th moment $\langle|\hat{N}|^p\rangle$ admits a time-independent upper bound \cite{kokcu2021fixed,bucaj2016kunz}. 
Note that $\hat{N}^2$ is expressible as a linear combination of elements in the algebra ($\sigma_j^z$) and product of elements in the algebra ($\sigma_j^z\sigma_{j'}^z$). Moreover, given that the algebra is composed of Pauli operators, and since the initial state is separable, we can readily compute the entries of the vector $\bm{e}^{(\operatorname{in})}$ and the matrix $\bm{E}^{(\operatorname{in})}$. Hence, the $p=2$ moment can be efficiently computed with  
$\gsim$. As noted above, larger moments quickly become intractable. 
Furthermore, we stress that our choice of a classically-tractable system size ($n=12$ qubits) is purely 
to enable us to compute the loss $\mathcal{L}_{\text{HST}}( U(\bm{\theta}_{t_m}),V_{t_m})$ for verification of correct unitary compilation.
In general, our scheme can achieve compilation at larger system sizes (e.g., $n=50$ qubits in Fig.~\ref{fig:unitary_training}).

We compile the time-evolution operators for 
the case of no magnetic field ($\xi=0$) and of a disordered magnetic field ($\xi=4$) with a Hamiltonian renormalized as $H_{\text{TFXY}}\to \frac{H_{\text{TFXY}}}{\|H_{\text{TFXY}}\|_{\operatorname{HS}_1}}$ to eliminate any norm dependence of the dynamics. 
Our ansatz is determined by Eq.~\eqref{eqn:LTFIM_ansatz-2} and is composed of $L=17$ layers, which is overparametrized to ensure trainability. 
In line with the strategy previously discussed, the circuit is first initialized with $\bm{\theta}_{t_0}=\bm{0}$, and thereafter initialized with $\bm{\theta}_{t_m}$ according to the parameters found in the ${t_{m-1}}$ step of optimization. 

In Fig.~\ref{fig:anderson}, we report a comparison between the compiled dynamics and a first-order Trotterization at the same depth (i.e., identical circuit structure with $\theta_{lk}=t/L$). We find that, despite using identical circuit structures, our scheme outperforms Trotterization by several orders of magnitude. 
Looking at the dynamics of the position operator $\sqrt{\hat{N}^2}$ (Fig.~\ref{fig:anderson}(a)) one can see that Trotterization fails to reproduce results beyond $t\approx60$ while our scheme continues to track them accurately at any of the times considered, with errors smaller than $10^{-6}$ (Fig.~\ref{fig:anderson}(b)). More generally, we find that the compiled unitaries reflect accurately the true time-evolution with $\mathcal{L}_{\operatorname{HST}}(U(\bm{\theta}_t),V_t) <10^{-5}$ (Fig.~\ref{fig:anderson}(c)), thus guaranteeing that the compiled circuits can faithfully reproduce the dynamics of \emph{all} observables, not just those supported by $\mathfrak{g}$ and products thereof. Furthermore, the errors entailed by the compiled circuit do not appear to vary substantially in the duration of the dynamics. Overall, this shows that for time-evolution operators whose associated Lie algebra $\mathfrak{g}$ has $\dimg \in \OC(\operatorname{poly}(n))$, our $\mathfrak{g}$-sim compilation schemes can determine a much more efficient circuit implementation than standard Trotterization.

\subsection{Supervised quantum machine learning}
\label{sec:supervised_QML}
As a final demonstration, we employ $\gsim$ for the training of a quantum-phase classifier, showcasing its applicability in the context of QML. 
Despite being trained fully classically on tractable states, such a classifier could then be employed to classify unknown quantum states.
Provided that meaningful training data points and quantum models can be found in algebras with polynomial dimension, such schemes could find utility in real experiments. In fact, this scenario of QML on quantum data is very similar to quantum metrology: protocols can be developed through classical simulations and still have merit when implemented through quantum technology. Additionally, in the spirit of Sec.~\ref{sec:pretraining}, our $\gsim$ approach could form the basis of approximate quantum models that are then refined on a quantum computer.

\subsubsection{Supervised QML}
In general problems of supervised QML one assumes repeated access to a training dataset $\mathcal{S}=\{(\rho_s,y_s)\}_{s=1}^N$ consisting of  of $N$ pairs of states $\rho_s$ together with labels $y_s=F(\rho_s)$ that have been assigned by an unknown underlying function $F$. 
The task is to learn parameters $\bm{\theta}$ of a function  $h_{\bm{\theta}}$ aiming at approximating $F$ as accurately as possible. 
Upon successful training, it is then possible to accurately predict the labels of previously unseen states.

As typical in tasks of QML, we consider a model $h_{\bm{\theta}}$ that relies on the expectation value of an observable $O$ after application of a circuit $U(\bm{\theta})$. That is, on $\ell_{\bm{\theta}}(\rho_s)=\Tr[O U(\bm{\theta})\rho_s U^\dagger(\bm{\theta})]$. 
Training relies upon minimization of a mean-squared error loss function, defined here as 
\begin{equation}
    \mathcal{L}(\bm{\theta})=\frac{1}{N}\sum_{s=1}^N(y_s-\ell_{\bm{\theta}}(\rho_s))^2\,.
    \label{eqn:QML_loss_function}
\end{equation}

\subsubsection{Training a binary quantum-phase classifier}
For our numerical study, we apply the $\gsim$ framework to the classification of ground states across a phase transition of the TFIM, in  Eq.~\eqref{eqn:TFIM_hamiltonian}, at system size $n=50$ qubits. We consider parameters $h_z,h_{xx}\in[0,1]$ in order to focus on a single phase transition (binary classification). For $h_z/h_{xx}>1$, the ground state is in the disordered phase, while for $h_z/h_{xx}<1$ it is in the antiferromagnetic phase. Furthermore, to ensure that the problem is non-trivial, we `disguise' the TFIM according to a random $V\in \mathcal{G}$ (Eq.~\eqref{eqn:randomunitaries} with $T=10$) such that
\begin{equation}
    H_{\text{disguised}}=VH_{\text{TFIM}}V^\dagger.
    \label{eqn:disguised_hamiltonian}
\end{equation}

To generate the training and test datasets, we start by randomly sampling values of $h_z$ and $h_{xx}$ and assign labels $y_s=-1$ ($+1$) to Hamiltonian parameters corresponding to the disordered (antiferromagnetic) phase. 
Each set of parameters corresponds to a distinct instance of $H_{\text{TFIM}}$ and, for each of them, we take the corresponding state $\rho_s$ to be the (approximate) ground state of this Hamiltonian instance. A $\gsim$ classical representation of $\rho_s$ consists of the vector of expectation values $(\bm{e}^{(s)})_\alpha=\Tr[\rho_s G_\alpha]$. Since exact diagonalization is intractable for $n=50$ qubits, to compute these, we resort to VQE in $\gsim$ with an overparametrized ansatz to obtain a circuit $U(\bm{\theta})$ such that $\rho_s\approx U(\bm{\theta}) \ket{0}\bra{0}^{\otimes n} U^\dagger(\bm{\theta})$. We then compute the classical representation $\bm{e}^{(s)}$ by applying Eq.~\eqref{eqn:gsim_ansatz_evolution} to a classical representation of the computational zero state. Finally, we transform this classical representation of the ground state of $H_{\text{TFIM}}$ to the corresponding `disguised' ground state of $H_{\text{disguised}}$ by applying $V$. Such procedure is repeated to generate a training dataset of 200 labeled states, equally divided in the two phases. To evaluate the performances of the classifier, we apply the same procedure to create a test dataset with the same number of states but that have not been seen by the optimizer during training.

\begin{figure}
    \centering
    \includegraphics[width=0.5\textwidth]{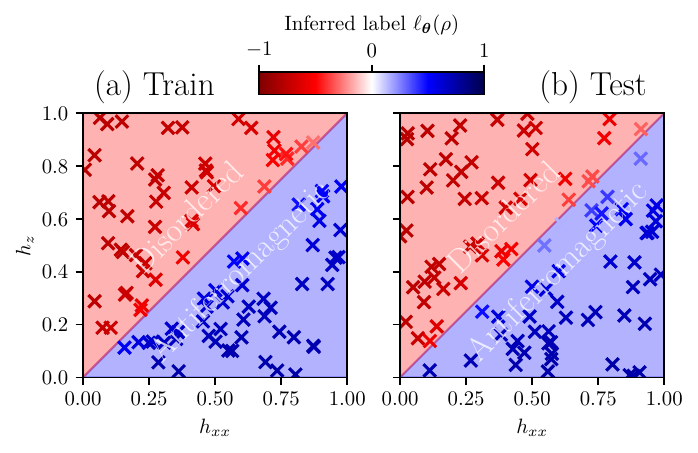}
    \caption{\textbf{Supervised QML with a classically-efficient circuit.} \textbf{(a)} Label assignments from a 50-qubit classifier on a training dataset of disguised TFIM ground states for $T=10$ (details in the main text), across the disordered-antiferromagnetic phase transition. The classifier is trained to minimize the loss function in Eq.~\eqref{eqn:QML_loss_function}. \textbf{(b)} Inference on a new set of data drawn from the same distribution. The model achieves 100\% classification accuracy on the test dataset. 
    }
    
    \label{fig:classifier}
\end{figure}

The quantum classifier is realized as 21 layers of the 2-local ansatz in Eq.~\eqref{eqn:LTFIM_ansatz-2} (see also  Appendix~\ref{appendix:ansatz}) and a measurement observable $O=\sigma^z_1\in i\g_0$. We train it to minimize the loss function in Eq.~\eqref{eqn:QML_loss_function} over the training dataset. Once the parameters trained, we assess accuracy of the classifier on the test dataset. 
Given trained parameters $\bm{\theta}^*$, to perform inference we assign labels as
\begin{equation}
    h_{\bm{\theta}^*}(\rho_s)=\operatorname{sgn}(\ell_{\bm{\theta}^*}(\rho_s))\in\mathcal{Y} \,.
    \label{eqn:assign_label}
\end{equation}
Results are depicted in Fig.~\ref{fig:classifier}(a), showing that for this problem the classifier achieves 100\% classification accuracy on the new data. Consistent classification accuracy is also verified across multiple random instances of the unitary disguise $V$, thus demonstrating successful application of $\gsim$ to a supervised QML problem. 

We note that in Appendix~\ref{appendix:implementing_classifier}, we discuss how the $\gsim$ phase classifier can be efficiently implemented on a quantum device.

\section{Conclusions}
Efficient classical simulations of quantum circuits are a valuable tool in scaling towards quantum advantage. In this work, we have presented $\mathfrak{g}$-sim, a classical simulation and optimization framework that relies on the Lie-algebraic structure of the circuits. Reformulating existing results on Lie-algebraic simulation \cite{somma2005quantum,somma2006efficient} into a modern presentation aimed at the quantum computing community, we further extended the scope of such simulations to a broader set of initial states, observables, or noise, and improved on the efficiency of their implementations. Moreover by comparing the scope of applicability of $\gsim$ to that of Wick-based simulations, we argue that our proposed methods enable new regimes for classical simulability that would be otherwise intractable. 
Of course, such classical simulations are only possible in restricted situations. Conditions allowing for such scalability with $\gsim$ were laid out, and demonstrations, with system sizes of up to $n=200$ qubits, highlighting distinctions compared to other classical simulation techniques were provided.

By introducing circuit optimization to this framework, we enabled the scalable classical study of paradigmatic examples relevant for variational quantum computing, demonstrating utility in studying scaling behaviors of VQA problems, improving trainability and mitigating barren plateaus via classical pre-training, and implementing supervised QML problems. These results expand the growing insights in classical pre-training of VQAs \cite{grant2019initialization,verdon2019learning,sauvage2021flip,rad2022surviving,liu2022mitigating,mitarai2022quadratic,ravi2022cafqa,cheng2022clifford,dborin2022matrix,mele2022avoiding,rudolph2022synergy}, Lie-algebraic study of trainability \cite{larocca2021diagnosing,larocca2021theory}, and the usefulness of classically simulable variational quantum algorithms~\cite{cerezo2023does,angrisani2023learning,bermejo2024quantum,lerch2024efficient,schreiber2023classical,jerbi2023shadows,shao2023simulating,basheer2023alternating,shaffer2023surrogate,anschuetz2022efficient,mele2024efficient}.

Furthermore, we expanded the framework of $\gsim$ to include compilation and circuit synthesis, which is a novel direction for simulation schemes of this type. By constructing and optimizing a circuit fidelity cost function that can be computed entirely in $\gsim$, we demonstrated that one can synthesize linear-depth circuits for unitary transformations in $e^{\g_0}$, where $\mathfrak{g}_0$ \eqref{eqn:g_tfxy} is the algebra used for simulations throughout this work. Compiling such rotations has already proven to be of great utility to the community, with e.g., the use of Givens rotation decomposition to prepare Hartree-Fock states with linear circuit depth \cite{kivlichan2018quantum,arute2020hartree}.
These can be used, e.g., as initial states for more refined circuit preparation~\cite{mizukami2020orbital}. However, our scheme does not necessarily require an underlying free-fermion structure, and therefore expands the class of compilable transformations to include \emph{any} system corresponding to a $\operatorname{poly}(n)$-dimensional algebra.

To date, Lie-algebraic considerations have been pivotal in many topics of quantum science including quantum error correction~\cite{zanardi1997noiseless,eastin2009restrictions}, controllability of quantum systems~\cite{khaneja2001time,dalessandro2010introduction,zeier2011symmetry,zimboras2015symmetry,marvian2022restrictions}, efficient measurements~\cite{yen2021cartan}, dynamical simulations~\cite{steckmann2021simulating,kokcu2022fixed}, and studying trainability properties of parametrized quantum circuits~\cite{larocca2021diagnosing,larocca2021theory}. Our work has demonstrated further utility in existing applications (dynamical simulations and studying trainability properties), and expanded this list to include classical pre-training of VQAs, efficient circuit compilation, and supervised QML. We anticipate that the $\gsim$ framework will provide helpful new perspectives and tools in the development of variational quantum computing.

\subsection{Future work}
In this work, we have explored a variety of applications of $\mathfrak{g}$-sim. Nonetheless, several applications beyond the scope of this work are apparent.

One of such applications is quantum error mitigation (QEM) \cite{cai2022quantum}, a class of techniques which seek to minimize noise-induced biases in quantum algorithm outputs without resorting to full-fledged quantum error correction. 
In particular, learning-based QEM methods \cite{czarnik2020error,strikis2020learning,google2020observation,montanaro2021error} require access to pairs of noisy circuit outputs (obtained from quantum hardware) and corresponding noiseless outputs (obtained from classically-efficient simulation techniques) in order to learn a function mapping noisy outputs to their correct noiseless values.
In this context, $\gsim$ could extend current classical simulation techniques that have been employed.

In a similar vein, $\mathfrak{g}$-sim has potential applications in the context of randomized benchmarking~\cite{knill2008randomized}. The greatest limitation in quantum computing is the effect of hardware errors in computations, both coherent and incoherent. There is thus a pressing need to comprehensively characterize the type and magnitude of errors present on quantum hardware. One approach to this is randomized benchmarking, which compares outputs of random gate sequences of increasing length to known expected results without resorting to standard process tomography. Recent work~\cite{helsen2022matchgate} has investigated the use of matchgate circuits for randomized benchmarking. Matchgates are closely related to $\mathfrak{g}_0$, and the underlying simulation schemes have much in common with $\mathfrak{g}$-sim. Much like the case of error mitigation, $\mathfrak{g}$-sim could be used to generate benchmarking data for non-matchgate circuits, expanding the scope of these techniques.

We remark that $\gsim$ has already played a central role in the study of classical simulability of certain variational quantum algorithms~\cite{cerezo2023does,angrisani2023learning,bermejo2024quantum,lerch2024efficient,schreiber2023classical,jerbi2023shadows,shao2023simulating,basheer2023alternating,shaffer2023surrogate,anschuetz2022efficient,mele2024efficient}. Here, it has been noted that while $\gsim$ can emulate the information processing capabilities of parametrized quantum circuits, its use still requires a quantum computer to obtain the components of the input state onto the irrep. The latter could necessitates a quantum computer for an initial data-acquisition phase, thus indicating that the whole algorithm is not fully classical end-to-end. It remains an open question for future research whether such ``quantum-enhanced'' simulations~\cite{cerezo2023does} can achieve a quantum advantage (if only polynomial), or whether $\gsim$ or other classical simulation techniques can be used for effective pre-training.

We also note the opportunity to further improve the compilation scheme in Sec.~\ref{sec:compilation}. Although its performance appears competitive with the Cartan decomposition approach of Ref.~\cite{kokcu2021fixed}, one disadvantage is that our approach requires re-optimization at each time step, whereas the former requires only one optimization. We believe that this limitation could be overcome by using a variational fast-forwarding ansatz \cite{cirstoiu2020variational}, which should work since the diagonalization unitary must necessarily be in $\mathcal{G}$. Furthermore, our scheme could trivially be extended to the compilation of evolution unitaries for time-dependent Hamiltonians. In the case of Hamiltonians with periodic Floquet driving, it should even be possible to construct a fast-forwardable solution, which to the best of our knowledge would not be possible with Cartan decomposition methods.

While the previous discussion was based on noiseless simulations, noisy simulations also have  many applications. 
Understanding the limitations imposed by noise and ways to mitigate them is primordial, and we expect the noisy simulations capabilities through $\gsim$ to play an important role in such research. Already, several of the studies can readily be extended to a noisy setup. These include the study of over-parametrization of Sec.~\ref{sec:overparametrization} that could be performed under Pauli noise channels providing further insights into noise effects when optimizing quantum circuits~\cite{garcia2023effects,duschenes2024characterization}. 
The pre-initializion  of Sec.~\ref{sec:pretraining} and the QML model of Sec.~\ref{sec:supervised_QML} could be re-trained including such noise, leading to more robust circuits. 

More generally tasks of state preparation could benefit from the inclusion of noise, with in particular the preparation of Hartree-Fock states~\cite{kivlichan2018quantum,arute2020hartree} that could be optimized to reflect more realistic conditions. Going further, one could envision noise-aware compilation, whereby one would aim at compiling a target unitary such that, depending on details of the noise, different circuit realizations would be identified through optimization. Adequacy of the scalable cost function used for the noiseless case~\eqref{eqn:adjointspaceunitarycostfunc} will need to be assessed for this noisy scenario. We leave this aspect for future works. These are all straightforward examples of use of noisy simulations with $\gsim$, but more applications could be expected.

As discussed earlier and developed further in the Appendices, several generalizations of $\gsim$ are possible. 
In particular, while presented here in the context of digital quantum computations, the methodologies can be easily adapted to analog computing. 
This opens up the possibility of exploring similar applications in this domain, and we anticipate that  Lie-algebraic simulations will be valuable for further investigations of the capabilities and constraints of analog and emerging digital-analog quantum platforms~\cite{lamata2018digital,parra2020digital,meitei2020gate,henriet2020quantum,daley2022practical,wurtz2023aquila}.

Finally, we recall that by studying the limitations of $\gsim$ and Wick-based techniques, we uncovered a deep connection between these classical simulation methods and quantum resources theories. Importantly, our results hint at the existence of different, non-equivalent, types of resource which make each of those methods fails. The exploration of this line of thought could lead to further understanding of what makes a quantum process truly ``quantum'', and thus deepen our understanding on the limitations of classical simulation strategies.

\section*{Acknowledgements}
The authors wish to thank Adrian Chapman, Nahuel L. Diaz, Tyson Jones, B\'{a}lint Koczor and Arthur Rattew for helpful technical conversations, and further thank Adrian Chapman for comments on the manuscript. MLG acknowledges the Rhodes Trust for the support of a Rhodes Scholarship. MLG was also  supported by the U.S. DOE through a quantum computing program sponsored by the Los Alamos National Laboratory (LANL) Information Science \& Technology Institute. ML acknowledges support by the Center for Nonlinear Studies at LANL.   ML and MC were supported by the Laboratory Directed Research and Development (LDRD) program of LANL under project numbers 20230049DR and 20230527ECR. LC was supported by the U.S. Department of Energy, Office of Science, Office of Advanced Scientific Computing Research through the Accelerated Research in Quantum Computing Program MACH-Q Project. FS was supported by the Laboratory Directed Research and Development (LDRD) program of
Los Alamos National Laboratory (LANL) under project number 20220745ER. This work was also supported by LANL's ASC Beyond Moore’s Law project and by by the U.S. Department of Energy, Office of Science, Office of Advanced Scientific Computing Research, under Computational Partnerships program. The authors acknowledge the use of the University of Oxford Advanced Research Computing (ARC) facility \cite{oxfordARC} in carrying out this work.

\clearpage
\newpage
\widetext
\appendix

\section*{Appendices for ``Lie-algebraic classical simulations for variational quantum computing''}

Here we present additional details for the main results in our manuscript. In Appendix~\ref{appendix:proofs}, we present proofs supporting the ability to perform noiseless and noisy simulations through $\gsim$.
In Appendix~\ref{appendix:efficient_impl}, we describe several techniques that significantly improve the efficiency of our implementation of $\gsim$. Further expanding these improvements, in Appendix \ref{appendix:gradients} we present efficient methods for computing gradients to enable optimization of circuits and dynamics with respect to observable cost functions. Then, Appendix \ref{appendix:ansatz} contains the ansatz circuits used throughout the work and further details on the benchmarking procedure used to evaluate the performance of our $\gsim$ implementation. In Appendix \ref{appendix:algebra_and_subalgebras}, we note several algebras of interest closely related to $\g_0$. In Appendix \ref{appendix:further_details_compilation}, we outline some subtleties of unitary compilation with $\gsim$. Finally, in Appendix \ref{appendix:implementing_classifier} we comment on practical implementation of a $\gsim$-trained phase classifier on a quantum device.

\section{Theorems and proofs}
\label{appendix:proofs}
\subsection{Representations of unitary evolution}
\label{appendix:unit_evolve}
\begin{lemma}[Invariance of the algebra]
Suppose $U\in\mathcal{G}$ and $iG_\alpha \in \g$. Then defining $\tilde{G}_\alpha=U^\dagger G_\alpha U$ for any $G_\alpha$, we have $i\tilde{G}_\alpha\in\mathfrak{g}$. That is, $i \g$ is an invariant subspace as per Definition~\ref{def:invariant}.
\label{lemma:group_operations}
\end{lemma}
\begin{proof}
    Since $U\in\mathcal{G}$, there exists some $iH\in\mathfrak{g}$ such that $U=e^{-iH}$. The exponential mapping is defined by
    \begin{equation}
        e^{-iH}=1-iH+\frac{(iH)^2}{2!}-\frac{(iH)^3}{3!}+\dots \,,
    \end{equation}
    and it follows that
    \begin{equation}
        \tilde{G}_\alpha=G_\alpha+i\left[H,G_\alpha\right]-\frac{1}{2}\left[H,\left[H,G_\alpha\right]\right]+\dots \,.
        \label{eqn:group_operation_nested_commutators}
    \end{equation}
    Since all the terms of Eq.~\eqref{eqn:group_operation_nested_commutators} are generated by nested commutators only involving elements of $\mathfrak{g}$ (up to appropriate factors of $i$ for each term), by Definition~\ref{def:dynamical_lie_algebra} of the Lie algebra  we have that $i\tilde{G}_\alpha\in\mathfrak{g}$.
\end{proof}

\begin{lemma}[Adjoint representation of unitary evolution, adapted from Appendix C in Ref.~\cite{somma2005quantum}]
    Suppose that $U\in\mathcal{G}$, and that $\{B^{(\lambda)}_\alpha\}_{\alpha=1}^{\dim(\irrep)}$ forms a Schmidt-orthonormal basis of $\irrep$.  
    Defining $\tilde{B}^{(\lambda)}_\alpha=U^\dagger B^{(\lambda)}_\alpha U$, we have 
    \begin{equation}
    \tilde{B}^{(\lambda)}_\alpha=\sum_{\beta} u_{\alpha\beta} B^{(\lambda)}_\alpha \, ,
    \label{eqn:group_operation_defn}
    \end{equation}
    where the corresponding matrix elements $u_{\alpha\beta}$ are those of the adjoint representation of $U$,
    \begin{equation}
        u_{\alpha\beta}=\left(\Phi_{\lambda}^{\operatorname{Ad}}(U)\right)_{\alpha\beta}.
    \end{equation}
    \label{lemma:adjoint_representation_group_operations}
\end{lemma}
\begin{proof}
    Since $U\in\mathcal{G}$, there exists some $iH\in \mathfrak{g}$ such that $U=e^{-iH}$. Furthermore, $U$ may be decomposed as
    \begin{equation}
        U=\lim_{M\to\infty}\prod_{m=1}^{M} U^{(\Delta)},
        \label{eqn:infinitesimal_product}
    \end{equation}
    corresponding to infinitesimal steps $\Delta=1/M$ and infinitesimal unitaries $U^{(\Delta)}$ defined as
    \begin{equation}
        U^{(\Delta)} = e^{-i\Delta H} = 1-i\Delta H+\mathcal{O}(\Delta^2).
        \label{eqn:infinitesimal_unitary}
    \end{equation}
    Given the definition of the invariant space $\irrep$, and the fact that $U^{(\Delta)}\in \mathcal{G}$ we can always write 
    \begin{equation}
        {U^{(\Delta)}}^\dagger B^{(\lambda)}_\alpha U^{(\Delta)}=B^{(\lambda)}_\alpha - \sum_{\beta} v^{(\Delta)}_{\alpha\beta} B^{(\lambda)}_\beta.
        \label{eqn:infinitesimal_group_operation}
    \end{equation}
    Using Equations~\eqref{eqn:infinitesimal_unitary} and \eqref{eqn:infinitesimal_group_operation} we obtain
    \begin{equation}
        {U^{(\Delta)}}^\dagger B^{(\lambda)}_\alpha U^{(\Delta)} = B^{(\lambda)}_\alpha+i[\Delta H, B^{(\lambda)}_\alpha]+\mathcal{O}(\Delta^2)=B^{(\lambda)}_\alpha - \sum_{\beta} v^{(\Delta)}_{\alpha\beta} B^{(\lambda)}_\beta.
        \label{eqn:first_order_group_operation}
    \end{equation}
    We may discard vanishing terms in $\mathcal{O}(\Delta^2)$ since we are in the infinitesimal limit $M\to\infty$. Doing so, and taking $\Tr[B^{(\lambda)}_\gamma\dots]$ of both sides of Eq.~\eqref{eqn:first_order_group_operation} yields
    \begin{align}
        &i\Tr[B^{(\lambda)}_\gamma[\Delta H, B^{(\lambda)}_\alpha]]=-\sum_\beta v^{(\Delta)}_{\alpha\beta}\Tr[B^{(\lambda)}_\gamma B^{(\lambda)}_\beta] \\
        \implies& -\sum_k \Delta h_k \Tr[B^{(\lambda)}_\gamma,[G_k,B^{(\lambda)}_\alpha]] = i\Delta v^{(\Delta)}_{\alpha\gamma} \\
        \implies& v^{(\Delta)}_{\alpha\gamma}=i\Delta \sum_k h_k(\Phi^{\operatorname{ad}}_{\mathfrak{g}}(G_k))_{\alpha\gamma}=i\Delta \left(\Phi^{\operatorname{ad}}_{\mathfrak{g}}(H)\right)_{\alpha\gamma},
        \label{eqn:determining_adjoint_unitary}
    \end{align}
    where we have expanded $H=\sum_k h_k G_k$ for $h_k\in\mathbb{R}$ in terms of a Schmidt-orthonormal basis $\{G_k\}$ of $\g$ and used the Definition~\ref{defn:adjoint_representation} of the adjoint representation. Substituting the result of Eq.~\eqref{eqn:determining_adjoint_unitary} into Eq.~\eqref{eqn:infinitesimal_group_operation} reveals that the infinitesimal group operation is determined entirely by the adjoint representation $\Phi^{\operatorname{ad}}_{\mathfrak{g}}(H)$
    \begin{equation}
        {U^{(\Delta)}}^\dagger B^{(\lambda)}_\alpha U^{(\Delta)}=\sum_{\beta}\left(I-i\Delta\Phi^{\operatorname{ad}}_{\mathfrak{g}}(H)\right)_{\alpha\beta}B^{(\lambda)}_\beta=\sum_\beta\left(e^{-i\Delta\Phi^{\operatorname{ad}}_{\mathfrak{g}}(H)}\right)_{\alpha\beta}B^{(\lambda)}_\beta+\mathcal{O}(\Delta^2).
        \label{eqn:infinitesimal_operation_final_form}
    \end{equation}
    We may again discard terms in $\mathcal{O}(\Delta^2)$ since we are in the infinitesimal limit $M\to\infty$. Noting that the full unitary is obtained by composing infinitesimal steps (Eq.~\eqref{eqn:infinitesimal_product}), we obtain from Eq.~\eqref{eqn:infinitesimal_operation_final_form} that
    \begin{equation}
        \tilde{B}^{(\lambda)}_\alpha =  U^\dagger B^{(\lambda)}_\alpha U = \sum_\beta\left(e^{-i\Phi^{\operatorname{ad}}_{\mathfrak{g}}(H)}\right)_{\alpha\beta}B^{(\lambda)}_\beta=\sum_\beta\left(\Phi_{\lambda}^{\operatorname{Ad}}(U)\right)_{\alpha\beta}B^{(\lambda)}_\beta,
    \end{equation}
    which is exactly the desired result.
\end{proof}

\subsection{Evolution of product of observables}
\label{appendix:simul_product}
Here, we provide details about how product of observables are simulated with $\gsim$. When specializing to a product of either a single or two observables, we retrieve the setting of Results.~\ref{thm:basic_evolution} and~\ref{thm:second_order_correlator_evolution} of the main text, respectively. However, this encapsulates more general situations. 
In the following, we wish to evaluate the expectation value
\begin{equation}\label{eq:expectation_product}
    \langle O  \rangle := \Tr[O U \rho^{(\operatorname{in})} U^{\dagger} ], \quad \text{where}\; O = \prod_a O^{(a)} \quad \text{and}\; O^{(a)}=\sum_{\alpha}(\bm{w}^{(a)})_\alpha B^{(\lambda_a)}_\alpha.
\end{equation}
Each $O^{(a)}$ is supported by a single irrep $\mathcal{L}_{\lambda_a}$, and we recall that the unitary realized by the circuit is defined through
\begin{equation}
    U=\prod_{l=1}^L e^{-i\theta_{l}H_l} \,, \label{eqn:PeriodicStructureAnsatz_app}
\end{equation}
with all the generators $i H_l \in \g$ such that $U \in \mathcal{G}$. For notational brevity, we dropped the dependency on the circuit parameters $\bm{\theta}$ in our notation.

First note that from the Definition.~\ref{defn:adjoint_representation_unit} of the adjoint representation we have $\Phi_{\lambda}^{\operatorname{Ad}}(VV') = \Phi_{\lambda}^{\operatorname{Ad}}(V)\Phi_{\lambda}^{\operatorname{Ad}}(V')$ for any $V$ and $V' \in \mathcal{G}$. Hence, we can obtain the adjoint representation of the overall circuit through
\begin{equation}
    \Phi_{\lambda}^{\operatorname{Ad}}(U)=\left(\prod_{l=1}^Le^{-i\theta_{l} \bar{\bm{H}}_{l}}\right).
    \label{eqn:adjoint_whole_circuit}
\end{equation}
with $\bar{\bm{H}}_{l}\equiv \Phi_{\lambda}^{\operatorname{ad}}(H_{l})$ the adjoint representations of the generators $H_l$. 
Eq.~\eqref{eqn:adjoint_whole_circuit} specifies how elements $B_\alpha^{(\lambda)}$ of the irrep basis are transformed under conjugation by the circuit:
\begin{equation}
    U^\dagger B_\alpha^{(\lambda)} U = \sum_{\beta}\left(\Phi_{\lambda}^{\operatorname{Ad}}(U)\right)_{\alpha \beta}B_\beta^{(\lambda)}\,.
    \label{eqn:whole_basis_transform}
\end{equation}

Second, making use of the cyclicity of the trace and of the identity $U U^{\dag}=I$, we get
\begin{equation}
    \langle O \rangle
    =\Tr\left[\left( \prod_a \tilde{O}^{(a)} \right)\rho^{(\operatorname{in})}\right], \quad \text{where}\; \tilde{O}^{(a)} = U^{\dagger} O^{(a)} U.
    \label{eqn:O_almost_final}
\end{equation}
Furthermore, from Eq.~\eqref{eqn:whole_basis_transform}, we know that we can write \begin{equation}
    \tilde{O}^{(a)}=\sum_{\alpha}(\tilde{\bm{w}}^{(a)})_\alpha B^{(\lambda_a)}_\alpha, \quad \text{with}\; \tilde{\bm{w}} = \bm{w} \cdot \Phi_{\lambda}^{\operatorname{Ad}}(U).
\end{equation}
By combining this expression to Eq.~\eqref{eqn:O_almost_final} we are in measure to evaluate the expectation value. In the following we start by the case where $O$ is the product of a single (equivalent to Result.~\ref{thm:basic_evolution}) and two observables (equivalent to Result.~\ref{thm:second_order_correlator_evolution}) before addressing to the more general case.

\medskip
     \textbf{(Product of a single observable)} For the case where $O \in \irrep$ we obtain:
     \begin{equation}
         \langle O \rangle = \sum_{\alpha}(\tilde{\bm{w}})_\alpha \Tr\left[B^{(\lambda)}_\alpha \rho^{(\operatorname{in})} \right]  = \tilde{\bm{w}} \cdot \bm{e},
     \end{equation}
     where we have defined the vector of expectation values $\bm{e}$ such that $\left(\bm{e}\right)_\alpha\equiv\Tr[B_{\alpha}^{(\lambda)}\rho^{(\operatorname{in})}]$.
     
\medskip
          \textbf{(Product of two observables)} For $O=O^{(1)} O^{(2)}$ with $O^{(1)} \in \mathcal{L}_{\lambda_1}$ $O^{(1)} \in \mathcal{L}_{\lambda_2}$  we obtain
     \begin{equation}
         \langle O \rangle = \sum_{\alpha_1, \alpha_2}(\tilde{\bm{w}}^{(1)})_{\alpha_1} (\tilde{\bm{w}}^{(2)})_{\alpha_2} \Tr\left[B^{(\lambda_1)}_{\alpha_1} B^{(\lambda_2)}_{\alpha_2} \rho^{(\operatorname{in})} \right]  = \tilde{\bm{w}}^{(1)} \cdot \bm{E} \cdot \tilde{\bm{w}}^{(2)}.
     \end{equation}
     where we have defined the matrix of expectation values $\bm{E}$ such that $\bm{E}_{\alpha\beta}\equiv \Tr[B_{\alpha}^{(\lambda_1)} B_{\beta}^{(\lambda_2)}\rho^{(\operatorname{in})}]$.

\medskip
          \textbf{(Product of observables)} For the product of $M$ observables, with $a=1, \hdots, M$ in Eq.~\eqref{eq:expectation_product}, we get
     \begin{equation}
         \langle O \rangle = \sum_{\alpha_1,\hdots, \alpha_M}\left(\prod_{a=1}^M \tilde{\bm{w}}^{(a)}_{\alpha_a} \right)  \Tr\left[\left(\prod_{a=1}^M B_{\alpha_a}^{(\lambda_a)}\right) \rho^{(\operatorname{in})} \right]  = \sum_{\alpha_1,\hdots, \alpha_M} \bm{T}_{\alpha_1,\hdots,\alpha_M} \tilde{\bm{w}}_{\alpha_1} \hdots \tilde{\bm{w}}_{\alpha_M}, 
     \end{equation}\label{eqn:product_arbitrary}
     where we have defined the tensor of expectation values $\bm{T}$ such that $\bm{T}_{\alpha_1,\hdots,\alpha_M}\equiv \Tr[\left(\prod_{a=1}^M B_{\alpha_a}^{(\lambda_a)}\right) \rho^{(\operatorname{in})}]$.

To conclude, let us briefly comment on the different complexities entailed. We assume that the terms $e^{-i\theta_{l} \bar{\bm{H}}_{l}}$ appearing in Eq.~\eqref{eqn:adjoint_whole_circuit} are known (i.e., that the exponentiation has already been performed) and aim at assessing the complexities entailed when evaluating Eq.~\eqref{eqn:product_arbitrary}. (i) First, we can evaluate the distinct $\tilde{\bm{w}}^{(a)}$ through matrix-vector multiplication with vector dimensions $\dim(\mathcal{L}_{\lambda_a})$ for each of the irreps $\mathcal{L}_{\lambda_a}$ involved. (ii) Second, we need to evaluate a number of  $\prod_a \dim(\mathcal{L}_{\lambda_a})$ expectation values to obtain the tensor $\bm{T}$. (iii) Lastly, we need to contract the tensor of expectation values together with the vectors $\tilde{\bm{w}}^{(a)}$. Overall we see that the complexity scales as $\mathcal{O}(\max_a \dim(\irrepa)^M)$. This complexity is dominated by the steps (ii) and (iii) and scales exponentially with the number of observables involved in the product decomposition of $O$. 
 As mentioned in the main text, the decomposition of an observable to be simulated through $\gsim$ is not unique and different decompositions can yield different computational scalings.

\subsection{Evolution of observables under noisy channels}
\label{app:noise}
In this Appendix, we extend results of the previous section in the presence of noise in the circuit: In Sec.~\ref{app:noisy_one_irrep}, we formalize the noisy setting under consideration and address the case where the observable is supported by a single irrep. In Sec.~\ref{app:noisy_product_irrep}, we rather consider the case of product of observables, each supported in a single irrep, as we shall see this case requires more care.

\subsubsection{Noisy simulation with observables supported by a single irrep}\label{app:noisy_one_irrep}
In the presence of noise, the unitary of Eq.~\eqref{eqn:PeriodicStructureAnsatz_app} is now replaced by the channel
\begin{equation}\label{eqn:PeriodicStructureAnsatzNoisy}
    \tilde{\mathcal{U}} = \bigcirc_{l=1}^L (\tilde{\Lambda}_l \circ \mathcal{U}_l),\quad \text{with}\;\; \mathcal{U}_l(\cdot):=U_l\, \cdot \, U_l^{\dagger} \quad \text{and} \;\; U_l=e^{-i \theta_l H_l}.
\end{equation}
Accordingly, the state obtained after the noisy circuit evolution is now given by $\rho^{(\operatorname{out})} \equiv \tilde{\mathcal{U}} ( \rho^{(\operatorname{in})})$.
Each of the $\tilde{\Lambda}_l$ terms capture the noise affecting the $l$-th gate of the circuit and can decomposed as
\begin{equation}\label{eq:Kraus}
    \tilde{\Lambda}_l(\cdot) = \sum_k E_{k,l} \cdot E^{\dagger}_{k,l},
\end{equation}
in terms of Kraus operators $E_{k,l}$ that satisfy $\sum_k E_{k,l}^\dag E_{k,l} = I$. When all the $\tilde{\Lambda}_l$ are the identity channel we retrieve the unitary evolution specified in Eq.~\eqref{eqn:PeriodicStructureAnsatz_app}.
To capture the type of noise that can be simulated with $\gsim$ we recall the following definition from the main text.
\begin{definition}[Normalizer of $\irrep$]\label{def:normalizer_app}
Given an irrep $\irrep$, we define its normalizer as all the operators that leaves the subspace invariant under conjugation: 
\begin{equation}
    \stab := \{A \in \mathcal{L}\, | \, A^{\dagger} X A \in \irrep, \, \forall X \in \irrep \}.
\end{equation}
\end{definition}
We highlight that by definition of the irreps we have $\G \subset \stab$ for any $\irrep$. However, 
the normalizer may contain many more operators $A \notin \G$.
In analogy to Definition~\ref{defn:adjoint_representation}, we can represent the action by conjugation of elements of $\stab$ onto $\irrep$ in terms of $\dim(\irrep) \times \dim(\irrep)$ matrices:
\begin{definition}[Adjoint representation of $\stab$]
Given an  operator  $A$ in $\stab$, its  adjoint representation on the irrep $\irrep$ is obtained via the map $\Phi_{\lambda}^{\operatorname{N}}:\stab \mapsto \mathbb{R}^{\dim(\irrep)\times \dim(\irrep)}$ and is defined by
\begin{equation}        \left(\Phi_{\lambda}^{\operatorname{N}}(A)\right)_{\alpha \beta}\equiv  \Tr[B_\alpha^{(\lambda)}A^\dag B_\beta^{(\lambda)}A] \,.
\label{eqn:adjoint_representation_defn_stab}
\end{equation}
\label{defn:adjoint_representation_stab}
\end{definition}
By linearity and Definition~\ref{def:normalizer_app} we see that, if $E_{k,l} \in \stab$ for all $k$ in the Kraus decomposition of Eq.~\eqref{eq:Kraus} then 
\begin{equation}
    \forall X \in \irrep, \quad \tilde{\Lambda}_l(A)\in \irrep.
\end{equation}
That is, if a noise channel has a Kraus decomposition with all its Kraus operators belonging to the normalizer, then the action of this channel also leaves $\irrep$ invariant.
In such case, we can define the action of the noise channel through
\begin{equation}
\Phi_{\lambda}(\tilde{\Lambda}_l) \equiv  \sum_k \Phi_{\lambda}^{\operatorname{N}}(E_{k,l}) \,.
\end{equation}
One can readily verify that given an observable $O=\sum_\alpha (\bm{w})_\alpha B_\alpha^{(\lambda)}$, supported by a single irrep $\irrep$, we get that $\tilde{\Lambda}_l (O) = \sum_\alpha (\bm{w}')_\alpha B_\alpha^{(\lambda)}$ where we can obtain the new weights as $\bm{w}' = \Phi_{\lambda}(\tilde{\Lambda}_l)\cdot \bm{w}$.
We are now in measure to generalize Res.~\ref{thm:basic_evolution} to noisy simulations:
\begin{result}[Noisy simulation of observables in the $\lambda$-th irrep]
Consider a noisy circuit as per Eq.~\eqref{eqn:PeriodicStructureAnsatzNoisy} with noise channels $\tilde{\Lambda}_l$ having a Kraus decomposition as given in Eq.~\eqref{eq:Kraus} with operators $E_{k,l}\in \stab$ for all $k$ and $l$. Let $O$ be an observable with support in $\irrep$ such that $O=\sum_\alpha (\bm{w})_\alpha B_\alpha^{(\lambda)}$ for $\bm{w}\in \mathbb{R}^{\dim(\irrep)}$.  Then, given $\bm{e}^{(\operatorname{in})}_\lambda$, we can compute
\begin{equation}
    \langle O(\bm{\theta}) \rangle=\bm{w}\cdot\bm{e}^{(\operatorname{out})}_\lambda \,,
\end{equation}
with the vector of output expectation values obtained as
\begin{equation}
\bm{e}^{(\operatorname{out})}_\lambda=\left(\prod_{l=1}^L \Phi(\tilde{\Lambda}_l)e^{-i\theta_{l}\bar{\bm{H}}_{l}}\right)\cdot\bm{e}^{(\operatorname{in})}_\lambda, 
    \label{eqn:gsim_ansatz_evolution_noisy}
\end{equation}
where $\bar{\bm{H}}_{l}\equiv \Phi_{\lambda}^{\operatorname{ad}}(H_{l})$ and $\Phi(\tilde{\Lambda}_l)\equiv \sum_k\Phi_{\lambda}^{\operatorname{N}}(E_{k,l})$.
\label{thm:basic_noisy_evolution}
\end{result}
Notably, Result~\ref{thm:basic_noisy_evolution} enables us to perform \emph{exact} noisy simulations while retaining the same time and memory complexity as the noiseless case. This is in contrast with typical noisy simulations that either are (i) approximate and based on simulating many pure state trajectories, corresponding to random realization of the noise, or are (ii) exact but incur a quadratic increase in memory and computing requirements due to manipulations of density matrices. We note that the case of noiseless simulations of operators in many irreps (Result~\ref{thm:multi_evolution}) is readily ported to the noisy case in a similar fashion, and again do not incur additional overhead. However, as is now discussed the case of product of operators (Result~\ref{thm:second_order_correlator_evolution}) requires more care to be ported to the noisy setup  and incurs additional limitations.

\subsubsection{Simulation with products of observables each supported by a single irrep}
\label{app:noisy_product_irrep}

As per Eq.~\eqref{eq:expectation_product}, we now consider the case where $O = \prod_a O^{(a)}$ and $O^{(a)} \in \mathcal{L}_{\lambda_a}$ for any $a=1, \hdots, M$. The reason why the case with product of observables requires more care is because we cannot deal with each of the observables $O^{(a)}$ separately.
In particular, we see that while 
\begin{equation}\label{eq:factorization_unitary}
\mathcal{U}_l(O) = U^{\dagger} O U = U^{\dagger} O^{(1)} U U^{\dag} O^{(2)} \hdots U^{\dag} O^{(M)} U = \prod^M_{a=1} U^{\dagger} O^{(a)} U
\end{equation}
holds when the channel is unitary, the same is not true for arbitrary noise channel: 
\begin{equation}
\tilde{\Lambda}_l^{\dagger}(O) \neq \prod^M_{a=1} \tilde{\Lambda}^{\dagger}_l(O^{(a)}).
\end{equation}
Still, upon further restrictions of the noise channels, we can perform noisy simulation of product of observables with $\gsim$. In particular we now require that the noise channels acting on our circuits adopt the form:
\begin{equation}\label{eq:Kraus_unit}
    \tilde{\Lambda}_l(\cdot) = \sum_{k=1}^{K} p_{k,l} V_{k,l} \cdot V^{\dagger}_{k,l},
\end{equation}
with the $p_{k,l}$ valid distributions for each $l$ (i.e., $p_{k,l} \geq 0$ and $\sum_k p_{k,l}=1$) and where the $V_{k,l}\in \staba$ need to be unitaries that all belong to the normalizer of each $\irrepa$. For ease of notations, we have assumed the same number $K$ of unitaries involved for each noise channel. 

Eq.~\eqref{eq:Kraus_unit} is a restriction of the Kraus decomposition in Eq.~\eqref{eq:Kraus}, but still encompasses many noise channels of interest  such as the Pauli noise channels, whereby each of the $V_{k,l}$ is a Pauli string. Now we have that
\begin{equation}
\tilde{\Lambda}^{\dagger}_l(O) = \sum_k p_{k,l} \prod^M_{a=1} V^{\dagger}_{k,l} O^{(a)} V_{k,l},
\end{equation}
that resembles Eq.~\eqref{eq:factorization_unitary}, except for the weighted sum, and can be used for the purpose of simulations. 
In particular, denoting as $\mathcal{V}_{k,l}(\cdot):=V_{k,l}\, \cdot \, V_{k,l}^{\dagger}$, the expectation value in Eq.~\eqref{eq:expectation_product} can be written as
     \begin{equation}
     \label{eq:expectation_product_noisy}
         \langle O \rangle = \sum_{\vec{k}}p_{\vec{k}}\Tr\left[\left( \prod_a \tilde{O}_{\vec{k}}^{(a)} \right)\rho^{(\operatorname{in})}\right], \quad \text{where}\; \tilde{O}_{\vec{k}}^{(a)} = \bigcirc^1_{l=L} (\mathcal{U}^{\dagger}_l \circ \mathcal{V}^{\dag}_{\vec{k}_l,l}  ) (O^{(a)}).
     \end{equation}
     where we have defined $\vec{k}=\{0, \hdots, K\}^L$ a vector of noise trajectories, with entries $\vec{k}_l$ indicating the unitary $V_{\vec{k}_l, l}$ that has been applied at layer $l$, and $p_{\vec{k}} = \prod^L_{l=1} p_{\vec{k}_l,l}$ the probability associated to the trajectory. 
Denoting as $\Phi_{\lambda_a}^{\operatorname{N}}(V_{k,l})$ the adjoint representation of the $V_{k,l} \in \staba$, we get that \begin{equation}\label{eq:product_evol_coeffs}
    \tilde{O}_{\vec{k}}^{(a)}=\sum_{\alpha}(\tilde{\bm{w}}_{\vec{k}}^{(a)})_\alpha B^{(\lambda_a)}_\alpha, \quad \text{with}\; \tilde{\bm{w}}_{\vec{k}} = \bm{w} \cdot \prod_{l=1}^L \left (\Phi_{\lambda_a}^{\operatorname{N}}(V_{\vec{k}_l,l}) \Phi_{\lambda_a}^{\operatorname{Ad}}(U_l) \right).
\end{equation}
Overall, we can estimate Eq.~\eqref{eq:expectation_product_noisy} through the sampling of $S$ noise trajectories $\vec{k}$ with probability $p_{\vec{k}}$, and evaluating the corresponding expectation values through Eq.~\eqref{eq:product_evol_coeffs} and the noiseless results for product of observables detailed in Sec.~\ref{appendix:simul_product}. The variance of such estimates scales as $1/S$ and the complexity as $\mathcal{O}(S \max_a \dim(\irrepa)^M)$ as per Sec.~\ref{appendix:simul_product}.

\subsection{Analog computing}
\label{appendix:analog_computing}

So far, we have only considered the evolution of states under the action of quantum circuits (or corresponding mixed-unitary channels). Still, the principles underpinning $\gsim$ can equally be applied to the case where the system is \emph{continuously} driven (i.e., in the analog computing paradigm). In the following, we limit ourselves to the noiseless case where the observable of interest is supported by a single irrep, but stress that extensions to the case of product of observables in differerent irreps (as discussed earlier in Sec.~\ref{appendix:simul_product}) or to the noisy case (as discussed earlier in Sec.~\ref{app:noise} and with similar restrictions on the noise that can be supported) readily follow.

In the following, we consider evolution under a time-dependent Hamiltonian 
\begin{equation}
    \mathcal{H}(t)= \sum_h c_h(t) H_h,
    \label{eq:ctl_hamiltonian}
\end{equation} 
with constituents $i H_h \in \g$ and control functions $c_h(t)$.
Given an initial state $\rho^{(in)}$, the state of the system at time $t$ can be expressed as $\rho(t) = U(t)\rho^{(in)} U^{\dagger}(t)$ where $U(t) = \mathcal{T} \big[ {\rm exp} (-i \int^t_0\, \mathcal{H}(t') \, dt') \big]$ with $\mathcal{T}$ the time-ordering operator.
In analogy to simulation in the quantum circuit scenario, we are now tasked with evaluating the expectation value $\langle O \rangle \equiv {\rm Tr}[O \rho(T)]$ of an observable $O$ supported by $\g$ for the state $\rho(t=T)$ obtained at the end of the evolution.

Defining (in the Heisenberg picture) the time-dependent operator $O(t) \equiv U^{\dagger}(t) O U(t)$, we have that $\langle O \rangle ={\rm Tr}[O(T)\rho^{(\operatorname{in})}]$. Furthermore, the time evolution of $O(t)$ satisfies the ordinary differential equation (ODE):
\begin{equation}
    \frac{d}{dt} O(t) = i [\mathcal{H}(t), O(t)], 
    \label{eqn:ode}
\end{equation}
with initial condition $O(t=0)=O$.
From the Definition~\ref{defn:adjoint_representation}, it can be verified that 
\begin{equation}
    [\mathcal{H}(t), B_{\alpha}^{(\lambda)}] = \sum_\beta \big[ \adrep{\mathcal{H}(t) } \big]_{\alpha \beta}B_{\beta}^{(\lambda)},   
\end{equation}
with the matrix $\adrep{ \mathcal{H}(t) } = \sum_h c_h(t) \adrep{ H_h}$.

Hence, decomposing $O(t)=\sum_\alpha (\mathbf{o}(t))_\alpha B_{\alpha}^{(\lambda)}$ 
in the observable basis $\{B_{\alpha}^{(\lambda)}\}$ of the irrep $\irrep$, one can recast Eq.~\eqref{eqn:ode} as the ODE
\begin{equation}
    \frac{d}{dt} \mathbf{o}(t) = i \adrep{ \mathcal{H} } \mathbf{o}(t)
    \label{eqn:ode_adjoint}
\end{equation}
over the vector $\mathbf{o}(t)$ that has dimension $\dim(\irrep)$.

In summary, evaluating the expectation value $\langle O \rangle$ consists in two steps. 
First, solve the dynamics of $\mathbf{o}(t)$ from $t=0$ to $t=T$ according to Eq.~\eqref{eqn:ode_adjoint}. This is achieved by standard numerical ODE solvers, and only requires acting on a $\dim(\irrep)$ vector space. 
Then, upon evaluation of $\mathbf{o}(T)$, one can compute $\langle O \rangle = \mathbf{o}(T) \cdot \mathbf{e}^{(\operatorname{in})}$, where we recall that the vector $\bm{e}^{(\operatorname{in})}$ contains  the expectation values $\Tr[B_{\alpha}^{(\lambda)}\rho^{(\operatorname{in})}]$ of all $B_{\alpha}^{(\lambda)}$ with respect to the initial state.

Overall, this extends the scope of $\gsim$ to analog (and digital-analog) quantum computing scenarios. As an example of application, in Ref.~\cite{cote2022diabatic} such capabilities were used to evaluate (and optimize) continuous diabatic schedules for tasks of ground state preparations over system sizes of up to $n=39$ spins.

\section{Efficient implementation of $\gsim$}\label{appendix:efficient_impl}

The $\gsim$ framework outlined in Section \ref{sec:gsim_principles} allows scalable classical simulations for products of observables supported by irreps of $\G$ whenever $\dimg\in\mathcal{O}(\operatorname{poly}(n))$. In addition to extending the scope of previous Lie-algebraic simulation techniques~\cite{somma2005quantum,somma2006efficient}, another contribution of this work is to introduce new optimizations which substantially improve the efficiency of $\gsim$. In the following subsections, we detail three main improvements:
\begin{itemize}
    \item Pre-computing eigendecompositions of gate generator adjoint representations reduces the time complexity of $\gsim$ evolution from $\mathcal{O}(\dim(\irrep)^3)$ to $\mathcal{O}(\dim(\irrep)^2)$ (\ref{sec:efficient_implementation_gsim}).
    \item Relevant matrices are extremely sparse when $\irrep$ has a Pauli basis, allowing speedups by use of sparse methods, even for subalgebras whose generators are not sparse (\ref{appendix:sparsity}).
    \item Eigendecompositions of gates with Pauli generators can be computed in $\mathcal{O}(\dim(\irrep))$ in the representation of a Pauli basis, instead of the usual $\mathcal{O}(\dim(\irrep)^3)$ (\ref{appendix:eigendecompositions}).
\end{itemize}

In practice, we found such optimizations to significantly accelerate our simulations, and all the numerics presented in this work rely on a Python implementation of $\gsim$ incorporating them.

\subsection{Efficient evolution}
\label{sec:efficient_implementation_gsim}
Evaluating each of the $LK$ terms $e^{-i\theta\bar{\bm{H}}_k}$ appearing in Eq.~\eqref{eqn:gsim_ansatz_evolution} of the simulation routine has time complexity scaling as $\mathcal{O}(\dim(\irrep)^3)$ using standard matrix exponentiation routines. However, given the full-rank eigendecompositions
\begin{equation}
    \bar{\bm{H}}_k=\bar{\bm{Q}}_k\text{diag}(\epsilon_k^{(1)},\hdots,\epsilon_k^{(\dim(\irrep))})\bar{\bm{Q}}_k^T\label{eqn:hamiltonianeigendecomposition}
\end{equation}
of the (adjoint representation of the) gate generators $\bar{\bm{H}}_k$, the action of these parametrized gates is given as
\begin{equation}
    \Phi^{\operatorname{Ad}}_{\lambda}(e^{-i\theta H_k)}=\bar{\bm{Q}}_k\text{diag}(e^{-i\theta\epsilon_k^{(1)}},\dots,e^{-i\theta\epsilon_k^{(\dim(\irrep))}})\bar{\bm{Q}}_k^T,
    \label{eqn:spectralpropagators}
\end{equation}
which can be applied to a vector $\bm{e}\in\mathbb{R}^{\dim(\irrep)}$ with $\mathcal{O}(\dim(\irrep)^2)$ time complexity.
Notably, these decompositions only need to be performed once per algebra and per gate generator used, the number of which is typically no more than polynomial in the system size $n$ for relevant applications. 
Pre-computing these eigendecompositions na\"ively scales as $\mathcal{O}(\dim(\irrep)^3)$, however, in Appendix~\ref{appendix:eigendecompositions} we detail a procedure by which the eigendecompositions for any Pauli-basis irrep (e.g. $i\g_0$) can be computed in $\mathcal{O}(\dim(\irrep))$.
An explicit algorithm for simulation incorporating these pre-computed eigendecompositions is provided in Algorithm \ref{alg:evolution}.

\begin{algorithm}[H]
\caption{Efficient evolution of observable expectation values using $\mathfrak{g}$-sim.}\label{alg:evolution}
\begin{algorithmic}
\INPUT Circuit parameters $\bm{\theta}$, input observables $\bm{e}^{(\text{in})}$, circuit generator eigendecompositions (Eq.~\eqref{eqn:spectralpropagators}) $\bar{\bm{Q}}_k$ and $\epsilon_k^{(g)}$ 
\State $\bm{e}^{(\operatorname{out})} \gets \bm{e}^{(\text{in})}$
\For{$l$ in $1,\dots,L$}
    \For{$m$ in $1,\dots,M$}
        \State $\bm{e}^{(\operatorname{out})} \gets \bar{\bm{Q}}_{k}^T\bm{e}^{(\operatorname{out})}$
        \For{$g$ in $1,\dots,\dim(\irrep)$}
            \State $(\bm{e}^{(\operatorname{out})})_g \gets e^{-i\theta_{lk}\epsilon_{k}^{(g)}}(\bm{e}^{(\operatorname{out})})_g$
        \EndFor
        \State $\bm{e}^{(\operatorname{out})} \gets \bar{\bm{Q}}_{k}\bm{e}^{(\operatorname{out})}$
    \EndFor
\EndFor
\OUTPUT Vector $\bm{e}^{(\operatorname{out})}$ of basis observables in $\irrep$
\COMPLEXITY $\mathcal{O}(M\dim(\irrep)^2)$

\end{algorithmic}
\end{algorithm}

\subsection{Sparsity of adjoint representation}\label{appendix:sparsity}
When the relevant basis of $\irrep$ consists of Pauli strings and the gate generators are also Pauli strings (e.g. when performing simulations in $i\g_0$ where all gates have Pauli operator generators), the adjoint representations $\bar{\bm{H}}_l$ of the  gate generators and their corresponding parametrized gates are of an extremely sparse form, having fewer than $2\dim(\irrep)$ nonzero entries per $(\dim(\irrep) \times \dim(\irrep))$-dimensional matrix. Thus, substantial speedups can be obtained by leveraging sparse linear algebra libraries. This is the case for most simulations presented in this work.

However, we sometimes must apply gates whose generators are not Pauli strings, such as $e^{-i\theta\sum_{j=1}^{n-1}\sigma^x_j \sigma^x_{j+1}}$ for VQE on the LTFIM (Sec.~\ref{sec:LTFIM}), or $e^{-i\theta\sum_{j=1}^{n-1}\sigma^z_j \sigma^z_{j+1}}$ for QAOA (Sec.~\ref{sec:QAOA}). For gates like these, the adjoint representation of the generator does \textit{not} admit an efficient sparse representation, and will in general be dense. Figure~\ref{fig:gate_sparsity} depicts this difference in sparsity. 
\begin{figure}
    \centering
    \includegraphics[width=0.48\textwidth]{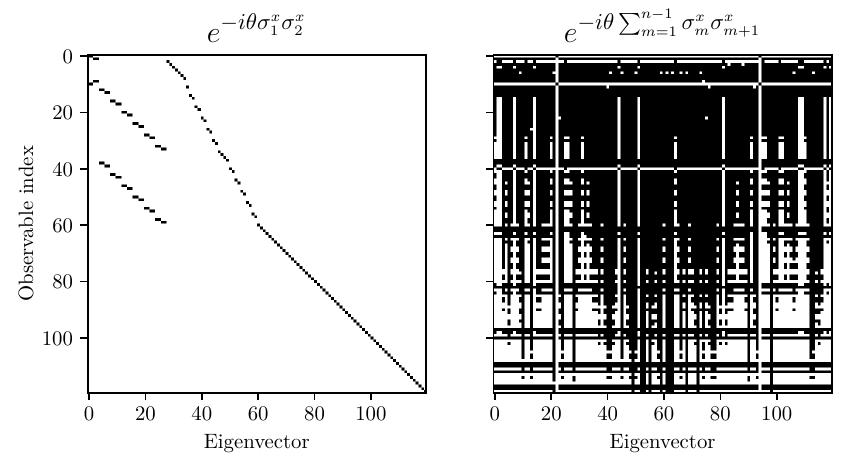}
    \caption{\textbf{Sparsity of gates generated by Pauli strings versus general gates.} We plot sparsity maps for the eigenvector matrix $\bar{\bm{Q}}$ where $\Phi^{\operatorname{Ad}}_{\lambda}(\exp(-i\theta H))=\bar{\bm{Q}}\text{diag}(e^{-i\theta\epsilon^{(1)}},\dots,e^{-i\theta\epsilon^{(\dimg)}})\bar{\bm{Q}}^T$, for $H=\sigma^x_1\sigma^x_2$ (left) and $H=\sum_{m=1}^{n-1} \sigma^x_m \sigma^x_{m+1}$ (right), for a system of $n=8$ qubits in the $\irrep=i\mathfrak{g}_0$ representation. White (black) entries in the map indicates (non) zero entries. The gate $\exp(-i\theta \sigma^x_1\sigma^x_2)$ has a sparse eigenvector matrix since it is generated by a single Pauli string, whereas the gate $\exp(-i\theta\sum_{m=1}^n \sigma^x_m \sigma^x_{m+1})$ has a dense eigenvector matrix since it is generated by a sum of Pauli strings. }
    \label{fig:gate_sparsity}
\end{figure}

For gates not generated by single Pauli strings, but rather sums of commuting Pauli strings, one can factorize the gate and still use sparse methods. For example
\begin{align}
    e^{-i\theta\sum_{j=1}^{n-1}\sigma^x_j \sigma^x_{j+1}} &= \prod_{j=1}^{n-1}e^{-i\theta \sigma^x_j \sigma^x_{j+1}} \\
    \implies \underbrace{\Phi^{\operatorname{Ad}}_{\lambda}\left(e^{-i\theta\sum_{j=1}^{n-1}\sigma^x_j \sigma^x_{j+1}}\right)}_{\text{dense}} &= \prod_{j=1}^{n-1}e^{-i\theta\overbrace{\Phi_{\lambda}^{\operatorname{ad}}(\sigma^x_j \sigma^x_{j+1})}^{\text{sparse}}},
\end{align}
since the adjoint representation is a homomorphism. However, in such cases the dense nature of the eigenvector matrix can make GPU-accelerated dense linear algebra methods highly appealing. For example, we consider tests on our hardware (a single core of an Intel Xeon Platinum 8268 for CPU methods vs a NVIDIA Tesla V100 for GPU methods), for our implementation of $\mathfrak{g}$-sim. We found that for systems with sparse gates only, GPU acceleration provided an overall speedup of a factor of 3-5 versus a single CPU core, whereas for systems with dense gates, GPU acceleration provided an overall speedup of a factor of 60-100. The choice of sparse or dense methods and CPU or GPU hardware must be made judiciously depending on problem type and the need for parallelizability.

\subsection{Efficient calculation and eigendecomposition of Pauli-basis adjoint representations}
\label{appendix:eigendecompositions}

For our Lie-algebraic simulations, we must be able to compute the adjoint representations (Definition \ref{defn:adjoint_representation}) of our gate generators $H_k$:
\begin{equation}
    \left(\bar{\bm{H}}_k\right)_{\alpha\beta}=-\Tr\left[H_k\left[G_\alpha,G_\beta\right]\right].
    \label{eqn:adjoint_generator}
\end{equation}
na\"ively computing an adjoint representation via Eq.~\eqref{eqn:adjoint_generator} would scale as $\mathcal{O}(\dim(\irrep)^2)$. However, if the basis operators $G_\alpha$ of $\irrep$ are Pauli strings
\begin{equation}
    G_\alpha = \bigotimes_{m=1}^n \sigma_m^{s_{\alpha, m}}, \quad s_{\alpha, m}\in \{x,y,z\},
\end{equation}
as is the case for $i\mathfrak{g}_0$, and $H_k$ is also a Pauli string, then for fixed $k$ and $\alpha$ ($\beta$), Eq.~\eqref{eqn:adjoint_generator} has a nonzero value for at most one value of $\beta$ ($\alpha$). That is, each row and column have at most one nonzero entry. Thus, one can construct the adjoint representation for $\bar{\bm{H}}_k$ in just $\mathcal{O}(\dim(\irrep))$ by iterating over $\alpha\in \{1,\dots,\dim(\irrep)\}$, inverting Eq.~\eqref{eqn:adjoint_generator} to find the Pauli string $\sigma$ such that $\Tr[H_k[G_\alpha,\sigma]]$ (which is possible since by assumption $H_k$ and $G_\alpha$ are Pauli strings), checking if $\sigma=G_\beta$ for some $\beta$ (which can be done in $\mathcal{O}(n)$ by using a hashmap), and populating the corresponding element if so.

Similarly, although na\"ively computing the eigendecomposition \eqref{eqn:hamiltonianeigendecomposition} would be $\mathcal{O}(\dim(\irrep)^3)$, one may compute it in just $\mathcal{O}(\dim(\irrep))$ in the case where the $H_k$ and $G_\alpha$ are Pauli strings. Since the $\bar{\bm{H}}_k$ are Hermitian and at most one element of each row/column is nonzero, the adjoint representations take the form (up to some permutation of rows and columns)
\begin{equation}
    \bar{\bm{H}}_k=2\left(\bigoplus_{m_1=1}^{(\dim(\irrep) - \operatorname{dim}(\ker(\bar{\bm{H}}_k)))/2} \bar{\bm{\sigma}}^y \right)\oplus\left(\bigoplus_{m_2=1}^{\operatorname{dim}(\ker(\bar{\bm{H}}_k))/2} \bar{\bm{0}}\right),
    \label{eqn:sparse_adjoint_form}
\end{equation}
where the factor of $2$ originates from the Pauli commutation relations, and $\bar{\bm{\sigma}}^y=\begin{bmatrix}0&-i \\ i & 0\end{bmatrix}$ and $\bar{\bm{0}}=\begin{bmatrix}0&0\\0&0\end{bmatrix}$ are the Pauli Y matrix and the zero matrix, respectively. Since the eigendecomposition of $\bar{\bm{\sigma}}^y$ is known, one may compute the eigendecomposition of $\bar{\bm{H}}_k$ in $\mathcal{O}(\dim(\irrep))$ by iterating over the $\bar{\bm{\sigma}}_y$ summands, populating the corresponding eigenvectors and eigenvalues, and using any set of vectors spanning $\ker(\bar{\bm{H}}_k)$ as the remaining eigenvectors with zero eigenvalues.

\section{Gradient calculations}
\label{appendix:gradients}
\subsection{Observable gradients}
\label{appendix:observable_gradients}
For VQA and QML applications, it is often beneficial to be able to efficiently calculate gradients $\grad_{\bm{\theta}}\langle O \rangle$ of observables $O=\sum_\alpha (\bm{w})_\alpha B_\alpha^{(\lambda)}$ for $\bm{w}\in\mathbb{R}^{\dim(\irrep)}$, $B_\alpha^{(\lambda)}\in\irrep$. Differentiating Eq.~\eqref{eqn:expectations_from_e_vec} with respect to $\theta_m$, we find
\begin{equation}
    (\grad_{\bm{\theta}} \langle O\rangle)_m=\frac{\partial\langle O\rangle}{\partial\theta_{m}}=\bm{w}^T\left(\frac{\partial \Phi_{\lambda}^{\operatorname{Ad}}(U(\bm{\theta}))}{\partial \theta_{m}}\right)\bm{e}_\lambda^{(\text{in})},
    \label{eqn:circuit_gradients}
\end{equation}
where we may compute
\begin{equation}
\frac{\partial \Phi_{\lambda}^{\operatorname{Ad}}(U(\bm{\theta}))}{\partial \theta_{m}}=-i\left(\bar{\bm{U}}_{m+1:M}\bar{\bm{H}}_{k}\bar{\bm{U}}_{1:m}\right),
\label{eqn:gradient}
\end{equation}
where $\bar{\bm{U}}_{a:b}\equiv e^{-i\theta_b \bar{\bm{H}}_{b}} e^{-i\theta_{b-1} \bar{\bm{H}}_{b-1}}\dots e^{-i\theta_a \bar{\bm{H}}_{a}}$, $m$ is a combined index for $l$ and $k$, and $M=LK$.
Na\"ively evaluating Eq.~\eqref{eqn:gradient} requires $\mathcal{O}(M^2)$ matrix-vector multiplications (each $\mathcal{O}(\dim(\irrep)^2)$), however by exploiting a recurrent property of the gradients, one may compute the gradient with only $\mathcal{O}(M)$ matrix-vector multiplications (i.e., with the same asymptotic complexity as simulating the circuit in the first place). This procedure, which follows the `reverse-mode' gradient calculation of Ref.~\cite{jones2020efficientcalculation} in conjunction with the efficient spectral approach of Algorithm \ref{alg:evolution}, is outlined in Algorithm \ref{alg:gradients}.
\begin{algorithm}[H]
\caption{Calculating gradients in $\mathfrak{g}$-sim.}\label{alg:gradients}
\begin{algorithmic}
\INPUT Circuit $\bm{\theta}$, input expectation value vector $\bm{e}^{(\text{in})}$, circuit generator adjoint representations $\bar{\bm{H}}_k$ and eigendecompositions $\bar{\bm{Q}}_k$, $\epsilon_k^{(g)}$, coefficients $\bm{w}$ of observable $O$
\State $\bm{\eta} \gets \bar{\bm{U}}\bm{e}^{(\text{in})}$  \Comment{Use the main loop of Algorithm \ref{alg:evolution}.}
\State $\bm{\phi} \gets \bm{w}$
\State $(\grad_{\bm{\theta}} \langle O\rangle)_M \gets -i\bm{\phi}^T \bar{\bm{H}}_{M}\bm{\eta}$
\For{$m$ in $N,\dots,2$}
    \State $\bm{\eta} \gets \bar{\bm{Q}}_{m} \bm{\eta}$
    \State $\bm{\phi} \gets \bar{\bm{Q}}_{m} \bm{\phi}$
    \For{$g$ in $1,\dots,\dim(\irrep)$}
        \State $\eta_g \gets e^{i\theta_{m}\epsilon_{m}^{(g)}}\eta_g$
        \State $\phi_g \gets e^{i\theta_{m}\epsilon_{m}^{(g)}}\phi_g$
    \EndFor
\State $\bm{\eta} \gets \bar{\bm{Q}}_{m}^T \bm{\eta}$
\State $\bm{\phi} \gets \bar{\bm{Q}}_{m}^T \bm{\phi}$

\State $(\grad_{\bm{\theta}} \langle O\rangle)_{m-1} \gets -i\bm{\phi}^T \bar{\bm{H}}_{m}\bm{\eta}$
\EndFor
\OUTPUT Gradient vector $\grad_{\bm{\theta}} \langle O\rangle $
\COMPLEXITY $\mathcal{O}(M\dim(\irrep)^2)$
\end{algorithmic}
\end{algorithm}
Higher order derivatives can also be computed, such as the Hessian
\begin{equation}
    (\grad_{\bm{\theta}}^2\langle O\rangle)_{jk}=\frac{\partial^2\langle O\rangle}{\partial \theta_j \partial \theta_k}=-\bm{h}^T \bar{\bm{U}}_{k+1:N}\bar{\bm{H}}_k\bar{\bm{U}}_{j:k}\bar{\bm{H}}_j\bar{\bm{U}}_{1:j} \bm{e}^{(\text{in})}.
    \label{eqn:hessian}
\end{equation}
However, since the circuit generators $\bar{\bm{H}}_k$ are typically non-invertible, Algorithm \ref{alg:gradients} cannot be generalized to evaluate Eq.~\eqref{eqn:hessian} and higher-order derivatives. Although still polynomial in complexity, evaluating a $P$-th order gradient for $P\geq 2$ will be a factor $\mathcal{O}(M^{P})$ slower than efficiently evaluating the gradient.

\subsection{Compilation of cost function gradients}
\label{appendix:compilation_gradients}
For unitary compilation (Sec.~\ref{sec:compilation}), the use of gradient-based optimizers requires the efficient calculation of $\partial \mathcal{L}_{\mathfrak{g}}/\partial \theta_m$. Differentiating Eq.~\eqref{eqn:adjointspaceunitarycostfunc}, we obtain
\begin{equation}
    \frac{\partial \mathcal{L}_{\mathfrak{g}}}{\partial \theta_m}=\frac{1}{\dimg}\operatorname{Im}\left[\Tr[ \bar{\bm{V}}^T\bar{\bm{U}}_{m+1:M}\bar{\bm{H}}_m\bar{\bm{U}}_{1:m}]\right],
    \label{eqn:Lgderivative}
\end{equation}
where we have used properties of the ansatz structure, given in Eq.~\eqref{eqn:PeriodicStructureAnsatz_app}, and $\bar{\bm{U}}^T_{n:m}\equiv \bar{\bm{U}}^T_n\bar{\bm{U}}^T_{n+1}\dots\bar{\bm{U}}^T_m$. Noting that the matrix terms in Eq.~\eqref{eqn:Lgderivative} have the same recurrent form as Eq.~\eqref{eqn:gradient}, we may compute the full gradient vector in $\mathcal{O}(N\dimg^3)$ by a straightforward modification of Algorithm \ref{alg:gradients} (with matrix-matrix multiplication instead of matrix-vector multiplication).

For efficient implementation, care should be taken with the order of operations when evaluating Eqs.~\eqref{eqn:adjointspaceunitarycostfunc} and \eqref{eqn:Lgderivative}, in particular to avoid any dense-dense matrix multiplication. 
For example, in Eq.~\eqref{eqn:Lgderivative}, one should store the products $(\bar{\bm{U}}_{m+1:m})^T\bar{\bm{V}}$ and $\bar{\bm{U}}_{1:m}$ as dense matrices, use sparse representations of the $\bar{\bm{U}}_m$ when applying the iterative step in Algorithm \ref{alg:gradients}, and use sparse representations of $\bar{\bm{H}}_m$ when evaluating each element of the gradient. We find that with this order of operations, the evaluation is substantially accelerated by the use of sparse-dense matrix multiplication algorithms, and further improved by the use of GPU acceleration (e.g. SpDM with cuSPARSE). Finally, the trace in Eq.~\eqref{eqn:Lgderivative} of the product of two dense matrices ($\bar{\bm{V}}^T\bar{\bm{U}}_{m+1:M}$ and $\bar{\bm{H}}_m\bar{\bm{U}}_{1:m}$) can be evaluated in a dense $\mathcal{O}(\dimg^2)$ operation as
\begin{equation}
    \Tr[ \bar{\bm{V}}^T\bar{\bm{U}}_{m+1:M}\bar{\bm{H}}_m\bar{\bm{U}}_{1:m}]=\sum_{\alpha\beta}\left[((\bar{\bm{U}}_{m+1:m})^T\bar{\bm{V}})\circ(\bar{\bm{H}}_m\bar{\bm{U}}_{1:m})\right]_{\alpha\beta},
\end{equation}
where $\circ$ denotes the Hadamard product. This avoids costly $\mathcal{O}(\dimg^3)$ dense-dense matrix multiplication entirely, which would otherwise be by far the most expensive operation in this procedure.

\subsection{Analog computing}\label{appendix:analog_computing_gradients}
Gradients of expectation values can also be obtained in the analog computing scenario that was laid out in Appendix~\ref{appendix:analog_computing}. 
While not providing full-fledged details of the implementation required, we specify the problem and point the reader toward dedicated references for implementation.

In the analog computing scenario, control of the system occurs at the Hamiltonian level through the control functions $c_h(t)$ appearing in Eq.~\eqref{eq:ctl_hamiltonian}. 
A first step towards defining the gradients of interest requires parametrizing these control functions in terms of a finite set of parameters $\bm{\theta}$ (i.e., $c_h(t) \rightarrow c_h(\bm{\theta}, t)$).
Upon such parametrization, and given an observable $O$, an optimization task would consist in identifying values of $\bm{\theta}$ that minimize the expectation value $\langle O \rangle$, with the gradients of interests denoted $\partial \langle O \rangle / \partial_{\bm{\theta}_m}$.

While similar in notations to the gradients of the circuit parameters in Eq.~\eqref{eqn:circuit_gradients}, we note that depending on the parametrization of the controls adopted, each of the control parameter may affect the control values along the whole evolution, and would appear particularly difficult to evaluate.
Still, progress have been made towards efficient evaluation of such gradients in the context of quantum optimal control and machine learning. 
In particular, methods presented in Refs.~\cite{machnes2018tunable,chen2018neural} generically enable such calculations and require to numerical solve an augmented ODE, but avoid the necessity of storing in memory the state of the system along each numerical step of its evolution.
Applications of these numerical methods have already found application in quantum optimal control problems~\cite{schafer2020differentiable,sauvage2022optimal} but to-date have remained limited to small system sizes due to the exponentially growth of the underlying Hilbert space. Incorporating them with $\gsim$ would circumvent such limitations for certain quantum dynamics.

\section{Review of Wick's theorem}\label{app:wick}
In this section, we review Wick's theorem from a Lie-algebraic perspective, which can be used to efficiently compute expectation values of arbitrary operators over the highest weight states of the circuit's Lie algebra~\cite{somma2005quantum,somma2006efficient}.  

Consider the Cartan-Weyl decomposition of $\g$, given by $\g=\mathfrak{h}\oplus \mathfrak{g}_+\oplus \mathfrak{g}_-$, where $\mathfrak{h}={\rm span}_\mathbb{R}i\{\widetilde{H}_k\}$ denotes the Cartan subalgebra, $\mathfrak{g}_+={\rm span}_\mathbb{R}i\{E_{\alpha' }\}$ and $\mathfrak{g}_-={\rm span}_\mathbb{R}i\{E_{-\alpha' }\}$, with $E_{\alpha' }$ and $E_{-\alpha' }$ being the sets of raising and lowering operators, respectively. Let us denote as $\ket{{\rm hw}}$ the highest weight state of the algebra. In what follows, we will study the task of simulating an expectation value  $\langle O(\bm{\theta}) \rangle \equiv\Tr[O U(\bm{\theta}) \dya{\rm{hw}} U^{\dagger}(\bm{\theta})]$. Note that more generally, we could also consider the case when the initial state is a generalized coherent state~\cite{gilmore1974properties,perelomov1977generalized,zhang1990coherent,delbourgo1977maximum}, i.e., a state of the form $\ket{\psi}=V\ket{{\rm hw}}$ with some (known) $V\in e^{\g}$ as the action of $V$ can be absorbed into $U(\thv)$. 

To begin, take $O\in i\g$. Then, as per Eq.~\eqref{eqn:adj_heis} we know that $U^{\dagger}(\bm{\theta})O U(\bm{\theta})$ will always be expressed as 
\begin{equation}
    U^{\dagger}(\bm{\theta})O U(\bm{\theta})=\sum_kc_{k}\widetilde{H}_k+\sum_{\alpha'} c_{\alpha '}E_{\alpha'}+c^*_{-\alpha '}E_{\alpha'}\,,
\end{equation}
where the coefficients $c_{k}$, $c_{\alpha '}$ and $c^*_{\alpha '}$ can be obtained from the adjoint representation of the unitary. Since\begin{equation}
    E_{\alpha'}\ket{{\rm hw}}=\bra{{\rm hw}}E_{-\alpha '}=0, \quad \widetilde{H}_k\ket{{\rm hw}}=\widetilde{h}_k\ket{{\rm hw}}\,,
\end{equation}
where $\widetilde{h}_k$ are the weights of $\ket{{\rm hw}}$, we can see that
\begin{equation}
    \langle O(\bm{\theta}) \rangle=\sum_kc_{k}\widetilde{h}_k\,.
\end{equation}
Note that up to this point, we have used the adjoint representation (just like in $\gsim$), but we have leveraged the fact that the expectation values of a highest-weight on the algebra elements can be readily computed via the Cartan-Weyl decomposition. This is the key idea behind Wick's theorem. 

Next, consider more general observables given by products of (polynomially-many) elements of the algebra such as $O=O_1 O_2$, with $O_1,O_2\in i\g$. These can belong to higher dimensional irreps, making a simulation via $\gsim$ more expensive according to Eq.~\eqref{theo:scaling}. However, we can again use the adjoint representation to evolve each element in the algebra as  $U^{\dagger}(\bm{\theta})O U(\bm{\theta})=U^{\dagger}(\bm{\theta})O_1 U(\bm{\theta})U^{\dagger}(\bm{\theta})O_2 U(\bm{\theta})$, which will lead to a summation of elements of the Cartan subalgebra, times raising and lowering root operators. Given that these act on the highest-weight, only a few of those terms will be non-zero. In the former case, only the terms $\langle \widetilde{H}_k\widetilde{H}_{k'}\rangle =\widetilde{h}_k\widetilde{h}_{k'} $, and $\langle E_{\alpha'}E_{-\alpha'} \rangle $ are non-zero. The latter can be  expressed in terms of components of the root vector (see~\cite{somma2005quantum}). 

More generally, the overall goal of Wick's theorem is to use the (known) commutation relations of the elements of the algebra to get all the lowering operator to the ``left'' and all the raising operators to the ``right'', as these will vanish when acting on the highest-weight state. While this commutation procedure could be computationally demanding in principle, in practice the properties of $\g$ can be used to greatly simplify the required calculations. For instance, when considering free-fermionic evolutions (see the main text), the highest weight, and generalized coherent states, correspond to  Gaussian states and the evaluation of expectation values over products of elements in the algebra can be mapped to the task of evaluating a Pfaffian~\cite{robledo2009sign}. Hence, the simulation of expectation values for free-fermionic evolutions becomes efficient when the initial state is a Gaussian, or a linear combination of few Gaussian states. 

When applying Wick's theorem to more general initial states, one would need to expand them into a linear combination of generalized coherent states, and apply Wick's theorem for each one, as well as for the cross terms.
In such case, the computational complexity scales with the number of generalized coherent states, and when such number is exponential, the overall cost becomes exponential, even for observables in $i\g$.

\section{Ansatz circuits and benchmarking}
\label{appendix:ansatz}
Details of all ansatz circuits used in this work are given in Table \ref{tab:ansatze}.
\subsection{Benchmark calculations}
\label{appendix:benchmarking}
For the compute time benchmarks in Fig.~\ref{fig:benchmarks}, we repeatedly simulate evolution of a state initialized to $\ket{\bm{0}}$ and transformed by a single layer ($L=1$) of the TFXY 2-local ansatz (see Table \ref{tab:ansatze}) using both our Python implementation of $\gsim$ and TensorFlow Quantum on a single core of an Intel Xeon Platinum 8268. We stress that unlike tensor network methods, the time complexity of applying a gate is independent of circuit depth, so we report our timings per gate. We collect timings using the \texttt{timeit} Python package in batches of repeated simulations (to minimize the influence of measurement distortion due to environmental factors) and average over batches. We bootstrap 95\% confidence intervals on these averages, but they are too small to be visualized in Fig.~\ref{fig:benchmarks}.

Memory requirements for state vectors are computed assuming double-precision arithmetic on a $2^n$-dimensional Hilbert space $\mathcal{H}$. Memory requirements for $\gsim$ are computed by profiling the actual memory footprint of the observable vector $\bm{e}$ and algebra data (adjoint representations of generators $\bar{\bm{H}}_k$ and their eigendecompositions).

\begin{table}
\centering
\begin{overpic}[abs,unit=1mm,width=\textwidth]{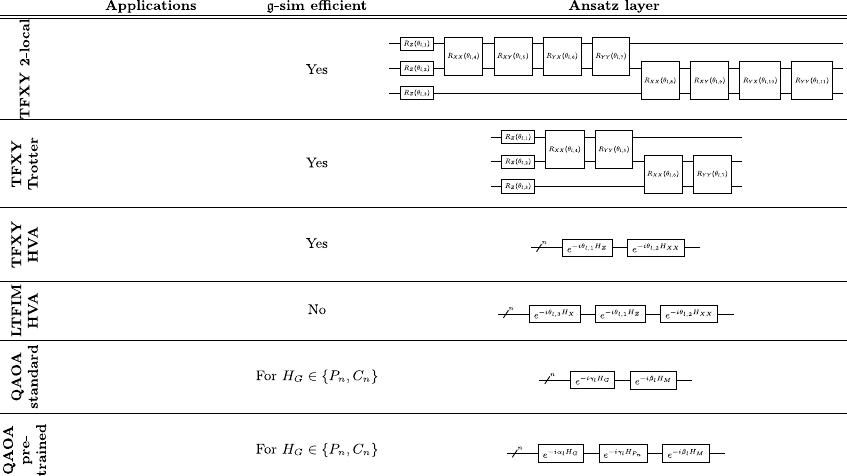}
\put(7,48){\parbox{2in}{
Benchmarking (\ref{sec:efficient_implementation_gsim}) \\ TFXY VQE (\ref{sec:overparametrization}) \\ Random compilation (\ref{sec:compilation_random}) \\ Phase classifier (\ref{sec:supervised_QML})
\vspace{0.45in} \\
Dynamics compilation (\ref{sec:compiling_dynamics}) 
\vspace{0.55in} \\
LTFIM pre-training (\ref{sec:LTFIM})
\vspace{0.4in} \\
LTFIM pre-training (\ref{sec:LTFIM})
\vspace{0.4in} \\
QAOA pre-training (\ref{sec:QAOA})
\vspace{0.45in} \\
QAOA pre-training (\ref{sec:QAOA}) } }
\end{overpic}
\caption{Layers of all variational ans\"{a}tze used in this work, with reference to their applications, whether or not they are classically efficient in $\gsim$ (for simulating dynamics of observables supported by $\mathfrak{g}$), and circuit diagrams for single layers of each ansatz. In diagrams where individual qubits are depicted, the example layer is given for $n=3$ for compactness.}
\label{tab:ansatze}
\end{table}

\section{Subalgebras of $\mathfrak{g}_0$}
\label{appendix:algebra_and_subalgebras}
Throughout this work, we use the algebra $\mathfrak{g}_0$, defined in Eq.~\eqref{eqn:g_tfxy}, which has dimension $\operatorname{dim}(\g_0) = n(2n-1)$ and encompasses many models of interest:
\begin{itemize}
    \item The Hamiltonian algebra for the XY model with open boundary conditions and free coefficients forms a proper subalgebra $\mathfrak{g}_{\text{XY}}\subset \mathfrak{g}_{0}$ with $\text{dim}(\mathfrak{g}_{\text{XY}})=n(n-1)$ \cite{kokcu2021fixed}, given by
    \begin{equation}
    \mathfrak{g}_{\text{XY}}={\rm span}_{\mathbb{R}}\left\langle\left(\{i\sigma^x_j \sigma^x_{j+1}\}_{j=1}^{n-1}\right)\cup \left(\{i\sigma^y_j \sigma^y_{j+1}\}_{j=1}^{n-1}\right)\right\rangle_{\operatorname{Lie}}.
    \label{eqn:g_xy}
    \end{equation}
    \item The Hamiltonian algebra for the transverse field Ising model (TFIM) with open boundary conditions and non-free coupling coefficients forms a proper subalgebra $\mathfrak{g}_{\text{TFIM}}\subset \mathfrak{g}_0$ with $\text{dim}(\mathfrak{g}_{\text{TFIM}})= n^2$ \cite{kazi2022landscape}, given by
    \begin{equation}
    \mathfrak{g}_{\text{TFIM}}={\rm span}_{\mathbb{R}}\left\langle \left(i\sum_{j=1}^{n-1}\sigma^x_j \sigma^x_{j+1}\right)\cup\left( i\sum_{j=1}^n\sigma^z_j\right)\right\rangle_{\operatorname{Lie}}.
    \label{eqn:g_tfim}
    \end{equation}
    \item For a path graph $P_n$ on n vertices, the circuit of the quantum approximate optimization algorithm~\cite{farhi2014quantum} has an associated algebra
    \begin{equation}
        \mathfrak{g}_{\text{QAOA,}P_n}={\rm span}_{\mathbb{R}}\left\langle \left(i\sum_{j=1}^{n-1}\sigma^z_j \sigma^z_{j+1}\right)\cup\left( i\sum_{j=1}^n\sigma^x_j\right)\right\rangle_{\operatorname{Lie}},
        \label{eqn:g_qaoa_path}
    \end{equation}
    which is isomorphic to the one of the TFIM, $\mathfrak{g}_{\text{QAOA,}P_n}\cong\mathfrak{g}_{\text{TFIM}}$. The two are related by a Hadamard transform.
\end{itemize}

\section{Further details on circuit compilation}
\label{appendix:further_details_compilation}
\subsection{Compiling small-$T$ global rotations}
\label{appendix:shallow_compilation}
In Sec.~\ref{sec:compilation_random}, we investigate the trainability of $\mathcal{L}_{\mathfrak{g}}$ for general random $V\in\mathcal{G}$. This includes both global and local dynamics (the respective locality is controlled by the choice of Hamiltonian generators $G_\alpha$ that are included in Eq.~\eqref{eqn:randomunitaries}). Our numerics suggest that compilation of local dynamics is scalable in system size $n$ and rotation angle (dynamics duration) $T$, with $\mathcal{L}_{\mathfrak{g}}$ vanishing exponentially in the number of iterations at all tested values for local dynamics. Similarly, for global dynamics, $\mathcal{L}_{\mathfrak{g}}$ vanishes exponentially in the number of iterations when $T$ is sufficiently large (beyond $T\approx 1$ for a normalized global Hamiltonian). However, for smaller values of $T$, a different convergence behavior is seen when compiling global dynamics. After a brief initial period where $\mathcal{L}_{\mathfrak{g}}$ vanishes exponentially in the number of iterations, a transition to very slow convergence is observed. This is demonstrated in Fig.~\ref{fig:shallow_global_compilation}. The reason for this phenomenon is not currently understood. 
\begin{figure}
    \centering
    \includegraphics[width=0.48\textwidth]{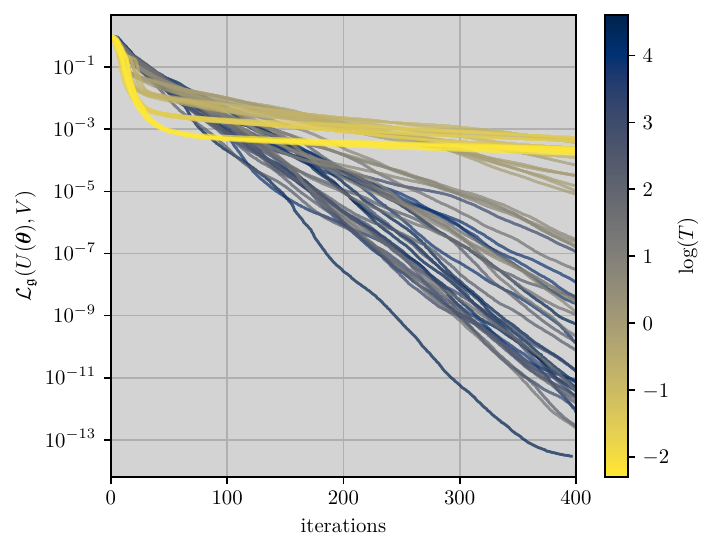}
    \caption{\textbf{Compiling short-time dynamics for global Hamiltonians.} For a system of $n=20$ qubits, we train an overparametrized (27-layer) ansatz to minimize $\mathcal{L}_{\mathfrak{g}}$, with a target unitary generated by a normalized global Hamiltonian. At larger values of $T$, the expected exponential convergence is observed. At $T\approx 1$, a transition is observed to a regime with poorer trainability.}
    \label{fig:shallow_global_compilation}
\end{figure}
\subsection{Center of $\mathcal{G}_{\operatorname{TFXY}}$}
\label{appendix:unfaithfulness}
In Section \ref{sec:compilation_faithfulness}, we observed that when $Z(\mathcal{G})$ is nontrivial (i.e., including non-identity terms), elements of $Z(\mathcal{G})$ can be compiled by the ansatz but cannot be distinguished by the adjoint-representation loss function $\mathcal{L}_{\g}$. Here we show that there is exactly one non-trivial element of $Z(\mathcal{G}_{\operatorname{TFXY}})$ up to a global phase.

We note that for a dynamical Lie algebra $\mathfrak{g}$ and its dynamical Lie group $\mathcal{G}=e^{\mathfrak{g}}$, we have $Z(\mathcal{G})=\mathfrak{g}' \cap \mathcal{G}$, where $\mathfrak{g}'\equiv \{A \in \mbb{C}^{d \times d} \mid \forall B \in \g, AB=BA \}$ is the commutant of $\mathfrak{g}$. Hence, given characterization of the commutant of the Lie algebra of interest, it is relatively easy to determine the center of its Lie group. 
For instance, we have $\mathfrak{g}'_{\operatorname{TFXY}}=\operatorname{span}\{I,(\sigma^z)^{\otimes n}\}$~\cite{kazi2022landscape}, and one can see that the non-identity element of the commutant $(\sigma^z)^{\otimes n}$ can be compiled up to a global phase, since $\sigma_z=ie^{-\frac{i\pi}{2}\sigma_z}=-ie^{-\frac{3i\pi}{4}\sigma_z}$. That is, $(\sigma^z)^{\otimes n}$ can be compiled with $\sigma_z$ rotations on each qubit, with angles $\pi/2$ or $\pi/4$, and thus belong to $\mathcal{G}_{\operatorname{TFXY}}$. This leaves us with exactly one (non-trivial) element of the group that cannot be detected by $\mathcal{L}_{\mathfrak{g}}$.

\section{Implementing the binary quantum-phase classifier on a quantum device}\label{appendix:implementing_classifier}
In Section \ref{sec:supervised_QML}, we described the training of a quantum-phase classifier, implemented via a variational ansatz circuit $U(\bm{\theta})$ defined in Eq.~\eqref{eqn:LTFIM_ansatz-2}. After training, a set of optimal parameters $\bm{\theta}^*$ is acquired, allowing classification of unseen states using the unitary $U_*\equiv U(\bm{\theta}^*)$. Straightforward implementation of the classifier on quantum hardware consists in applying the unitary $U_*$ to an unknown state, measuring the operator $O=\sigma^z_1$, and assigning a label according to Eq.~\eqref{eqn:assign_label}. However we could also replicate the model by means of measurements and post-processing only. Indeed, while Result~\ref{thm:basic_evolution} was framed as an evolution of expectation values, it can equivalently be seen as an (Heisenberg) evolution of the observable.

Given the adjoint representation $\bar{\bm{U}}_{*}=\Phi^{\operatorname{Ad}}_{\lambda}(U_*)$ of the trained circuit (computed in $\gsim$), and a description of the observable $O \in i\g$ as the vector $\bm{w}$ such that $O=\sum (\bm{w})_\alpha G_\alpha$, we have 
\begin{equation}
    O':=  U_*^{\dagger} O U_*=\sum_{\alpha=1}^{\dimg} (\bm{w}')_\alpha G_\alpha, \text{  with  } \bm{w}'= \bm{w}^{\rm T} \bar{\bm{U}}_{*}.
    \label{eqn:inference_observable}
\end{equation}
Hence, one recovers the classifier outputs as a weighted sum of the expectation values of the observables $G_\alpha$. 

Depending on the characteristics of the quantum device at hand (i.e., measurement versus gate fidelity and cycle times), one approach may be preferred to another. For the measurement-based one, one may further reduce the amount of measurements required by pruning the model (i.e., approximating Eq.~\eqref{eqn:inference_observable} with a subset of the terms involved). 
For instance, this could be achieved by discarding terms corresponding to small weights $|(\bm{w}')_\alpha|<\delta$ for some threshold $\delta$, or by discarding greater-than-$k$-local terms. 
In particular, the later could be combined with classical shadows~\cite{huang2020predicting} by means of random measurements (rather than direct measurements of expectation values of all the $G_\alpha$) ensuring small approximation error due to the bounded locality.

\begin{figure}
    \centering
    \includegraphics[width=0.35\textwidth]{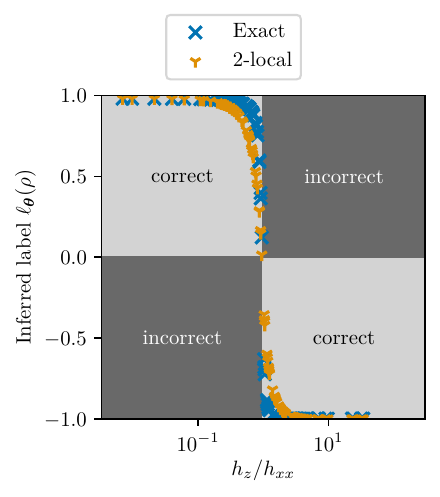}
    \caption{\textbf{Comparison of exact label inferences to their 2-local approximation.} We compare the exact inferred labels of the TFIM phase classifier on an unseen training dataset to a 2-local approximation. The 2-local approximation still achieves 100\% classification accuracy, although $\ell_{\bm{\theta}}(\rho)$ changes less rapidly near the phase transition.}
    \label{fig:shadow_classifier}
\end{figure}
To study further the effect of such locality-pruning, we compare the accuracy of a trained classifier to its 2-local approximation, for an undisguised ($T=0$) TFIM of $n=50$ qubits. 
As can be seen in Figure~\ref{fig:shadow_classifier}, despite slight differences in the output of the models, the approximated classifier retains full accuracy (with 100\% of the testing data correctly classified) that can be explained by the high locality of the TFIM order parameters.

\end{document}